\documentclass[12pt, reqno]{amsart}
\usepackage{mathrsfs}
\usepackage{amsmath}
\usepackage[height=22cm, width=16.5cm, hmarginratio={1:1}]{geometry}
\usepackage{diagbox}
\usepackage{tikz}
\usetikzlibrary{positioning, shapes.geometric}
\usepackage{newclude}
\usepackage{appendix}
\usepackage{extarrows}

\def\sspan{\mathrm{span}}
\usepackage[hidelinks,hyperindex,breaklinks]{hyperref}
\usepackage{longtable}
\usepackage{tikz-cd}
\usetikzlibrary{positioning, arrows,decorations.pathmorphing,shapes}
\usetikzlibrary{shapes.geometric}
\usepackage{array}

\author{Shuai Guo}
\address{School of Mathematics Science, Peking University,
      No 5. Yiheyuan Road,  Beijing 100871, China}
\email{guoshuai@math.pku.edu.cn}

\author{Ce Ji}
\address{Department of Mathematical Sciences, Tsinghua University,
      No 1. Tsinghuayuan, Beijing 100084, China}
\email{cji@tsinghua.edu.cn}

\author{Qingsheng Zhang}
\address{School of Sciences,  Great Bay University,
     Great Bay Institute for Advanced Study,
     Dongguan, 523000, China}
\email{zhangqingsheng@gbu.edu.cn}

\newcommand{\<}{\langle}
\renewcommand{\>}{\rangle}

\newcolumntype{C}[1]{>{\centering\arraybackslash$}m{#1}<{$}}
\newlength{\mycolwd}                                         
\newlength{\mycolwdm}                                         
\newlength{\mycolwda}                                         
\newlength{\mycolwdb}                                         
\newlength{\mycolwdc}                                         
\newlength{\mycolwdd}                                         
\newlength{\mycolwddd}                                         
\newlength{\mycolwddc}                                         
\newcommand{\givz}{\mathfrak{u}}
\newcommand{\givw}{\mathfrak{v}}
\newcommand{\Frob}{H}
\newcommand{\cyA}{\mathfrak A}
\newcommand{\cyB}{\mathfrak B}
\newcommand{\ttau}{t}
\newcommand{\QD}{Q_{\Delta}}

\newcommand{\vac}{\mathbf{v}}
\newcommand{\dvac}{\nu}
\newcommand{\cmn}{{\substack{\mu\\ \nu}}}
\newcommand{\sv}{\varphi}

\newcommand{\elltau}{\tau}

\newcommand{\id}{ \hspace{-2pt}\text{I}}
\usepackage {color,graphicx,psfrag, verbatim,amssymb,amscd,enumerate,subfigure}

\newcommand{\tw}{\mathrm{tw}}

\def\beq{\begin{equation}}
\def\eeq{\end{equation}}
\newcommand{\be}{\begin{equation*}}
\newcommand{\ee}{\end{equation*}}

\DeclareMathOperator{\Cont}{Cont}

\DeclareMathOperator{\ad}{ad}

\DeclareMathOperator{\Res}{Res}

\def\Mgn{\overline{\mathcal{M}}_{g,n}}

\newtheorem{dummy}{}[section]
\newtheorem{lemma}[dummy]{Lemma}
\newtheorem{proposition}[dummy]{Proposition}
\newtheorem{theorem}[dummy]{Theorem}
\newtheorem{corollary}[dummy]{Corollary}
\newtheorem{conjecture}[dummy]{Conjecture}

\theoremstyle{definition}
\newtheorem{convention}[dummy]{Convention}
\newtheorem{definition}[dummy]{Definition}

\newtheorem{remark}[dummy]{Remark}

\usepackage{graphicx}

\newcommand{\M}{\overline{\mathcal M}}

\newcommand{\D}{\mathcal D}

\newcommand{\pd}{\partial}

\newcommand{\cV}{\mathcal{V}}
\newcommand{\cF}{\mathcal{F}}

\newcommand{\cA}{\mathcal{A}}
\newcommand{\cD}{\mathcal{D}}

\newcommand{\I}{\mathcal{I}}

  \newcommand{\cC}{\mathcal{C}}

\newcommand{\Mbar}{\overline{\mathcal M}}

\newcommand{\ch}{\mathrm{ch}}

\def\cG{\mathcal{G}}

\newtheorem{manualtheoreminner}{Theorem}
\newenvironment{manualtheorem}[1]{%
	\renewcommand\themanualtheoreminner{#1}%
	\manualtheoreminner
}{\endmanualtheoreminner}

\begin{document}
	\title{A generalization of the Witten conjecture  through  spectral curve}
	
	\maketitle

\begin{abstract}
	
	We propose a generalization of the Witten conjecture, which connects a descendent enumerative theory with a specific  reduction of KP integrable hierarchy. Our conjecture is realized by two parts:
    Part I (Geometry) establishes a correspondence between the geometric descendent potential (apart from ancestors) and   the topological recursion of specific spectral curve data $(\Sigma, x,y)$;
	Part II (Integrability) claims that the TR descendent potential, defined at the boundary points of the spectral curve (where $dx$ has poles), is a tau-function of a certain reduction of the multi-component KP hierarchy.

 In this paper, we show the geometric part of the conjecture for any formal descendent theory by using a generalized Laplace transform. Subsequently, we prove the integrability conjecture   for the one-boundary cases. As applications, we generalize and prove the $r$KdV integrability of negative $r$-spin theory conjectured by Chidambaram, Garcia-Failde and Giacchetto.
We also show the KdV integrability of the total descendent potential associated with  the Hurwitz space $M_{1,1}$, whose Frobenius manifold was initially introduced by Dubrovin.

\end{abstract}
	
\setcounter{section}{-1}
\setcounter{tocdepth}{1}
\tableofcontents

\section{Introduction}

The celebrated Witten conjecture \cite{Wit91} initiated the study of the connection between enumerative geometry and integrable hierarchies. It claimed that the generating function of the intersection numbers of stable classes on the moduli space of curves is a tau-function of the KdV hierarchy.
In \cite{Wit93}, Witten extended his conjecture to the case of ADE singularities, relating quantum singularity theory and the ADE-integrable hierarchies.
The Witten conjecture has analogues in various enumerative theories as well. For example, the Toda conjecture \cite{EHY95, EY94, Get04} states that the equivariant Gromov--Witten theory of $\mathbb P^1$ is governed by the equivariant Toda hierarchy.

The original Witten conjecture was proved by Kontsevich \cite{Kon92} using a combinatorial description of the moduli spaces in terms of ribbon graphs. For the  $A_{r-1}$ singularity, the generalized Witten conjecture relates the theory of $r$-spin curves and the $r$KdV hierarchy, and was proved by Faber--Shadrin--Zvonkine \cite{FSZ06}. For general quantum singularities, Fan--Jarvis--Ruan developed a mathematical formulation using Witten's equation, solution of which can be counted via the corresponding enumerative geometry called the FJRW theory. Then they established and proved the generalized Witten conjecture for the ADE-type singularities \cite{FJR13}. The Toda conjecture was proved by Okounkov--Pandharipande \cite{OP06}. Indeed, they proved that the equivariant Gromov--Witten potential of $\mathbb P^1$ satisfies the $2$-Toda equation.

By reviewing the history of the developments related to the Witten conjecture, one can see that all the integrable hierarchies mentioned above are reductions of (multi-component) KP hierarchy.
On the other hand, they all have an underlying global spectral curve in the sense of  \cite{EO07}. Indeed, such observation could date back to the seminal paper   \cite{ADKMV06}, where the authors claim that  the boundary expansion of the topological string partition function on a Riemann surface gives rise to  a tau-function of the multi-component KP hierarchy.

This suggests the following  diagram as a generalization of the Witten conjecture:
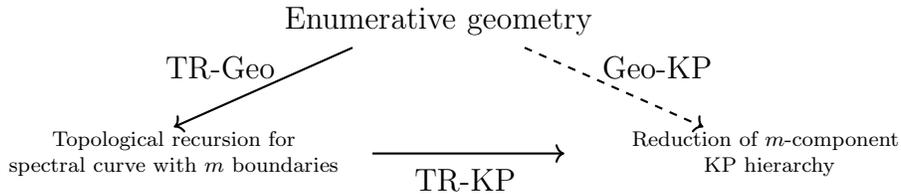
\begin{figure}[htbp]\label{Pic}
	\centering
	\begin{tikzpicture}[scale=0.88]
		\node [] at (0,2) {Enumerative geometry};
		\node [] at (5,0) {$\text{Reduction of $m$-component } \atop \text{KP hierarchy}$};
		\node [] at (-4,0) {$\text{Topological recursion for} \atop \text{spectral curve with $m$ boundaries}$};
		\draw [->, thick] (-1,0) -- (1.9,0); \node[] at (0.4, -0.4) {TR-KP};
		\draw[<-, thick] (-4, 0.4) -- (-1.3, 1.6); \node[] at (-3.3, 1.3){TR-Geo};
		\draw[->, thick, dashed] (1.3, 1.6) -- (4, 0.4); \node[] at (3.3, 1.3){Geo-KP};
	\end{tikzpicture}
\end{figure}

\noindent Two (conjectural) structures arise in the above picture:
\begin{enumerate}[(i)]
	\item {\bf Integrability conjecture}: Topological recursion (TR) of a given spectral curve defines a set of multi-differentials $\omega_{g,n}$ on a Riemann surface with marked points(boundaries).
	When  $\omega_{g,n}$ are expanded using local coordinates near the boundaries,  it results in a tau-function of  a  certain  reduction of  the multi-component KP hierarchy.
	For the spectral curve with genus $\frak g\geq 1$, the non-perturbative contributions are considered.
	\item {\bf Enumerative geometry}:  A hidden geometric theory are linked to the   spectral curve, such that  the local expansion of the $\omega_{g,n}$   at both critical points and  boundary points carries significant geometric implications: the former defines a semi-simple cohomological field theory (CohFT) that characterizes the geometric theory, while the latter recovers (part of) the geometric descendent potential via an oscillatory integral.
\end{enumerate}
The (i) and (ii) jointly generalize the Witten conjecture from the perspective of topological recursion, which assigns a descendent theory to an integrable system. To be more precise, the conjecture claims that a specific change of variable transforms the (non-perturbative) descendent potential of the enumerative geometry into a tau-function of a specific reduction of the (multi-component) KP hierarchy.

 \subsection{Topological recursion, KP integrability and enumerative geometry} \label{sec:intro1}
A set of spectral curve data $\cC$ is defined by a triple $(\Sigma, x, y)$, where
 \begin{itemize}
\item $\Sigma$ is a Riemann surface of genus $\mathfrak{g}$ with marked points (boundaries) $\mathbf b=\{b_1,\cdots,b_m\}$ and  a fixed symplectic basis $\{\cyA_1, \ldots, \cyA_{\mathfrak{g}}, \cyB_1, \ldots, \cyB_{\mathfrak{g}}\}$ for $H_1(\Sigma, \mathbb{C})$;
\item $x$ is a (possibly multi-valued) function  defined on $\Sigma-   \mathbf b$ whose differential $dx$ is a meromorphic differential on $\Sigma$  having poles  exactly  at the boundaries $b_i$ and having only finitely many simple zeros $z^\beta$, $\beta=1,\cdots,N$, called the critical points;
\item $y$   is a holomorphic function   defined near each zero $z^\beta$ of $dx$  satisfying  $dy(z^\beta)\neq 0$.
\end{itemize}

Taking the spectral curve data  as input, the topological recursion  \cite{EO07} gives an algorithm that  produces a sequence of multi-differential forms
   ${\omega_{g,n}}$ (see \S \ref{sec:def-TR} for the definition) for non-negative integers $g,n$.
When expanded under a suitable basis,  the coefficients of   $\omega_{g,n}$  can be  viewed as the ``invariants" of the spectral curve data {with respect to this basis}.
Inspired by the enumerative geometry where the ancestor invariants and descendent invariants are defined, we consider the following two types of invariants.
\begin{enumerate}[(i)]
\item {\bf TR ancestors:}
Using the local Airy coordinates near the critical points, we can  define a set of meromorphic $1$-forms $\{d\xi_n^{\bar\beta}(z)\}$ as a basis.
For $2g-2+n>0$, the expansion of $\omega_{g,n}$ under this basis defines the {\it TR ancestors} $\<\bar e_{\beta_1}\bar\psi^{k_1},\cdots,\bar e_{\beta_n}\bar\psi^{k_n}\>_{g,n}^{\rm TR}$.	
(See \S \ref{sec:def-TR}.)
\item{\bf TR descendents:}
Pick $\Lambda=(\lambda_1,\cdots,\lambda_m)$, such that $\lambda_i^{-1}$ is a local coordinate near $b_i$ satisfying $\lambda_i(b_i)={\infty}$, $i=1,\cdots,m$.
By using basis $\{d\lambda^{-k}_{i}\}$, the expansion of $\omega_{g,n}$ near the boundary point defines the \emph{TR descendents}
$\<\alpha^{i_1}_{k_1},\cdots,\alpha^{i_n}_{k_n}\>^{\Lambda}_{g,n}.$
(See \S \ref{sec:TR-KP}.)
\end{enumerate}

The main object in this paper is the generating series of the TR descendents $\<-\>^{\Lambda}_{g,n}$:
$$
Z({\bf p};\hbar)
=\exp\bigg(\sum_{g\geq 0, n\geq 0}\hbar^{2g-2+n}\sum_{1\leq i_1,\cdots,i_n\leq m \atop k_1,\cdots,k_n\geq 1}
\<\alpha^{i_1}_{k_1},\cdots,\alpha^{i_n}_{k_n}\>_{g,n}^{\Lambda}
\frac{p^{i_1}_{k_1}\cdots p^{i_n}_{k_n}}{n!\cdot k_1\cdots k_n}\bigg).
$$

\begin{remark}
   While \(\omega_{g,n}\) admits expansion around any point on the spectral curve, we emphasize that boundary expansions are of particular importance: it is primarily through these expansions that the critical link between TR descendants and geometric descendants becomes apparent.  See §\ref{sec:TR-des-Geo} for details.
\end{remark}

For the Riemann surface $\Sigma$ of genus $\frak g\geq 1$,
integrals of $\omega_{g,n}$ along cycles on $\Sigma$ do not necessarily vanish.
Hence, by considering the expansion of these integrals, we obtain more quantities.
Inspired by the works of Eynard--Mari\~{n}o \cite{EM11} and Borot--Eynard \cite{BE12}, where the non-perturbative partition function and spinor kernel were introduced, in this paper,  we define the non-perturbative generating series
 $Z^{\rm NP}_{\mu,\nu}({\bf p};\hbar;w)$, where  $\mu,\nu\in\mathbb R^{\frak g}$ and $w\in\mathbb C^{\frak g}$ are  parameters. See \S \ref{sec:NP-TR} for details.

 \subsubsection{TR-KP}
First we introduce the (m-component) KP integrability conjecture for TR.
 \begin{conjecture}[$m$-KP integrability conjecture of the topological recursion]\label{conj:EO-m-KP} \text{$\quad$} \\
Let $\cC=(\Sigma,x,y)$ be the spectral curve with boundaries ${\bf b}=\{b_1,\cdots,b_m\}$, and let $\Lambda=(\lambda_1,\cdots,\lambda_m)$ be a local coordinate system near the boundaries.
Let $\frak g=\frak g(\Sigma)$. We have:
\begin{itemize}
\item If $\frak g=0$, then the generating series $Z({\bf p};\hbar)$ is a tau-function of the $m$-component KP hierarchy with KP times \{$\frac{p_k^i}{k}\}_{k\geq 1, 1\leq i\leq m}$.
\item If $\frak g\geq 1$, then the non-perturbative generating series $Z^{\rm NP}_{\mu,\nu}({\bf p};\hbar;w)$ is a tau-function of the $m$-component KP hierarchy with KP times \{$\frac{p_k^i}{k}\}_{k\geq 1, 1\leq i\leq m}$.
\end{itemize}
 \end{conjecture}

When the function $x$ is meromorphic, we also conjecture the generating series $Z({\bf p};\hbar)$ and $Z^{\rm NP}_{\mu,\nu}({\bf p};\hbar;w)$ satisfy  certain reduction structure for  certain choice of $\Lambda$.

In this paper we study the $m=1$ case of Conjecture~\ref{conj:EO-m-KP} (with a focus on the subcase where $x$ is meromorphic). For simplicity, we denote the boundary point  in this case  by $b$ and denote $\Lambda$ by $\lambda$ for simplicity. The exploration of the multi-component KP cases will be addressed in future works.

 \begin{manualtheorem}{I} [TR-KP] \label{KP-TR}
Assume that the spectral curve $\cC=(\Sigma,x,y)$ has a single boundary point,  and that the function $x$ is meromorphic with $x\in  \mathbb C[\lambda]$ near the boundary.
\begin{itemize}
	\item For the case $\frak g=0$, the generating series $Z({\bf p};\hbar)$ is a tau-function of the polynomial-reduced KP hierarchy ({see Definition~\ref{def:lambda-reduction-KP}}).
	\item For the case $\frak g\geq 1$, the non-perturbative generating series $Z^{\rm NP}_{\mu,\nu}({\bf p};\hbar;w)$ is a tau-function of the polynomial-reduced KP hierarchy.
\end{itemize}
 \end{manualtheorem}

 \begin{remark}
It is a classical result that the tau-function $\mathcal T({\bf p})$ of the KP hierarchy  can be constructed from algebraic curve, as studied by Dubrovin~\cite{Dub81}, Krichever~\cite{Kri77}, Mumford~\cite{Mum}, et~al.
We will see in \S \ref{sec:NP-TR} that the tau-function $\mathcal T({\bf p})$ is equal to
$Z^{\rm NP}_{\mu,\nu}({\bf p};\hbar;w)$ evaluated at $\hbar=0$.
Thus, our result generalize this classical result to a ``quantum" version.
\end{remark}

\subsubsection{TR-Geo}
We then explore the link between the enumerative geometry and the spectral curve.
It is known that  the TR ancestors are related to the geometric ancestors of the CohFT  associated with the spectral curve \cite{DOSS14,Eyn14}.
In general, the geometric descendent theory of a CohFT is unknown.
Following~\cite{GZ25}, we formally define the geometric descendent invariants of $S$- and $\dvac$-calibrated CohFT by generalizing the Kontsevich--Manin formula (see \S \ref{sec:descedentforCohFT}).
The geometric definition of the descendent invariants always corresponds to the formal definition with a specific $S$-matrix and $\dvac$-vector.

Now we explain how the TR descendents are related with the geometric descendents.
\begin{definition}\label{conj:TR-Geo-des}
	A path $\gamma \subset \Sigma - \bf b$ is called admissible associated with $e^{-x/\givz}$, if it satisfies the following conditions:
	\begin{itemize}
		\item  Each end of the path $\gamma$  tends to  a boundary point in $\bf b$  .
		\item   $x(z)$  tends to infinity \footnote{As $dx$ has poles at the boundaries,  $|x|$ tends to infinity when  $z$ approaches the boundary. } along $\givz\cdot\mathbb R_{\gg0}$ as $z$ approaches the two boundaries along $\gamma$.
	\end{itemize}
\end{definition}

For an admissible path $\gamma$ associated with $e^{-x/\givz}$, we define a corresponding class $\Phi(\gamma,-\givz)$ for $S$-calibrated CohFT (see~Lemma~\ref{lem:path-class}).
Then we can relate the geometric descendents $\<- \>^{\cD}_{g,n}$ with multi-differentials $\omega_{g,n}$ of the topological recursion in the following way.
\begin{manualtheorem}{II}[TR-Geo, Theorem~\ref{thm:TR-des-(g,n)}]\label{thm:TRGE}
For $2g-2+n\geq 0$  and admissible paths $\gamma_i$ associated with $e^{-x(z_i)/\givz_i}$, $i=1,\cdots,n$,
$$
  \int_{\gamma_1}\!\!\cdots\!\!\int_{\gamma_n}e^{-\sum_i x(z_i)/\givz_i}\omega_{g,n} = \delta_{(g,n)}^{(0,2)}\cdot\frac{\eta(\Phi(\gamma_1,-\givz_1),\Phi(\gamma_2,-\givz_2))}{-\givz_1-\givz_2}
   +\Big\<\frac{\Phi(\gamma_1,-\givz_1)}{\givz_1+\psi_1},\cdots,\frac{\Phi(\gamma_n,-\givz_n)}{\givz_n+\psi_n}\Big\>^{\cD}_{g,n}.
$$
\end{manualtheorem}

For the case of $(g,n)=(0,1)$, the geometric descendents $\<-\>_{0,1}^{\cD}$ are usually encoded into a single function, called the  $J$-function.
\footnote{ For an $S$- and $\dvac$-calibrated CohFT, we define the $J$-function by the $S$-matrix and $\dvac$-vector (see~\S \ref{sec:descedentforCohFT}).}
We propose the following conjecture:
 \begin{conjecture}\label{conj:laplace-J}
There exists a $\dvac$-vector $\dvac(\givz)$ (see~\S \ref{sec:descedentforCohFT}), such that the $J$-function
	$$
	J(-\givz):=-\givz S^{*}(-\givz)\dvac(\givz)
	$$
satisfies the following relation
$$
\int_{\gamma}   e^{-x(z)/\givz} y(z)dx(z) = -\eta( J(-\givz) , \Phi(\gamma,-\givz)),
$$
where $\gamma$ is an admissible path associated with $e^{-x(z)/\givz}$.
 \end{conjecture}
 \begin{remark}
 This conjecture relates the spectral curve data and the one-point function of the geometric theory, which a priori is independent by definition. We successfully examine  the conjecture in the two examples considered  in this paper, and serve them as evidence.
 \end{remark}

\subsubsection{Geo-KP}
With descendent invariants determined by topological recursion, we complete the diagram by showing  how to directly realize the total descendent potential as a tau-function of KP hierarchy. This is achieved by regarding  the inverse Laplace transform as a change of basis.

Let $\cD({\bf t};\hbar)$ be the total descendent potential of CohFT associated with the spectral curve $\cC=(\Sigma,x,y)$ and calibrated by $S$-matrix and $\dvac$-vector.
For the case with $\frak g(\Sigma)\geq 1$, we introduce the non-perturbative total descendent potential and denote it by $\cD^{\rm NP}_{\mu,\nu}({\bf t};\hbar;w)$.
Combining the $m$-KP integrability conjecture of the topological recursion and the TR-Geo correspondence, we conclude the following generalization of the Witten conjecture:
\begin{conjecture}[Generalization of the Witten conjecture]\label{conj:generalized-witten}
Let $\cC=(\Sigma,x,y)$ be the spectral curve with $m$ boundaries.
We denote its geometric total descendent potential by $\cD({\bf t};\hbar)$  and, for $\frak g(\Sigma)\geq 1$,  its non-perturbative total descendent potential by $\cD^{\rm NP}_{\mu,\nu}( {\bf t};\hbar;w)$. There exist a coordinate transformation ${\bf t}={\bf t(p)}$ and a quadratic function $\QD({\bf p},{\bf p})$
such that
\begin{itemize}
\item For $\frak g=0$, the generating series $e^{\frac{1}{2} \QD({\bf p,p})}\, \cD(\hbar\cdot {\bf t(p)};\hbar)$ is a tau-function of the $m$-component KP hierarchy with KP times \{$\frac{p_k^i}{k}\}_{k\geq 1, 1\leq i\leq m}$.
\item For $\frak g\geq 1$, the non-perturbative generating series $e^{\frac{1}{2}\QD({\bf p,p})}\cD^{\rm NP}_{\mu,\nu}(\hbar\cdot {\bf t(p)};\hbar;w) $  is a tau-function of the $m$-component KP hierarchy with KP times \{$\frac{p_k^i}{k}\}_{k\geq 1, 1\leq i\leq m}$.
\end{itemize}
Furthermore, the generating series for both cases satisfy  certain reduction structure.
\end{conjecture}
In this paper, we study the one boundary cases, we have the following result:
 \begin{manualtheorem}{III} [Geo-KP] \label{Geo-KP}
There are a  coordinate transformation ${\bf t}={\bf t(p)}$ and a quadratic function $\QD({\bf p},{\bf p})$
such that the function $e^{\frac{1}{2}\QD({\bf p},{\bf p})}\cdot\cD(\hbar\cdot {\bf t(p)};\hbar)$ for the ${\frak g}=0$ case, or its non-perturbative analogue $e^{\frac{1}{2}\QD({\bf p},{\bf p})}\cdot\cD^{\rm NP}_{\mu,\nu}(\hbar\cdot {\bf t(p)};\hbar;w)$  for the $\frak g\geq 1$ case, is a tau-function of the polynomial-reduced KP hierarchy with KP times $\{\frac{p_k}{k}\}_{k\geq 1}$.
 \end{manualtheorem}

\subsection{Applications}
We will illustrate the whole story with two examples : the $r$KdV integrability of the descendent deformed negative $r$-spin theory and the KdV integrability of the descendent potential of the Hurwitz space $M_{1,1}$ introduced by Dubrovin.
The detailed proofs can be found in \S \ref{sec:geo-negative-rspin} and \S \ref{sec:TR-weierstrass}, respectively.
\subsubsection{The negative $r$-spin Witten conjecture}  \label{sec:negative-r-spin}
We prove a generalization of “the negative r-spin Witten conjecture” proposed in~\cite{CGG22}.
We first recall the geometry of the negative $r$-spin theory.
Let $r$ be a fixed positive integer and $0\leq a_1,\cdots, a_n \leq r-1$, $s$ be integers satisfying
$$
\tilde  D_{g,n}^{r,s}( { \vec{a}}) := (2g-2+n)\cdot s- \textstyle \sum_{i=1}^n a_i \in   r \mathbb Z.
$$
 We introduce   the proper moduli space of twisted stable $r$-spin curves  	
$$
\M^{r,s}_{g,\vec{a}} = \{(C,p_1,\cdots,p_n,  L) :   \   L^{r} \cong \omega_{\log}^s(-\textstyle \sum_{i=1}^n  a_i [p_i]) \}.
$$
Here $\omega_{\log} = \omega( \sum_{i=1}^n  [p_i])$.  We see $\deg \omega_{\log}^s(-\textstyle \sum_{i=1}^n  a_i [p_i])=\tilde  D_{g,n}^{r,s}( { \vec{a}})  $
 and  the Riemann-Roch theorem gives
$$
 D_{g,n}^{r,s}( { \vec{a}}):=\dim H^1(C,L) -\dim H^0(C,L)   =  -\deg L+g-1 =  -\tfrac{1}{r} \tilde  D_{g,n}^{r,s}( { \vec{a}}) +g-1.
$$

Let $\mathcal C_{g,\vec{a}}^{r,s}$ be the universal curve of the moduli space $\M^{r,s}_{g,\vec{a}}$ and let $\mathcal L_{g,\vec{a}}^{r,s}$ be the universal bundle on it, we have the morphisms
$$
\mathcal C_{g,\vec{a}}^{r,s} \xlongrightarrow{\pi}  \M_{g,\vec{a}}^{r,s} \xlongrightarrow{f}  \M_{g,n}.
$$
For $-r+1\leq s\leq -1$, one sees $\mathrm R^0 \pi_* \mathcal L_{g,\vec{a}}^{r,s}$ vanishes.   Following  \cite{CGG22,Chi08},  we consider the vector bundle of rank $ D_{g,n}^{r,s}( { \vec{a}})$ :
$$
\mathcal V^{r,s}_{g,\vec{a}} :=  \mathrm R^1 \pi_* \mathcal L_{g,\vec{a}}^{r,s},
$$
and introduce the following twisted class
$$
c_{\rm top}( \mathcal V^{r,s}_{g,\vec{a}} )  \in H^{  D_{g,n}^{r,s}( { \vec{a}})}(\M_{g,\vec{a}}^{r,s}) ,
$$
{where $c_{\rm top}( \mathcal V^{r,s}_{g,\vec{a}} ) $ means the top Chern class of $\mathcal V^{r,s}_{g,\vec{a}} $.}

In this paper, we will focus on the $s=-1$ case. We will consider the following descendent invariants of the negative $r$-spin theory and its deformation.

\begin{definition}\label{def:des-neg-r}
	Let  $\bar\Frob=\sspan_{\mathbb Q}(\phi_{0},\cdots,\phi_{r-1})$ be the state space, where $\{\phi_{a}\}_{a=0}^{r-1}$ are vectors associated to the integers $0\leq a\leq r-1$. Let $\{\psi_i\}_{i=1}^n$ be the psi-classes, denoting the first Chern class of the universal cotangent line bundle over $\M_{g,\vec{a}}^{r,-1}$ with respect to the $i$-th marked point.
The descendent correlators for the negative $r$-spin theory are defined by
$$
\<  \phi_{a_1}\psi_1^{k_1},\cdots, \phi_{a_n}\psi_n^{k_n}\>_{g,n}^{r}
:=\frac{(-1)^{D^{r,-1}_{g,n}(\vec{a})}}{r^{g-1}}
  \int_{\M_{g,\vec{a}}^{r,-1}}  c_{\rm top}( \mathcal V^{r,-1}_{g,\vec{a}} )\prod_i  \psi_i^{k_i},
$$
and the descendent correlators for the $\epsilon$-deformed negative $r$-spin theory are defined by
$$
\<  \phi_{a_1}\psi_1^{k_1},\cdots, \phi_{a_n}\psi_n^{k_n}\>_{g,n}^{r,\epsilon}
:=\sum_{m\geq 0}\frac{\epsilon^m}{m!}\<  \phi_{a_1}\psi_1^{k_1},\cdots, \phi_{a_n}\psi_n^{k_n},\phi_0,\cdots,\phi_0\>_{g,n+m}^{r}.
$$
\end{definition}

\begin{manualtheorem}{1}
\label{conj:negative-witten-conjecture}
Let   $\cD^{r,\epsilon}$ be the total descendent potential defined by
\beq\label{def:partitionfunction-r}
\cD^{r,\epsilon }(\mathbf t;\hbar) = \exp\bigg(\sum_{g\geq 0}\hbar^{2g-2} \sum_{n\geq 0} \frac{1}{n!} \<\mathbf t(\psi_1),\cdots,\mathbf t(\psi_n)\>_{g,n}^{r,\epsilon}\bigg),
\eeq
 where
$ \mathbf t(\psi) = \textstyle \sum_{k\geq 0,1\leq a \leq r-1}  t_k^a  \phi_{a} \psi^k \in \Frob[[\psi]]$, $\Frob=\sspan_{\mathbb Q}(\phi_{1},\cdots,\phi_{r-1})$.
Then it is a tau-function of $r$KdV hierarchies for any $\epsilon$,  with KP times $\{\frac{p_k}{k}\}_{k\geq1}$ under the change of variables
\beq\label{eqn:wittentime}
t^a_k
=-\frac{(-1)^k}{\sqrt{-r}}\,\frac{\Gamma(k+\frac{a}{r})}{\Gamma(\frac{a}{r})}\, p_{rk+a},\qquad k\geq 0,\quad a=1,\cdots,r-1.
\eeq
\end{manualtheorem}
The special case $\epsilon=0$ is the ``Negative $r$-spin Witten   conjecture" formulated in~\cite{CGG22},
our formulation generalizes and proves their conjecture for arbitrary  $\epsilon$.

\subsubsection{KdV integrability of the descendent theory of the Hurwitz space $M_{1,1}$}\label{sec:intro-weierstrass}
The Hurwitz space (with fixed type of ramification) $M_{\mathfrak g, r_1-1,\cdots,r_m-1}$ is the moduli space parameterizing the collections $(\Sigma,b_1,\cdots,b_m;x)$, where $\Sigma$ is a genus $\mathfrak{g}$ Riemann surface with marked points $b_1,\cdots,b_m$ and a fixed symplectic basis for $H_1(\Sigma,\mathbb Z)$, and $x$ is a marked meromorphic function $x: \Sigma\to \mathbb P^1$ such that $x^{-1}(\infty)=r_1b_1+r_2b_2+\cdots+r_mb_m$.
In~\cite[Lecture 5]{Dub96}, Dubrovin constructed a homogeneous semi-simple Frobenius manifold structure on a covering of Hurwitz spaces by viewing $x$ as the superpotential in the sense of~\cite[Appendix I]{Dub96}.
The CohFT associated with the Hurwitz is then uniquely reconstructed from this Frobenius manifold by Givental--Teleman reconstruction theorem~\cite{Giv01a,Tel12}.
In~\cite{DNOPS19}, it is proved that this CohFT is identified with the one defined by topological recursion on spectral curve data $(\Sigma,x, y(z)=z)$.
It follows immediately from our result that (part of) the non-perturbative descendent generating series of the CohFT associated with the Hurwitz space $M_{\mathfrak g,r-1}$ gives a tau-function of polynomial-reduced KP hierarchy.

In this paper, we show the explicit results for the case $\mathfrak{g}=1$, $r=2$, i.e., the Hurwitz space $M_{1,1}$.
It has been shown in~\cite{DNOPS19} that the CohFT associated with the Hurwitz space $M_{1,1}$
coincides with the CohFT associated with the Weierstrass curve (see \S \ref{sec:TR-weierstrass}).
We further establish the relationship between the two corresponding descendent generating series.
By applying the TR-KP correspondence, the non-perturbative TR descendent generating series is a tau-function of KdV hierarchy.
We obtain the following theorem.
\begin{manualtheorem}{2} \label{thm:elliptic-KdV}
Restricted to the $\phi_1$ direction, the descendent invariants for the  Hurwitz space $M_{1,1}$ match with the TR descendent invariants of the Weierstrass curve:
$$
\< \phi_1 \psi_1^{k_1}, \cdots,\phi_1 \psi_n^{k_n} \>_{g,n}^{\D} = \frac{1}{\prod_{i=1}^n (2k_i+1)!!}\cdot \<\alpha_{2k_1+1},\cdots,\alpha_{2k_n+1}\>_{g,n}^{\lambda},
$$
where $\lambda=\sqrt{2\, x}$.
Furthermore, the non-perturbative modification $\cD^{\rm NP}_{\mu,\nu}({\bf t};\hbar;w)$ of the total descendent potential
$$
\cD({\bf t};\hbar) = \exp\bigg(\sum_{g\geq 0}\hbar^{2g-2}\sum_{n\geq 0}\< \phi_{i_1} \psi_1^{k_1}, \cdots,\phi_{i_n} \psi_{n}^{k_n} \>_{g,n}^{\D}  \frac{t_{k_1}^{i_1}\cdots t_{k_n}^{i_n}}{n!}\bigg)
$$
of the Hurwitz space $M_{1,1}$ is a tau-function of KdV hierarchy with KdV times $p_{2k+1}/(2k+1)$ under the change of variables
$t^i_k=\delta_{i,1}(2k-1)!!p_{2k+1}$, $k\geq 0$.
\end{manualtheorem}

\subsection{Outline of the paper}
This paper is organized as follows.
In \S \ref{sec:CohFT} and \S \ref{sec:TR}, we review the axiomatic theory for the enumerative geometry and the topological recursion, respectively.
In \S \ref{sec:TR-Geo}, we establish the relationship between the TR descendent invariants and the geometric descendent invariants.
In \S \ref{sec:NP-TR}, we introduce the non-perturbative contributions in the topological recursion, and propose our generalization of the Witten conjecture via the KP integrability conjecture for the (non-perturbative) topological recursion.
In \S \ref{sec:TRKPproof}, we prove the KP integrability conjecture.
In \S \ref{sec:geo-negative-rspin}, we generalize and prove the negative $r$-spin Witten conjecture.
In \S \ref{sec:TR-weierstrass}, we establish the descendent theory of the Weierstrass curve and prove its non-perturbative total descendent potential is a tau-function of  the KdV hierarchy.

\textbf{Acknowledgements.}
The first version of the paper was presented at the conference “Integrable Systems and Their Applications 2023” in Sochi, and the revised version was presented at “Integrable Systems and Their Applications 2024” in Sochi and at the “Sino-Russian Interdisciplinary Mathematics Conference 2024” in Moscow. The authors would like to thank the organizers for providing this opportunity.
The authors would also like to thank Petr Grinevich, Maxim Kazarian, Alexander Zheglov, Jian Zhou, Zhengyu Zong and Chenglang Yang for their helpful discussions.
This work is partially supported by the National Key Research and Development Program of China (Grant Nos. 2023YFA1009802
and 2020YFE0204200) and the National Natural Science Foundation of China (Grant Nos.
12225101, 12061131014 and 11890660).

\section{Axiomatic theory for enumerative geometry}
\label{sec:CohFT}
In this section, we review the axiomatic theory for enumerative geometry, including the notion of CohFT, the ancestor potential as its generating series, and the descendent potential of a CohFT calibrated by $S$-matrix and $\dvac$-vector. We show some necessary results without giving  proof and we refer readers to~\cite[\S 1, 2]{GZ25} for more details.

\subsection{Cohomological field theory}\label{subsec:CohFT}
The concept of CohFT was introduced by Kontsevich and Manin \cite{KM94} to capture the axiomatic properties of virtual fundamental class in Gromov--Witten theory. It turns out that the CohFT properties hold generally in various kinds of enumerative geometries.

Fix a field $\mathbb F$. The state space  $\Frob$ is a finite dimensional $\mathbb F$-vector space with a non-degenerate symmetric 2-form $\eta$.
Let $N=\dim \Frob$ and $\{\phi_i\}_{i=1}^{N}$ be a basis, called the flat basis, of $\Frob$, we denote by $\{\phi^i\}_{i=1}^{N}$ the dual basis of $\{\phi_i\}_{i=1}^{N}$ with respect to $\eta$.
\begin{definition}\label{def:cohft}
Let $\mathbb A$ be an $\mathbb F$-algebra.
A CohFT $\Omega=\{\Omega_{g,n}\}_{2g-2+n>0}$ on $(\Frob, \eta)$ is defined to be a set of maps to the cohomological classes on the moduli space $\Mbar_{g,n}$ of stable curves
$$
\Omega_{g, n} \colon \Frob ^{\otimes n}  \rightarrow   \mathbb A \otimes H^*(\Mbar_{g,n},\mathbb{Q})
$$
satisfying the following two axioms:

{\bf 1) $S_n$-invariance Axiom:} $\Omega_{g,n}(v_1,\cdots,v_n)$ is invariant under permutation between $v_i\in \Frob$, $i=1,\cdots,n$, corresponding to the marked points of $\Mbar_{g,n}$.

{\bf 2) Gluing Axiom:}
The pull-backs $q^*\Omega_{g,n}$ and $r^*\Omega_{g,n}$ of the gluing maps
\begin{equation*}
q\colon \Mbar_{g-1, n+2}\to \Mbar_{g,n} \quad {\text{ and }}\quad
r\colon \Mbar_{g_1, n_1+1}\times \Mbar_{g_2, n_2+1} \to \Mbar_{g_1+g_2, n_1+n_2}
\end{equation*}
 are equal to the contraction
of $\Omega_{g-1, n+2}$ and $\Omega_{g_1, n_1+1} \otimes \Omega_{g_2, n_2+1}$ by the bivector $\sum_i \phi_i\otimes \phi^i$.
\end{definition}

A CohFT induces a {\it quantum product} $*$ on $\Frob$ by
$$
\eta(v_1* v_2, v_3)=\Omega_{0,3}(v_1, v_2, v_3),
$$
where axioms 1) and 2) ensure its commutativity and associativity respectively.

We introduce the ancestor correlator $\<-\>_{g,n}^{\Omega}$, $2g-2+n>0$, for the CohFT $\Omega$:
\beq\label{def:bracket-anc}\textstyle
\<v_1\bar\psi^{k_1},\cdots,v_n\bar\psi^{k_n}\>^{\Omega}_{g,n}
:=\int_{\Mbar_{g,n}}\Omega_{g,n}(v_1, \cdots , v_n)\psi_1^{k_1}\cdots\psi_n^{k_n},
\eeq
and $\<-\>_{g,n}^{\Omega}:=0$ for $2g-2+n\leq 0$.
We usually drop the superscript $\Omega$ in $\<-\>_{g,n}^{\Omega}$ if no confusion arises.
The genus-$g$ ancestor potential $\bar\cF_{g}({\bf s})$ of $\Omega$ is defined by
$$\textstyle
\bar\cF_{g}({\bf s}):=\sum_{n\geq 0}\frac{1}{n!}\<{\bf s}(\bar\psi_1),\cdots,{\bf s}(\bar\psi_n)\>_{g,n},
$$
where ${\bf s}(\givz)=\sum_{k\geq 0}s_k\givz^k\in \Frob[[\givz]]$.
The total ancestor potential $\cA({\bf s};\hbar)$ is defined by
\beq\label{def:FandA-CohFT}
\cA({\bf s};\hbar):=e^{\sum_{g\geq 0}\hbar^{2g-2}\bar\cF_{g}(\bf s)}.
\eeq

Let $\mathcal U \subset  \Frob $ be a (formal) neighborhood of $0\in \Frob$,
for $\ttau\in \mathcal U$, we define $\Omega^{\ttau}$ by
$$\textstyle
\Omega^{\ttau}_{g,n}(v_1,\cdots, v_n):=\sum_{m\geq 0}\frac{1}{m!}p^{m}_{*}\Omega_{g,n+m}(v_1,\cdots, v_n,\ttau,\cdots,\ttau),
$$
it is convergent under suitable topology on $\mathbb A$
~\footnote{
 For example,  let $\ttau=\sum_i \ttau^i\phi_i$ and  $\mathbb A=\mathbb Q[[\ttau^1,\cdots,\ttau^{N}]]$. Since $p^{m}_{*}\Omega_{g,n+m}(-,\ttau,\cdots,\ttau) = O(\ttau^n)$, the summation is convergent under the $\ttau$-adic topology.
}.
Then, $\Omega^{\ttau}$  for  $\ttau \in \mathcal U$ gives another CohFT on $\Frob$ with the symmetric $2$-form $\eta$ remains unchanged (see e.g. \cite[Proposition 7.1]{Tel12}) and with $\mathbb F$-algebra $\mathbb A[[\ttau]]$.
In this way,   we obtain a family of CohFTs, called the shifted CohFTs, parameterized by the coordinate $\ttau \in \mathcal U$.
For the shifted CohFT $\Omega^{\ttau}$, we use the notation
$$
*_{\ttau},\qquad
\<-\>_{g,n}^{\ttau},\qquad
\bar\cF_{g}^{\ttau},\qquad
\cA^{\ttau},
$$
to denote, respectively, the quantum product, ancestor correlators, genus-$g$ ancestor potential, and total ancestor potential.

\subsection{Vacuum vector}
In the original definition of the CohFT, one requires the following {\it flat unit axiom}: there exist an element ${\bf 1}\in\Frob$, called the flat unit, such that
$$
p^*\Omega_{g,n}(v_1, \cdots , v_n)=\Omega_{g,n+1}(v_1, \cdots , v_n, {\bf 1}), \qquad 2g-2+n>0.
$$
If the flat unit axiom holds for $\Omega$, then it holds for $\Omega^{\ttau}$ with the same element ${\bf 1}$ and this element is the unity of the quantum product $*_{\ttau}$.

Usually when one consider the genus zero part of shifted CohFTs with flat unit,  a (formal) Frobenius manifold  structure \cite{Dub96}  on $\mathcal U$  naturally arise (see e.g. \cite{Tel12, LP04}).
We have $\Frob\cong T_{\ttau}\mathcal U$ by identifying $\phi_i$ with $\pd_{\ttau^i}$, and via this identification, the $2$-form $\eta$ defines a flat metric on $\mathcal U$.

In general, the flat unit does not necessarily exist, and the flat unit axiom is generalized to the vacuum axiom by Teleman~\cite{Tel12}.
\begin{definition}\label{def:vacuum}
A distinguished element $\vac(\givz)=\sum_{m\geq 0}\vac_m \givz^m\in \Frob[[\givz]]$ is called the {\it vacuum vector}, and $\Omega$ is called the CohFT with vacuum, if it satisfies the following {\it vacuum axiom}:
\beq\label{eqn:vacuum-axiom}
 p^*\Omega_{g,n}(v_1, \cdots , v_n)=\Omega_{g,n+1}\big(v_1, \cdots , v_n, \vac(\psi_{n+1}) \big), \qquad 2g-2+n>0,
\eeq
where $p\colon\Mbar_{g,n+1}\to \Mbar_{g,n}$ is the forgetful map defined by forgetting the last marked point.
\end{definition}

If the vacuum axiom holds for $\Omega$ with vacuum vector $\vac(\givz)$, then it holds for $\Omega^{\ttau}$ with vacuum vector $\vac^{\ttau}(\givz)$ which is determined by the following quantum differential equation (QDE) (see~\cite[Proposition 7.3]{Tel12} and~\cite[Proposition 1.4]{GZ25}):
$$
\givz\pd_{\ttau^i}\vac^{\ttau}(\givz)=\phi_i*_{\ttau}\vac^{\ttau}(\givz)-\phi_i, \qquad i=1,\cdots,N.
$$
It follows immediately ${\bf 1}^{\ttau}:=\vac_0^{\ttau}$ is the unity of the quantum product $*_{\ttau}$ and
$$\textstyle
\vac^{\ttau}(\givz)=\sum_{k\geq 0}\pd^k_{{\bf 1}^{\ttau}}({\bf 1}^{\ttau}) \, \givz^k .
$$

The vacuum vector is called flat if it is covariantly constant with respect to $\ttau$.
In such case one has $\vac^{\ttau}(\givz)={\bf 1}$ being exactly the flat unit.
When the vacuum vector is non-flat,  ${\bf 1}^{\ttau}=\vac_0^{\ttau}$ is called a {\it non-flat unit} and $\mathcal U$ is called a generalized Frobenius manifold~\cite{LQZ22}.

\subsection{Homogeneity condition}\label{sec:hom-CohFT}
We call the (generalized) Frobenius manifold $\mathcal U$ is homogeneous if there is an {\it Euler vector field} $E$ which has the form
$$\textstyle
E=\sum_i ((1-d_i)\ttau^i+r^i)\partial_i,
$$
where $d_i, r^i$ are some constants, such that the quantum product $*_{\ttau}$ and the metric $\eta$ satisfy the following equations: for $v_1,v_2\in T_{\ttau}\mathcal U\cong \Frob$,
\begin{align}
[E,v_1*_{\ttau}v_2]-[E,v_1]*_{\ttau}v_2-v_1*_{\ttau}[E,v_2]=&\, v_1*_{\ttau}v_2,\label{eqn:Euler-qp}\\
E\big(\eta(v_1,v_2)\big)-\eta([E,v_1],v_2)-\eta(v_1,[E,v_2])=&\, (2-\delta)\eta(v_1,v_2). \label{eqn:Euler-eta}
\end{align}
Here $\delta$ is a rational number called the {\it conformal dimension}.

Given the Euler vector field $E$, we call the shifted CohFT $\Omega^{\ttau}$ {\it homogeneous} if the following equations hold: for $2g-2+n>0$ and $v_1,\cdots,v_n\in T_{\ttau}{\mathcal U}\cong \Frob$,
\begin{align*}
&\, \textstyle E(\Omega^{\ttau}_{g,n}(v_1,\cdots,v_n))-\sum_{i=1}^{n}\Omega^{\ttau}_{g,n}(v_1,\cdots,[E,v_i],\cdots,v_n)\\
=&\, \big((g-1)\delta+n-\deg\big)\Omega^{\ttau}_{g,n}(v_1,\cdots,v_n),
\end{align*}
where the operator $\deg$ is defined by $\deg \alpha=k\cdot \alpha$ for $\alpha\in H^{2k}(\Mbar_{g,n},\mathbb Q)$.

We introduce the operator $\mu$ defined by the Euler vector field $E$ via
\beq\label{def:mu}
\mu(v)=(1-\tfrac{\delta}{2})v-\nabla_{v}E,
\eeq
where $\nabla$ is the Levi-Civita connection of $\eta$.
The vacuum axiom and the homogeneity condition for $\Omega^{\ttau}$ gives us
the following homogeneity condition for $\vac^{\ttau}(\givz)=\sum_{m\geq 0}\vac^{\ttau}_m \givz^m$:
$$
(\mu+\tfrac{\delta}{2}+m)\vac^{\ttau}_{m}=-E*_{\ttau}\vac^{\ttau}_{m+1}, \qquad m\geq 0.
$$
We refer readers to~\cite[\S 8.4]{Tel12} and~\cite[\S 1.5]{GZ25} for the proof.

Now we describe the homogeneity conditions at $\ttau=0$.
We call $\Omega$ is {\it homogeneous} if there exists an Euler vector field $E$ such that for $2g-2+n>0$ and $v_1,\cdots,v_n\in \Frob$,
\begin{align*}
&\, \textstyle p_{*}\Omega_{g,n+1}(v_1,\cdots,v_n,E|_{\ttau=0})
-\sum_{i=1}^{n}\Omega_{g,n}(v_1,\cdots,(\tfrac{\delta}{2}+\mu)v_i,\cdots,v_n)\\
=&\, \big((g-1)\delta-\deg\big)\Omega_{g,n}(v_1,\cdots,v_n).
\end{align*}
This is exactly the homogeneity condition for $\Omega^{\ttau}$ taking value at $\ttau=0$.
Clearly, this implies the following homogeneity condition for $\vac(\givz)=\sum_{m\geq 0}\vac_m \givz^m$: let $E_{0}=E|_{\ttau=0}$, then
\beq\label{eqn:vacuum-euler}
(\mu+\tfrac{\delta}{2}+m)\vac_{m}=-E_{0}*\vac_{m+1}, \qquad m\geq 0.
\eeq
By the definition of $\Omega^{\ttau}$, it is straightforward to check that the homogeneity condition holds for $\Omega^{\ttau}$ (resp. $\vac^{\ttau}(\givz)$) if and only if it holds for $\Omega$ (resp. $\vac(\givz)$).

\subsection{Actions on CohFTs}
\label{subsec:CohFT-reconstruction}
Now we recall the $R$-action and the $T$-action on an arbitrary CohFT $\Omega$ with state space $(\Frob,\eta)$.
We follow the setting in~\cite{PPZ15} and we refer readers to \cite{CGL18} for a generalization of the setting.

Let $R(\givz)$ be a formal power series
$$
R(\givz)={\rm Id}+R_1\givz+R_2\givz^2+\cdots, \quad R_k\in \mathrm{End}(\Frob)\otimes {\mathbb{E}}
$$
satisfying the symplectic condition
$
R^*(\givz)R(-\givz)={\rm Id},
$
where $R^*(\givz)$ is the adjoint of $R(\givz)$ with respect to $\eta$.
Here  $\mathbb E$ is an algebraic extension of the fractional field  of $\mathbb A$.
We also introduce the matrix $V$  associated to the $R$-matrix.:
\beq\label{eqn:V-matrix}
V(\givz, \givw)=\sum_{k,l\geq 0}V_{k,l}\givz^k\givw^l=\frac{{\rm Id}-R^*(-\givz)R(-\givw)}{\givz+\givw}.
\eeq

\begin{definition}[$R$-action]
Let $\mathcal{G}_{g,n}$ be the set of stable graphs of genus $g$ with $n$ legs. For $\Gamma\in \mathcal{G}_{g,n}$,  define
$
{\rm Cont}_\Gamma\colon \Frob ^{\otimes n}  \rightarrow   \mathbb E \otimes H^*(\Mbar_{g,n},\mathbb{Q})
$
by the following construction:
\begin{itemize}
  \item[1.] placing $\Omega_{g(v), n(v)}$ at each vertex $v\in \Gamma$,
  \item[2.] placing $R^*(-\psi_i)\cdot$ at each leg $l_i \in \Gamma$ labeled by $i=1,\ldots, n$,
  \item[3.] placing $V(\psi_{v'}, \psi_{v''})\phi_i\otimes \phi^i$ at each edge $e\in \Gamma$ connection vertexes $v'$ and $v''$.
\end{itemize}
Then we define  $R\cdot\Omega$ to be the  CohFT on the state space $(\Frob,\eta)$ with the maps
$$\textstyle
(R\cdot\Omega)_{g,n}=\sum_{\Gamma\in \mathcal{G}_{g,n}}\frac{1}{\left|{\rm Aut}(\Gamma)\right|}\xi_{\Gamma,*}\Cont_{\Gamma},
$$
where $\xi_{\Gamma}\colon \prod_{v\in \Gamma}\Mbar_{g(v),n(v)}\to \Mbar_{g,n}$ is the canonical map with image equal to the boundary stratum associated to the graph $\Gamma$, and its push-forward $\xi_{\Gamma,*}$ induces a homomorphism from the strata algebra on  $\prod_{v\in \Gamma}\Mbar_{g(v),n(v)}$ to the cohomology ring (see~\cite{CGL18, PPZ15} for details).
\end{definition}

Let $T(\givz)$ be a power series in $\givz$ starting from degree 1:
$$\textstyle
T(\givz)=\sum_{k\geq 1}T_k\givz^k, \quad T_k\in \Frob\otimes\mathbb E.
$$

\begin{definition}[$T$-action]
The $T$-action on the CohFT $\Omega$ is defined by
$$\textstyle
(T\cdot\Omega)_{g,n}(-)
=\sum_{m\ge 0}\frac{1}{m!}p^m_{*}\Omega_{g,n+m}(-\otimes T(\psi_{n+1})\otimes\cdots\otimes T(\psi_{n+m})),
$$
where $p^{m}\colon \overline{\mathcal{M}}_{g, n+m} \mapsto \Mgn$ is the forgetful map forgetting the last $m$ marked points.
We define the symmetric two form  $\eta$ on $T\cdot \Omega$  to be the same as one on $\Omega$.
\end{definition}

The $R$-actions and $T$-actions satisfy the following equation (see, e.g.,~\cite{PPZ15}):
$$
T\cdot T'\cdot \Omega=(T+T')\cdot \Omega,\qquad
R\cdot T\cdot \Omega=(RT)\cdot R\cdot \Omega.
$$

\subsection{Semisimple CohFT and its classification}\label{subsec:Frob-mfd}
A CohFT is called {\it semi-simple} if there exists a basis $\{e_\alpha\}_{\alpha=1}^{N}$, called the {\it canonical basis}, of   $\Frob \otimes \mathbb E$, such that
$e_\beta * e_\gamma=\delta_{\beta\gamma}e_\beta$.
The canonical basis is orthogonal, i.e., for $\beta\ne \gamma$, $\eta(e_\beta,e_{\gamma})=0$.
We define $\Delta_{\beta}=\eta(e_\beta,e_{\beta})^{-1}$ and let $\{\bar e_\beta=\Delta_\beta^{\frac{1}{2}}e_\beta\}_{\beta=1}^{N}$ be the normalized canonical basis.
The transition matrix between the flat basis and the normalized canonical basis is denoted by $\Psi$, given by $\phi_i= \Psi^{\bar\beta}_i \bar e_{\beta}$. Furthermore, we define $\widetilde\Psi^{\beta}_i=\Delta_\beta^{\frac{1}{2}}\Psi^{\bar\beta}_i$ as the transformation matrix between the flat basis and the canonical basis.
\begin{convention}\label{conv:notations-coord}
In this paper, we will always use notation $\{\phi_i\}$ for the flat basis, $\{e_\beta\}$ for the canonical basis and $\{\bar e_\beta\}$ for the normalized canonical basis. The corresponding coordinates $\{s_k\}$ on the big phase space will be denoted by $\{s_k^i\}$, $\{s^{\beta}_k\}$ and $\{s_k^{\bar\beta}\}$, respectively.
\end{convention}

Now we recall the classification result of semi-simple CohFTs according to Givental and Teleman~\cite{Giv01a,Giv01b,Tel12}.
For a $N$-dimensional semi-simple CohFT $\Omega$ on $(\Frob,\eta)$ with vacuum $\vac(\givz)$, there is an $R$-matrix $R(\givz)$ and a $T$-vector $T(\givz)$ such that
\beq\label{eqn:Giv-Tel}\textstyle
\Omega=R\cdot T\cdot (\oplus_{\beta=1}^{N}\Omega^{\rm KW_{\beta}}).
\eeq
Here $\Omega^{\rm KW_{\beta}}$ is the trivial CohFT on $(\mathbb F\{\bar e_\beta\}, \eta)$, i.e., $\Omega^{\rm KW_{\beta}}_{g,n}(\bar e_\beta,\cdots,\bar e_\beta)=1$.
In fact, let $\bar {\bf 1} = \sum_\beta  \bar e_\beta$, then the $T$-vector is explicitly given by:
\beq\label{eqn:T-vacuum}
T(\givz)=\givz\cdot {\bar {\bf 1}}-\givz \cdot R(\givz)^{-1}\vac(\givz).
\eeq

Furthermore, suppose the existence of the Euler vector field $E$, and let $E_0=E|_{\ttau=0}$.
Then the CohFT $\Omega$ constructed by equation~\eqref{eqn:Giv-Tel} is homogeneous if and only if the vacuum vector satisfies the homogeneity condition~\eqref{eqn:vacuum-euler} and the $R$-matrix satisfies the homogeneity condition:
\beq\label{eqn:R-matrix-euler}
[R_{m+1}, E_0*]=(m+\mu)R_{m},\qquad m\geq 0.
\eeq
See~\cite[Proposition 8.5 and Remark 8.2]{Tel12} for the proof.
If $E_0*$ has distinct nonzero eigenvalues, then equation~\eqref{eqn:R-matrix-euler} (resp. equation~\eqref{eqn:vacuum-euler}) uniquely and explicitly determines the $R$-matrix (resp. the vacuum vector) from data $(\Frob, \eta, *, E_0)$.
If not, the semi-simplicity ensures that for generic $\ttau\in\mathcal U$, $\Omega^{\ttau}$ is semi-simple and  $E*_{\ttau}$ has distinct nonzero eigenvalues.
One can first reconstruct the $R$-matrix $R^{\ttau}$ (resp. $\vac^{\ttau}$) for $\Omega^{\ttau}$ by equation $[R^{\ttau}_{m+1}, E*_{\ttau}]=(m+\mu)R^{\ttau}_{m}$ (resp. $(\mu+\tfrac{\delta}{2}+m)\vac^{\ttau}_{m}=-E*_{\ttau}\vac^{\ttau}_{m+1}$), $m\geq 0$, and then one gets $\Omega$ from $\Omega^{\ttau}$ by taking $\ttau=0$.
This result is called the {\it Givental--Teleman reconstruction theorem}.

\begin{remark}\label{rem:reconstruction-0}
In~\cite{Tel12}, Teleman has proved that for a semi-simple CohFT $\Omega$, the vacuum axiom (i.e., equation~\eqref{eqn:vacuum-axiom}) always holds for $n\geq 1$, but it may fail for $n=0$.
Furthermore, given a semi-simple CohFT $\Omega$ on $(\Frob,\eta)$ (without requiring the vacuum axiom),
one always has an $R$-matrix $R(\givz)$ and a $T$-vector $T(\givz)$ such that the CohFT
$\Omega':=R\cdot T\cdot (\Omega^{\rm KW})^{\oplus N}$ satisfies
$\Omega_{g,n}=\Omega'_{g,n}$ for $n\geq 1$.
If the CohFT $\Omega$ satisfies the vacuum axiom~\eqref{eqn:vacuum-axiom} for $n=0$ (which is always satisfied by $\Omega'$), then one has
$$\textstyle
\Omega_{g,0}=\frac{1}{2g-2} p_{*}(p^{*}\Omega_{g,0}\cdot \psi_{1})
=\frac{1}{2g-2} p_{*}\Omega_{g,1}(\vac(\psi_{1}))\cdot\psi_{1}
=\frac{1}{2g-2} p_{*}\Omega'_{g,1}(\vac(\psi_{1}))\cdot\psi_{1}
=\Omega'_{g,0},
$$
where we have used $p_{*}(p^{*}\alpha\cdot \beta)=\alpha\cdot p_{*}\beta$ and $p_{*}\psi_1=2g-2$ on $\Mbar_{g,0}$.
Otherwise, $\Omega_{g,0}\ne\Omega'_{g,0}$.
But if $\Omega$ is homogeneous, then as it has been pointed out in~\cite[Proposition 3.15]{CGG22}, the homogeneity condition
gives $\Omega_{g,0}=\Omega'_{g,0}$ except for the part of degree $(g-1)\delta$.
\end{remark}

We show some alternative formulations of the classification of semi-simple CohFTs.
Firstly,
$$
\Omega=R\cdot \tilde T\cdot \Omega^{\rm top},
$$
where $ \Omega^{\rm top} = T_1\givz \cdot (\oplus_{\beta=1}^{N}\Omega^{\rm KW_{\beta}})$ and $\tilde T(\givz) = T(\givz) -  T_1 \givz$.
We note here $T_1=\sum_{\beta}(1-\Delta_{\beta}^{-\frac12})\bar e_\beta$ and one has formula
$\Omega^{\rm top}_{g,n}(\bar e_{\beta_1},\cdots,\bar e_{\beta_n})=\delta_{\beta_1,\cdots,\beta_n}\cdot \Delta_{\beta_1}^{\frac{2g-2+n}{2}}$ for $2g-2+n>0$.
Secondly,
$$
\Omega=T_{\vac}\cdot R\cdot T_{R}\cdot (\oplus_{\beta=1}^{N}\Omega^{\rm KW_{\beta}}),
$$
where
$T_R(\givz):=\givz\cdot \bar{\bf 1}-\givz\cdot R(\givz)^{-1}{\bf 1}$ and $T_{\vac}(\givz):=\givz\cdot {\bf 1}-\givz\cdot \vac(\givz)$.

\begin{remark}
In \cite{PPZ15}, the composition $R\cdot T_R \cdot $ of actions $R\cdot$ and $T_{R}\cdot$ is denoted by ``$R.$".
\end{remark}

Now we show the generating series version of the classification of semi-simple CohFTs.
We denote by $\cD^{\rm KW}({\bf t};\hbar)$, ${\bf t}=\sum_{k}t_k\givz^k\in \mathbb C\{e\}[[\givz]]$ the Kontsevich--Witten tau-function, which is the total ancestor potential of the trivial CohFT.
We define $\cD_N^{\rm KW}({\bf q};\hbar)=\prod_{\beta=1}^{N}\cD^{\rm KW}({\bf q}^{\bar\beta};\hbar)$,
where ${\bf q}=\sum_{k}q_k\givz^k\in \Frob[[\givz]]$.
\begin{proposition}[\cite{Giv01a, Tel12}]\label{prop:AncR-action}
For the semi-simple CohFT $\Omega=R\cdot T\cdot (\oplus_{\beta=1}^{N}\Omega^{\rm KW_{\beta}})$,
its total ancestor potential $\cA$ is given by the following formula:
\beq\label{eqn:formula-A}
\cA({\bf s};\hbar)=\big(\widehat{T_{\vac}}\widehat{R}\widehat{\Delta}\cD_N^{\rm KW}\big)({\bf s};\hbar).
\eeq
The formula is explained as follows.
Firstly,
$$\textstyle
\big(\widehat{\Delta}\cD_N^{\rm KW}\big)({\bf q};\hbar)
=\prod_{\beta=1}^{N}\cD^{\rm KW}(\Delta_{\beta}^{\frac{1}{2}}{\bf q}^{\bar\beta};\Delta_{\beta}^{\frac{1}{2}}\hbar).
$$
Secondly, the operator $\widehat{R}$ acts as follows: for a function $\cG({\bf q};\hbar)$ on $\Frob[[\givz]]$,
$$
(\widehat{R}\mathcal G)({\bf s};\hbar)=[e^{\frac{\hbar^2}{2}V(\pd_{\bf q},\pd_{\bf q})}\mathcal G]({\bf q(s)};\hbar)
$$
where $V(\pd_{\bf q},\pd_{\bf q})=\sum V^{\bar\alpha,\bar\beta}_{k,l}\pd_{q^{\bar\alpha}_{k}}\pd_{q^{\bar\beta}_{l}}$
whose coefficients are given by $V^{\bar\alpha,\bar\beta}_{k,l}=[\givz^k\givw^{l}]\eta(\bar e_\beta,V(\givz,\givw)\bar e_\alpha)$
and the coordinate transform ${\bf q(s)}$ is given by ${\bf q}(\givz)-\givz\cdot{\bf 1}=R^{-1}(\givz)({\bf s}(\givz)-\givz\cdot {\bf 1})$.
Lastly, the $T_{\vac}$-action behaves as a shift on the coordinate: for a function $\cG({\bf s};\hbar)$ on $\Frob[[\givz]]$,
$$
(\widehat{T_{\vac}}\mathcal G)({\bf s};\hbar)=\mathcal G({\bf s};\hbar)|_{s_k\to s_k-\vac_{k-1}, k\geq 2}.
$$
\end{proposition}
\begin{remark}
The original description of formula~\eqref{eqn:formula-A} is proved by Givental by using quantization of quadratic Hamiltonian~\cite{Giv01a}, here we have adopted Givental’s formula as the definition for simplicity.
\end{remark}
\begin{remark}
We call $\ttau\in\mathcal U$ a semi-simple point if the shifted CohFT $\Omega^{\ttau}$ is semi-simple,
all the results in this subsection hold for $\Omega^{\ttau}$ at semi-simple point $\ttau$.
For shifted CohFTs, the terms $\Delta_{\beta}$, $\vac$, $R$, $T$, and $\cA$ may depend on the parameter $\ttau$, and we use a superscript to denote this dependence: $\Delta^{\ttau}_{\beta}$, $\vac^{\ttau}$, $R^{\ttau}$, $T^{\ttau}$, and $\cA^{\ttau}$, we have
$$
\cA^{\ttau}({\bf s};\hbar)=\big(\widehat{T^{\ttau}_{\vac}}\widehat{R^{\ttau}}\widehat{\Delta^{\ttau}}\cD_N^{\rm KW}\big)({\bf s};\hbar).
$$
\end{remark}

\subsection{Descendent invariants}  \label{sec:descedentforCohFT}
The descendent invariants are usually defined by the intersection numbers on the moduli space of certain enumerative geometry.
It is usually connected to the total ancestor potential by
$$
\cD=e^{F_1(\ttau)} \widehat {(S^{\ttau})^{-1}} \cA^{\ttau}
$$
according to Givental’s version of the Kontsevich-Manin formula \cite{Giv01a,KM94}.
Here $S^\ttau=S^{\ttau}(\givz)=\sum_{k\geq 0} S^{\ttau}_k \givz^{-k}$ is a matrix valued series known as the Givental's $S$-matrix.
The $S$-matrix consist of the genus zero data of the descendent theory, it satisfies $S^{\ttau,*}(-\givz)S^{\ttau}(\givz)=\id$ and the QDE:
\begin{equation}\label{eqn:QDE}
	\givz\partial_{\ttau^i} S^{\ttau}(\givz)=\phi_i*_{\ttau}S^{\ttau}(\givz),\qquad i=1,\cdots,N.
\end{equation}
Here $S^{\ttau,*}$ is the adjoint matrix of $S^{\ttau}$ with respect to the symmetric $2$-form $\eta$.

Given a CohFT $\Omega^{\ttau}$, a fixed choice of such a solution $S^{\ttau}(\givz)$ to the QDE~\eqref{eqn:QDE} is called an $S$-calibration of the Frobenius manifold / CohFT.
For CohFT with flat unit {\bf 1} and calibrated by an $S$-matrix, we can define its total descendent potential by the Kontsevich--Manin formula.
For CohFT with non-flat unit, to define the total descendent potential, we need to further introduce the $\dvac$-calibration~\cite{GZ25}.
\begin{definition}\label{def:dvac}
The $\dvac$-vector $\dvac^{\ttau}(\givz)$ is an vector-valued series in $\Frob((\givz^{-1}))$ which satisfies
$$
	\givz\pd_{\ttau^i}\dvac^{\ttau}(\givz)=\phi_i*_{\ttau}\dvac^{\ttau}(\givz)-\phi_i,\qquad i=1,\cdots,N.
$$
Note the $\dvac$-vector is not uniquely determined by the above conditions. A choice of $\dvac^{\ttau}(z)$ is called a $\dvac$-calibration of the (generalized) Frobenius manifold (of the CohFT).
\end{definition}
Given an $S$-matrix $S^{\ttau}(\givz)$ and $\dvac$-vector $\dvac^{\ttau}(\givz)$, we define the $J$-function by
$$
J^{\ttau}(-\givz)=-\givz S^{\ttau,*}(-\givz)\dvac^{\ttau}(\givz).
$$
The $J$-function satisfies $\pd_{\ttau^i}J^{\ttau}(-\givz)=S^{\ttau,*}(-\givz)\phi_i$ and has form
$$\textstyle
J^{\ttau}(-\givz)=\ttau-\ttau_0-\givz v(\givz)+\sum_{k\geq 0}\sum_{a=1}^{N}J^{\ttau}_{k,a}\phi^a\cdot (-\givz)^{-k-1},
$$
for some constant vectors $\ttau_0\in \Frob$ and $v(\givz)\in \Frob[\givz]$.
We also define the matrix $W^{\ttau}$:
\beq\label{eqn:W-S}
W^{\ttau}(\givz,\givw)=\sum_{k,l\geq 0}W^{\ttau}_{k,l}\givz^{-k}\givw^{-l}:=\frac{S^{\ttau,*}(\givz)S^{\ttau}(\givw)-\id}{\givz^{-1}+\givw^{-1}}.
\eeq

We introduce a linear function $J^{\ttau}({\bf t})$ and a quadratic function $W^{\ttau}({\bf t,t})$, defined by:
$$\textstyle
J_{-}^{\ttau}({\bf t}):=\sum_{k\geq 0}\sum_{a=1}^{N}J^{\ttau}_{k,a}\cdot t_k^{a},\qquad
W^{\ttau}({\bf t},{\bf t}):=\sum_{k,l\geq 0}\eta(t_k, W_{k,l}^{\ttau}\cdot t_l).
$$
We also introduce two functions, $F_0(\ttau)$ and $F_1(\ttau)$, defined by
$$\textstyle
\pd_{\ttau^i}F_0(\ttau)=J_{0,i},\qquad
\pd_{\ttau^i}F_1(\ttau)=\int_{\Mbar_{1,1}}\Omega^{\ttau}_{1,1}(\phi_i),\qquad
i=1,\cdots,N,
$$
determined up to constants $F_0(0)$ and $F_1(0)$, respectively.
We denote
$$\textstyle
F_{\rm un}^{\ttau}({\bf t};\hbar) =F_1(\ttau)+\frac{1}{\hbar^2}\big(F_0(\ttau)+J_{-}^{\ttau}({\bf t})+\frac{1}{2}\, W^{\ttau}({\bf t},{\bf t})\big).
$$

\begin{definition}
The total descendent potential $\cD({\bf t};\hbar)$ for an $S$- and $\dvac$-calibrated CohFT $\Omega$ is defined by the following generalized Kontsevich--Manin formula:
\beq   \label{eq:descendent-ancestor}
\cD({\bf t};\hbar) :=e^{F_{\rm un}^{\ttau}({\bf t}-\ttau; \hbar)} \cA^{\ttau}({\bf s(t)};\hbar),
\eeq
where the coordinate transformation ${\bf s(t)}$ is given by
${\bf s}(\givz)=[S^{\ttau}(\givz){\bf t}(\givz)]_{+}-\ttau$.
\end{definition}
It is proved in~\cite{Giv01a} (for semi-simple cases) and ~\cite{GZ25} (for general cases) that
the total descendent potential $\cD({\bf t};\hbar)$ does not depend on the base point $\ttau$.
Therefore, one can define $\cD({\bf t};\hbar)$ at any base point $\ttau$, in particular, at $\ttau=0$.
We define
$$
S=S^{\ttau}|_{\ttau=0},\quad
\dvac=\dvac^{\ttau}|_{\ttau=0},\quad
J=J^{\ttau}|_{\ttau=0},\quad
W=W^{\ttau}|_{\ttau=0},\quad
F_{\rm un}:=F_{\rm un}^{\ttau}|_{\ttau=0},
$$
then we have
$$
\cD({\bf t};\hbar) =e^{F_{\rm un}({\bf t};\hbar)} \cA({\bf s(t)};\hbar),
$$
where the coordinate transformation ${\bf s(t)}$ is given by ${\bf s}(\givz)=[S(\givz){\bf t}(\givz)]_{+}$.

Given the total descendent potential $\cD$, the genus-$g$ descendent potential is defined by the following equation:
$$\textstyle
\log(\cD({\bf t};\hbar))=\sum_{g\geq 0}\hbar^{2g-2}\cF_g({\bf t}),
$$
and we define the descendent correlators $\<-\>^{\cD}_{g,n}$ by
\beq\label{def:bracket-des}
\<\phi_{a_{1}}\psi^{k_{1}},\cdots,\phi_{a_{n}}\psi^{k_{n}}\>^{\cD}_{g,n}
:=\pd_{t^{a_1}_{k_1}}\cdots\pd_{t^{a_n}_{k_n}}\cF_{g}({\bf t})|_{\bf t=0}.
\eeq
In the follows, we will drop the superscript $\cD$ in $\<-\>_{g,n}^{\cD}$ if no confusion arises.
The generalized Kontsevich--Manin formula can be rewritten in the correlator form as follows:
	\beq\label{eqn:des-anc-S}
	\<\phi_{a_1}\psi^{k_1},\cdots,\phi_{a_n}\psi^{k_n}\>^{\cD}_{g,n}
	=\<S(\bar\psi)\phi_{a_1}\bar\psi^{k_1},\cdots,S(\bar\psi)\phi_{a_n}\bar\psi^{k_n}\>^{\Omega}_{g,n},
\qquad 2g-2+n>0,
	\eeq
	where the correlators with insertion $\phi_{a}\bar\psi^{k}$ of negative $k$ are set to be $0$.
\begin{remark}
It is easy to see that $\<\phi_{a_1},\cdots,\phi_{a_n}\>^{\Omega}_{g,n}=\<\phi_{a_1},\cdots,\phi_{a_n}\>^{\cD}_{g,n}$ for $2g-2+n>0$.
Hence, the notation $\<\phi_{a_1},\cdots,\phi_{a_n}\>_{g,n}$ causes no confusion in this case.
For $2g-2+n\leq 0$, we always interpret $\<\phi_{a_1},\cdots,\phi_{a_n}\>_{g,n}$ as the descendent correlators, defined by equation~\eqref{def:bracket-des}.
\end{remark}

\section{Topological recursion and its relation with CohFT}
\label{sec:TR}
In this section, we first review the topological recursion proposed by Eynard and Orantin, which produces multi-differentials from the spectral curve data.
Then we introduce several important properties of the multi-differentials that we will use in this paper.
Finally, we explain its relations with CohFT.

\subsection{Geometry of curves}
Let $\Sigma$ be a genus $\frak g$ Riemann surface with fixed symplectic basis $\{\cyA_i,\cyB_i\}_{i=1}^{\frak g}$ for $H_{1}(\Sigma,\mathbb C)$, there exist unique $\frak g$ independent holomorphic $1$-forms $du_j(z)$ on $\Sigma$ satisfying $\oint_{z\in \cyA_i}du_j(z)=\delta_{i,j}$, and their $B$-cycle integrals give
$$
\oint_{z\in \cyB_i}du_j(z)=\tau_{ij}.
$$
Here $\tau=(\tau_{ij})$ gives a point in the Siegel upper half-plane.

For $\frak g(\Sigma)>0$, we introduce the theta-function on $\Sigma$ with characteristic $[\cmn]\in\mathbb R^{2\frak g}$
(here $[\cmn]$ represents the coordinate of the point $\nu+\tau \mu\in\mathbb R^{\frak g}\oplus \tau\mathbb R^{\frak g}$):
$$
\theta[\cmn](w|\tau) =\sum_{n\in\mathbb Z^{\frak g}} e^{\pi i (n+\mu)^{t}\tau (n+\mu)+2\pi {\bf i} (n+\mu)^{t}(w+\nu)},
$$
where $w\in\mathbb C^{\frak g}$.
It is straightforward to check that for $m,k\in \mathbb Z^{\frak g}$,
\beq\label{eqn:quasiperiod-theta}
\theta[\cmn](w+k+\tau m |\tau)=e^{-\pi i m^{t} \tau m -2\pi {\bf i} m^{t} w +2\pi {\bf i} (k\mu-m\nu)}\theta[\cmn](w|\tau).
\eeq
For ${\frak g}(\Sigma)=0$, we formally define $\theta[\cmn](w|\tau)=w$ where $\mu,\nu,\tau$ are just symbols and mean nothing.
Let ${\bf c}=\frac{{\bf a}+\tau {\bf b}}{2}$ be a characteristic with ${\bf a,b}\in\mathbb Z^{\frak g}$, we call ${\bf c}$ is odd (resp. even) if the theta-function $\theta[{\bf c}](w|\tau)$ is an odd (resp. even) function of $w$, this is equivalent to say ${\bf a\cdot b}$ is an odd (resp. even) integer number.
In the follows of this paper, we will always use ${\bf c}$ to denote an odd characteristic.

We introduce the Bergman kernel $B(z_1,z_2)$ of the Riemann surface $\Sigma$:
$$
B(z_1,z_2):=d_{z_1}d_{z_2}\log(\theta[{\bf c}](u(z_1)-u(z_2)|\tau)),
$$
where $u(z)=\int_{z'=p}^{z}du(z)$ for a fixed point $p$.
For genus $\frak g=0$ case, the formula means
$$
B(z_1,z_2)=d_{z_1}d_{z_2}\log(z_1-z_2)=\tfrac{dz_1dz_2}{(z_1-z_2)^2}.
$$
One can check that the Bergman kernel does not depend on ${\bf c}$.
By equation~\eqref{eqn:quasiperiod-theta},
$$
\oint_{z'\in \cyA_i} B(z,z')=0,\qquad
\oint_{z'\in \cyB_i} B(z,z')=2\pi {\bf i} du_i(z),\qquad
i=1,\cdots,\frak g.
$$

For the Riemann surface $\Sigma$ of genus $\frak g\geq 1$, we introduce a holomorphic $1$-form on $\Sigma$:
$$
dh_{\bf c}(z):=\sum_{i=1}^{\frak g}\frac{\pd \theta[{\bf c}]}{\pd w_i}(0|\tau)du_i(z),
$$
and we introduce the prime form $E(z_1,z_2)$ by
\beq\label{def:primeform}
E(z_1,z_2):=\frac{\theta[{\bf c}](u(z_1)-u(z_2)|\tau)}{\sqrt{dh_{\bf c}(z_1)}\sqrt{dh_{\bf c}(z_2)}}.
\eeq
Here $\sqrt{dh_{\bf c}(z)}$ is well defined because zeros of $dh_{\bf c}(z)$ are all double, see~\cite{Fay73} for detailed explanation.
For $\frak g=0$ case, we formally define
$$
dh_{\bf c}(z):=dz,\qquad E(z_1,z_2):=\frac{z_1-z_2}{\sqrt{dz_1}\sqrt{dz_2}}.
$$
Fix a point $z_0\in\Sigma$, let $\eta$ be a local coordinate near $z_0$, we have the following formula relating the Bergman kernel $B$ and the prime form $E$:
\beq\label{eqn:int-bergman}
\exp\bigg(\frac{1}{2}\int_{\eta(z_0)}^{\eta(z)}\!\!\!\int_{\eta(z_0)}^{\eta(z)}\Big(B(z_1,z_2)-\frac{d\eta(z_1)d\eta(z_2)}{(\eta(z_1)-\eta(z_2))^2}\Big)\bigg)
=\frac{\eta(z)-\eta(z_0)}{E(z,z_0)\sqrt{d\eta(z)d\eta(z_0)}}.
\eeq

\subsection{Definition of the topological recursion and the structure of $\omega_{g,n}$}
\label{sec:def-TR}
We recall how a sequence of multi-differentials $\{\omega_{g,n}\}_{g,n\geq 0}$ is defined from the spectral curve data $\cC=(\Sigma,x,y)$~\cite{EO07}.
We denote by $\{z^\beta\}_{\beta=1}^{N}$ the set of critical points of $x$, where $N$ is the number of critical points. For any $z$ around the   $z^\beta$ we have the local involution $\bar{z}\in \Sigma$ such that $x(\bar{z})=x(z)$. Note that $\bar{z}$ is locally dependent on $z^\beta$, and we omit the dependence of $\beta$ in the notation to avoid complexity.
In this paper, for a positive integer $n$, we always use the notation $[n]$ to denote the set of integers $\{1,\cdots,n\}$.
\begin{definition}\label{def:omegagn}
	Given a set of spectral curve data $\cC=(\Sigma,x,y)$, a family of multi-differentials $\{\omega_{g,n}\}_{g,n\geq 0}$ are defined as follows. For $2g-2+n\leq 0$, $\omega_{0,0}(z):=0$, $\omega_{1,0}(z):=0$, and
	$$
	\omega_{0,1}(z):=y(z)dx(z), \quad \omega_{0,2}(z_1, z_2):=B(z_1, z_2).
	$$
	Let $K_\beta$ be recursion kernel  around the critical point $z^\beta$, defined by
	$$
	K_\beta(z_0,z)=\frac{\int_{z'=\bar z}^{z}B(z_0,z')}{2(y(z)-y(\bar z))dx(z)}.
	$$
	Then for $2g-2+n+1>0$, the multi-differentials $\omega_{g,n+1}$ are defined recursively as follows:
	\begin{align*}
		\omega_{g,n+1}(z_0,z_{[n]})
		:=\sum\limits_{\beta}\mathop{\Res}_{z=z^\beta}K_\beta(z_0,z)\bigg(\omega_{g-1,n+2}(z,\bar z,z_{[n]})
		+\sum^{\prime}_{\substack{g_1+g_2=g\\I\bigsqcup J=[n]}}\omega_{g_1,|I|+1}(z,z_{I})\omega_{g_2,|J|+1}(\bar z,z_{J})\bigg).
	\end{align*}
	Here for any subset $I\subset [n]$, $z_{I}:=\{z_i\}_{i\in I}$.
	The symbol $\sum\limits^{\prime}$ means we exclude $\omega_{0,1}$ in the summation.
	For $g\geq 2$,
	\beq\label{def:wg0}
	\omega_{g,0}:=\frac{1}{2-2g}\sum_{\beta}\mathop{\Res}_{z=z^\beta}\omega_{g,1}(z)\int_{z'=z^\beta}^{z}\omega_{0,1}(z').
	\eeq
\end{definition}
\begin{remark}
	In the original literature of Eynard and Orantin~\cite{EO07}, the $0$-forms $\omega_{0,0}$ and $\omega_{1,0}$ are also defined (\cite[equation (4.17) and (4.18)]{EO07}) which do not necessarily vanish. These two terms play the role of $F_0(\ttau)$ and $F_1(\ttau)$ on the geometric side.
	Since these two terms do not affect the integrability, for simplicity, we set them to be $0$.
\end{remark}

From the definition of the topological recursion, it is deduced~\cite{EO07, EO08} that
$\omega_{g,n}$ is a symmetric meromorphic multi-differential on $n$ copies of $\Sigma$.
Moreover, for $2g-2+n>0$, at each copy of $\Sigma$, $\omega_{g,n}$ has poles only at the critical points with vanishing residues, and the order of poles is at most $6g-4+2n$.

We will also use the following scaling property of multi-differential $\omega_{g,n}$:  for a set of spectral curve data $\cC=(\Sigma, x, y)$ and its rescaling $\tilde \cC=(\Sigma, x/c_1, y/c_2)$,  the corresponding multi-differentials $\omega_{g,n}$ and  $\tilde\omega_{g,n}$ are related to each other by:
\beq\label{eqn:scaling-xy}
\tilde \omega_{g,n}=(c_1 c_2)^{2g-2+n}\omega_{g,n}.
\eeq

Around each critical point $z^\beta$, we define the local Airy coordinates $\eta^\beta=\eta^{\beta}(z)$ by
$$\textstyle
x(z)=x^\beta + \frac{1}{2}\,\eta^\beta(z)^2,
$$
where $x^\beta=x(z^\beta)$.
Following \cite{Eyn11, Eyn14}, for each $\beta \in [N]$, $k\geq 0$, we introduce
\beq\label{def:xi}
d\xi^{\bar\beta}_{k}(z) :=-\mathop{\Res}_{z'=z^\beta}\,\frac{(2k-1)!!}{\eta^\beta(z')^{2k+1}}\cdot B(z',z).
\eeq
The differential $d\xi^{\bar\beta}_{k}(z)$ is a  globally defined meromorphic $1$-form on $\Sigma$ with single pole of order $2k+2$ at $z^\beta$.
This differential 1-form can be viewed as a global section in $H^0\left(\Sigma,\omega_{\Sigma}\big( (2k+2) z^\beta \big)\right)$.
Near $z^\beta$ it admits the Laurent expansion in terms of $\eta^\beta$, with vanishing residue:
\beq\label{eqn:localxi}
d\xi^{\bar\beta}_{k}(z)
=\left(-\frac{(2k+1)!!}{(\eta^{\beta})^{2k+2}}+\mathrm{regular\ part}\right)d\eta^\beta.
\eeq

It is shown in~\cite{Eyn14} that the differentials $d\xi^{\bar\beta}_{d}(z)$ serve as a preferred basis of $\omega_{g,n}$ with $2g-2+n>0$: there are some coefficients $\< - \>_{g,n}^{\rm TR}$ such that
\beq\label{eqn:str-omegagn}
	\omega_{g,n}(z_1,\cdots,z_n)
	=\sum_{\substack{k_1,\cdots,k_n\geqslant 0\\ \beta_1,\cdots,\beta_n\in [N]}}
	\<\bar e_{\beta_1}\bar\psi^{k_1},\cdots,\bar e_{\beta_n}\bar\psi^{k_n}\>^{\rm TR}_{g,n}
	d\xi^{\bar\beta_1}_{k_1}(z_1)\cdots d\xi^{\bar\beta_n}_{k_n}(z_n).
\eeq
We note that the summation in the expansion for $\omega_{g,n}$ contains only finite terms since the pole of $d\xi_k^{\bar\beta}(z)$ has degree $2k+2$ and  the poles of $\omega_{g,n}$ are bounded by degree $6g-4+2n$.
For $(g,n)=(0,2)$, we have the following expansion of $\omega_{0,2}(z_1,z_2)=B(z_1,z_2)$ with the first variable $z_1$ near the critical point $z^\beta$~\cite[equation (4.9)]{Eyn14}:
\beq\label{eqn:equiv-B}
B(z_1, z_2)-B(\bar{z}_1,z_2)=-2\sum_{k\ge 0}\frac{\eta^\beta(z_1)^{2k}}{(2k-1)!!}d\eta^\beta(z_1)d\xi^{\bar\beta}_{k}(z_2).
\eeq

\begin{lemma}\label{lem:res-dxi}
For any non-negative integers $k,m$ and $\alpha,\beta\in [N]$, we have
$$
\mathop{\Res}_{z=z^{\alpha}}x(z)^m d\xi_k^{\bar\beta}(z)=0.
$$
Furthermore, by equation~\eqref{eqn:str-omegagn}, we see for any $ m\geq 0$, $\alpha\in [N]$ and $i\in [n]$,
$$
\mathop{\Res}_{z=z^{\alpha}}x(z_i)^m \omega_{g,n}(z_1,\cdots,z_n)=0,\qquad 2g-2+n>0.
$$
\end{lemma}
\begin{proof}
Just notice that near the critical point $z^{\alpha}$,
the local expansion of $x(z)^n$ is a polynomial of $(\eta^{\alpha})^2$ while the negative part of the local expansion of $d\xi_k^{\bar \beta}(z)$ is $-\delta_{\alpha\beta} \frac{(2k+1)!!}{(\eta^\alpha)^{2k+2}}d\eta^\alpha$.
\end{proof}

\subsection{CohFT from the topological recursion}
\label{sec:TR-CohFT}
We recall how a semi-simple CohFT $\Omega$ is constructed from the spectral curve data $\cC=(\Sigma,x,y)$ \cite{DOSS14,Eyn14}.

The state space associated with $\cC$ is taken to be $\Frob:=\sspan\{\bar e_\beta\}_{\beta=1}^{N}$ and the symmetric $2$-form $\eta$ is given by $\eta(\bar e_\beta,\bar e_\gamma)=\delta_{\beta,\gamma}$.
The $R$-matrix and $T$-vector associated with the spectral curve data $\cC$ are defined by:
\beq\label{def:EO-R}
\frac{\givz}{\sqrt{2\pi \givz}}\cdot\int_{\mathfrak L_{\gamma}}e^{-x(z)/\givz}d\xi_0^{\bar\beta}(z)
\asymp e^{-x^\gamma/\givz} \cdot \eta(R(-\givz)\bar e_\gamma,\bar e_\beta),
\eeq
\beq\label{def:EO-T}
\frac{\givz}{\sqrt{2\pi \givz}}\cdot \int_{\mathfrak L_{\gamma}}e^{-x(z)/\givz}dy(z)
\asymp e^{-x^\gamma/\givz}\cdot \big(\givz-\eta(\bar e_\gamma,T(\givz))\big).
\eeq
Here $\mathfrak L_\gamma$, called the Lefschetz thimble, is a path in $\Sigma$ passing only one critical point $z^\gamma$ such that for any $z\in {\mathfrak L_\gamma}$, $x(z)-x(z^\gamma)\in \mathbb R_{\geq 0}$.
It is proved by Eynard~\cite{Eyn14} that  matrix $R(\givz)$ satisfies the symplectic condition: $R^{*}(-\givz)R(\givz)=\id$.
\begin{definition} Let $R$ and $T$ be defined as above.
The CohFT $\Omega$ associated with the spectral curve data {$\cC$} is defined as follows:
$$
\Omega:=R\cdot T\cdot  (\oplus_{\beta=1}^{N}\Omega^{\rm KW_{\beta}}).
$$
\end{definition}
It is clear that this is a semi-simple CohFT with vacuum, where the vacuum vector $\vac(\givz)$ is determined by the $R$-matrix and the $T$-vector via equation~\eqref{eqn:T-vacuum}.
Given the CohFT $\Omega$, the ancestor correlator $\<-\>^{\Omega}_{g,n}$ and the total ancestor potential $\cA({\bf s};\hbar)$ are defined by equations~\eqref{def:bracket-anc} and \eqref{def:FandA-CohFT}, respectively.
Moreover, ${\mathcal A}({\bf s};\hbar)$ can be explicitly computed using the formula~~\eqref{eqn:formula-A}.

\begin{remark}
Let $\Delta_{\beta}^{-\frac{1}{2}}=1-T_1^{\bar\beta}=\frac{y'(z^\beta)}{\sqrt{x''(z^\beta)}}$,
then $\{e_\beta=\Delta_{\beta}^{-\frac{1}{2}}\bar e_\beta\}_{\beta=1}^{N}$ gives a canonical basis.
Suppose we have a flat basis $\{\phi_i\}_{i=1}^{N}$, we define $\Psi$ and $\widetilde\Psi$ by
$\phi_i=\sum_\beta \Psi^{\bar\beta}_{i} \bar e_\beta=\sum_\beta \widetilde\Psi^{\beta}_{i} e_\beta$.
\end{remark}

By comparing the two types of graph sums of the Givental--Teleman reconstruction theorem and of the topological recursion,
Dunin-Barkowshi--Orantin--Shadrin--Spitz~\cite{DOSS14} (similar arguments can be also found in \cite{Eyn14, Mil14}) proved that for $2g-2+n>0$,
\beq\label{eqn:wgn-zeta}
	\omega_{g,n}(z_1,\cdots,z_n)
	=\sum_{\substack{k_1,\cdots,k_n\geqslant 0\\ \beta_1,\cdots,\beta_n\in [N]}}
	\<\bar e_{\beta_1}\bar\psi^{k_1},\cdots,\bar e_{\beta_n}\bar\psi^{k_n}\>^{\Omega}_{g,n}
	d\zeta^{\bar\beta_1}_{k_1}(z_1)\cdots d\zeta^{\bar\beta_n}_{k_n}(z_n),
\eeq
where the basis $\{d\zeta_k^{\bar \beta}(z)\}$ is related with $\{d\xi_{l}^{\bar \gamma}(z)\}$ by the following formula
\beq\label{eqn:zeta-xi}
d\zeta_k^{\bar\beta}(z)=\sum_{l=0}^{k}(R_{l})^{\bar\beta}_{\bar\gamma} \,d\xi_{k-l}^{\bar\gamma}(z).
\eeq
Moreover, the basis $\{d\zeta_k^{\bar \beta}(z)\}$ satisfies the equation (this is a rewriting of the results in ~\cite{DOSS14,Eyn14}, see, e.g.,~\cite[equation~(4.7)--(4.9)]{DOSS14}):
\beq\label{def:zeta}
d\zeta_{k}^{\bar\beta}(z)=\bigg(-d\circ \frac{1}{dx(z)}\bigg)^kd\xi_0^{\bar\beta}(z),\quad k\geq 0,\, \beta\in [N].
\eeq
From equation~\eqref{eqn:wgn-zeta} we see that $d\zeta^{\bar\beta}_k(z)$ can be viewed as the dual basis of the ancestor normalized canonical basis basis $\bar e_\beta\bar\psi^k$.
We denote by $d\zeta_k^{\beta}(z)$ and $d\zeta_k^i(z)$ the dual basis of $e_\beta\bar\psi^k$ and $\phi_i\bar\psi^k$, respectively.

Let $\zeta_k^{\bar\beta}(z)$ be the (might multi-valued) function defined by $\zeta_k^{\bar\beta}(z):=\int d\zeta_k^{\bar\beta}(z)$
and we introduce the operator $D$ given by
\be
D:=-\frac{d}{dx(z)},
\ee
then equation~\eqref{def:zeta} shows $\zeta^{\bar\beta}_k(z)=D^k\zeta^{\bar\beta}_0(z)$.
In this paper, we also denote $\zeta^{\bar\beta}_0(z)$ by $\zeta^{\bar\beta}(z)$.

We show how the vacuum vector is computed directly from the spectral curve data:
\begin{proposition} The vacuum vector associated with the spectral curve data $\cC$ has the following formula:
$$
\vac(\givz)=-\sum_{\beta}\bar e_{\beta}\sum_{m\geq 0}\givz^{m}\sum_{\gamma}\mathop{\Res}_{z\to z^\gamma}  y(z)d\zeta_{m}^{\bar\beta}(z).
$$
\end{proposition}
\begin{proof}
By equation~\eqref{eqn:localxi} and \eqref{eqn:zeta-xi}, near $z=z^\gamma$ we have
$$\textstyle
d\zeta_{m}^{\bar\beta}=-\sum_{l=0}^{m}(R_{l})^{\bar\beta}_{\bar\gamma}\cdot \frac{(2m-2l+1)!!}{(\eta^\gamma)^{2m-2l+2}}\, d\eta^\gamma+{\text{ regular part }}.
$$
By equation~\eqref{def:EO-T}, near $z=z^\gamma$ we have
$$\textstyle
y(z)=\eta^\gamma-\sum_{k\geq 0}\frac{T_{k+1}^{\bar\gamma}}{(2k+1)!!}\cdot (\eta^\gamma)^{2k+1}+{\text{ even part }}.
$$
The Proposition follows from equation~\eqref{eqn:T-vacuum} and straightforward  computations.
\end{proof}

\section{Descendent invariants from topological recursion}
\label{sec:TR-Geo}
In this section, we first introduce the TR descendent invariants, from a spectral curve data. We then examine the relationship between TR descendents and geometric descendents, establishing the TR-Geo correspondence as shown in the diagram in the Introduction.

\subsection{TR descendents }
\label{sec:TR-KP}
We recall that the multi-differentials $\omega_{g,n}$ with $2g-2+n>0$ are meromorphic with  poles only at critical points, so they have Taylor expansions at points away from the critical points, particularly at the boundary points.
For $(g,n) = (0,2)$, to take the expansion of $\omega_{0,2}$ at an arbitrary point, the pole structure need to be considered.

\begin{definition} [TR descendents]\label{def:TR-des}
	Pick $\Lambda=(\lambda_1,\cdots,\lambda_m)$, such that $\lambda_i$ is the local coordinate near $b_i$ satisfying $\lambda_i(b_i)={\infty}$ for each $i\in [m]$.
	For $(i_1,\cdots,i_n)\in [m]^{\times n}$ and $(k_1,\cdots,k_n)\in \mathbb Z_{>0}^{\times n}$,
we define   the \emph{TR descendent invariants}
	$\<\alpha^{i_1}_{k_1},\cdots,\alpha^{i_n}_{k_n}\>^{\Lambda}_{g,n}$
	by taking the   expansion  of the multi-differential forms ${\omega_{g,n}}$ at the boundary points. Namely,
	for $2g-2+n\geq0$, near  $z_{1}=b_{i_{1}},\cdots, z_{n}=b_{i_{n}}$  we define
	\beq\label{eqn:stable-omega-gn-boundary}
	\omega_{g,n}(z_1,\cdots,z_n)
	=\delta_{g,0}\delta_{n,2}\frac{\delta_{i_1,i_2}d\lambda_{i_1,1}d\lambda_{i_2,2}}{(\lambda_{i_1,1}-\lambda_{i_2,2})^2}
	+\sum_{k_1,\cdots, k_n\geq 1}\<\alpha^{i_1}_{k_1},\cdots,\alpha^{i_n}_{k_n}\>_{g,n}^{\Lambda } \frac{d\lambda_{i_1,1}^{-k_1}\cdots d\lambda_{i_n,n}^{-k_n}}{k_1\cdots k_n}.
	\eeq
	For the cases $(g,n) = (0,0), (0,1), (1,0)$, all the invariants $\<-\>^{\Lambda}_{g,n}$ are taken to be zero.
	We define the generating series of TR descendents:
	\beq\label{eqn:generating-series-m-KP-F-A-TR}
	Z({\bf p};\hbar)=\exp\bigg(\sum_{g\geq 0, n\geq 0}\hbar^{2g-2+n}\sum_{1\leq i_1,\cdots,i_n\leq m \atop k_1,\cdots,k_n\geq 1}
	\<\alpha^{i_1}_{k_1},\cdots,\alpha^{i_n}_{k_n}\>_{g,n}^{\Lambda} \frac{p^{i_1}_{k_1}\cdots p^{i_n}_{k_n}}{n!\cdot k_1\cdots k_n}\bigg).
	\eeq
\end{definition}
We note here the choices of local coordinates $\lambda_i$ are not unique, and various options will lead to different invariants and the generating series.

Observe that the TR descendents and the TR ancestors come from the coefficients of $\omega_{g,n}$ by expanding under bases $\{d\zeta^{\bar\beta}_k\}$ and $\{d\lambda_i^{-k}\}$, respectively,
so the relation of these two kinds of correlators can be easily read off from the relation of these two bases.
This relation can be interpreted as the following formula of the generating series:
\beq\label{eqn:DKP-ACohFT}
Z({\bf p};\hbar)=e^{\frac{1}{2}Q({\bf p},{\bf p})}\cdot\cA(\hbar\cdot {\bf s(p)};\hbar),
\eeq
where $Q({\bf p},{\bf p})$ is defined by
\beq\label{eqn:def-Q(p,p)}
Q({\bf p},{\bf p})=\sum_{i,j=1,\cdots,m; \, k,l\geq 1}
\<\alpha^i_k,\alpha^j_{l}\>_{0,2}^{\Lambda}\frac{p^i_k}{k}\frac{p^j_{l}}{l},
\eeq
and ${\bf s(p)}$ is the coordinate transformation determined by the local expansion of $d\zeta_k^{\bar\beta}(z)$ at boundary points:
\beq\label{eqn:s-p}
d\zeta_n^{\bar\beta}(z)\big|_{b_i} = \sum_{k\geq a_i(n,\bar\beta)} c^{n,\bar\beta}_{k,i} \, d\lambda_i^{-k},\quad i=1,\cdots,m,
\qquad\Longrightarrow \qquad
s_n^{\bar\beta}=\sum_{i=1}^{m}\sum_{k\geq a_i(n,\bar\beta)}c^{n,\bar\beta}_{k,i}\, p^i_k.
\eeq
Here $a_i(n,\bar\beta)$ are some positive integers depending on $i$, $n$, and $\beta$.
Since $x$ has pole at the boundary point, we have $a_i(n+1,\bar\beta)>a_i(n,\bar\beta)$.

In this paper, our main interest is in the case with only one boundary.
We select  a local coordinate $\lambda$  at the boundary point $b$ that satisfies $\lambda(b)=\infty$, we denote the invariant $\<-\>_{g,n}^{\Lambda}$ by $\<-\>_{g,n}^{\lambda}$, and we drop the superscript $i$ in $\alpha_k^i$ and $p_k^i$.
By the structure of $\omega_{g,n}$ (equation~\eqref{eqn:str-omegagn}) with $2g-2+n>0$ and by Lemma~\ref{lem:res-dxi}, for the case having only one boundary point, we have for any $k\geq 0$,
\beq\label{eqn:red-omegagn}
\mathop{\Res}_{z_i=b}x(z_i)^{k}\omega_{g,n}(z_1,\cdots,z_n)=0,\qquad i=1,\cdots,n.
\eeq
For $(g,n)=(2,0)$, one may expect the following equations: for any $k\geq 0$,
\beq\label{eqn:red-omega02}
\mathop{\Res}_{z_i=b}x(z_i)^{k} \bigg(\omega_{0,2}(z_1,z_2)-\frac{d\lambda(z_1)d\lambda(z_2)}{(\lambda(z_1)-\lambda(z_2))^2}\bigg)=0,\qquad i=1,2.
\eeq
However, this is not always true for arbitrary local coordinate $\lambda$, we have the following result:
\begin{lemma}\label{lem:red-Z}
The equation~\eqref{eqn:red-omega02} holds if and only if the local coordinate $\lambda$ is selected such that
the local expansion of $x(z)$ near the boundary is a polynomial of $\lambda$.
\end{lemma}
\begin{proof}
By the symmetry of $\omega_{0,2}$ and $\frac{d\lambda(z_1)d\lambda(z_2)}{(\lambda(z_1)-\lambda(z_2))^2}$ on two variables, we just need to consider $i=1$ case of equation~\eqref{eqn:red-omega02}.
Notice that $x(z_1)^{k}\omega_{0,2}(z_1,z_2)$ has only poles at $z_1=b$ or $z_1=z_2$, we have
\beq\label{eqn:red-omega02-1}
\mathop{\Res}_{z_1=b}x(z_1)^{k} \omega_{0,2}(z_1,z_2)=-\mathop{\Res}_{z_1=z_2}x(z_1)^{k} \omega_{0,2}(z_1,z_2)
=-dx(z_2)^{k},
\eeq
where we have used the property of Bergman kernel: $\mathop{\Res}_{z_1=z_2}f(z_1)\omega_{0,2}(z_1,z_2)=df(z_2)$.

Consider the Laurent expansion of $x(z_1)^{k}$ near the boundary in variable $\lambda(z_1)^{-1}$, we denote by $[x(z_1)^{k}]_{+}$ the non-negative part of the expansion, since the $2$-differential $\frac{d\lambda(z_1)d\lambda(z_2)}{(\lambda(z_1)-\lambda(z_2))^2}$ is regular at $z_1=b$, we have
\beq\label{eqn:red-omega02-2}
\mathop{\Res}_{z_1=b}x(z_1)^{k}\frac{d\lambda(z_1)d\lambda(z_2)}{(\lambda(z_1)-\lambda(z_2))^2}
=\sum_{n\geq 1} \mathop{\Res}_{z_1=b}x(z_1)^{k}\lambda(z_1)^{-n-1}d\lambda(z_1)d\lambda(z_2)^n =-d[x(z_2)^{k}]_{+}.
\eeq
By comparing equation~\eqref{eqn:red-omega02-1} with equation~\eqref{eqn:red-omega02-2}, the Lemma is proved.
\end{proof}
\begin{corollary}\label{cor:red-Z-x}
For the spectral curve with only one boundary point, let $\lambda$ be the local coordinate near the boundary point such that $x(z)\in\mathbb C[\lambda]$,
then the generating series $Z({\bf p};\hbar)$ satisfies the following equations:
\beq\label{eqn:red-Z}
\widehat{x(z)^k}(Z({\bf p};\hbar))=0,\qquad  k\geq 1,
\eeq
where the quantization procedure ``\,  $\widehat{\, }$ " is defined by
$\widehat{\lambda^{i}}:=i\pd_{p_i}$, $i\geq 1$.
\end{corollary}
\begin{proof}
Suppose $x(z)^k=\sum_{i=0}^{kr}c_{k,i}\lambda^i$, then equations~\eqref{eqn:red-omegagn} and \eqref{eqn:red-omega02} are equivalent to the following equations for correlators:
$$\textstyle
\sum_{i=1}^{kr}c_{k,i}\<\alpha_i,-\>^{\lambda}_{g,n}=0.
$$
This is exactly the correlator version of equation~\eqref{eqn:red-Z}.
\end{proof}

\subsection{Geometric descendents from topological recursion}\label{sec:TR-des-Geo}
In \S \ref{sec:descedentforCohFT}, we introduced descendent invariants for $S$- and $\dvac$-calibrated CohFTs.
Now we explain how TR descendents are related with geometric descendents.

\begin{lemma}\label{lem:path-class}
Given an admissible path $\gamma$ with respect to $e^{-x(z)/\givz}$, there is a class $\Phi(\gamma,-\givz)$ such that
$$
\eta(\bar e_\beta,S(-\givz)\Phi(\gamma,-\givz))=\givz\int_{\gamma}e^{-x(z)/\givz}d\zeta_0^{\bar\beta}(z).
$$
\end{lemma}
\begin{proof}
We first notice that near each boundary $b_i$, $i=1,\cdots,m$, there is a local coordinate $\lambda_i$ on $\Sigma$ such that $x(z)=\lambda_i^{r_i}$ for some integers $r_i\in \mathbb Z_{>0}$.
Moreover, near the boundary points, the $1$-forms $d\zeta_{0}^{\bar\beta}$ is regular and can be expanded by using local coordinates $\lambda_i$.
We fix a constant $M\in\mathbb R$ such that in each of the preimages of $\{p\in\mathbb P^1: |p|>M\}$ (near the boundary $b_i$) under the map $x$, $d\zeta_{0}^{\bar\beta}(z)$ has a local expansion in local coordinate $\lambda_i$.

Let $\mathcal P$ be the set consists of paths of the following form:
$$
\gamma=\gamma_{-}+\gamma_M+\gamma_{+},
$$
where $\gamma_{\pm}$ is a preimage of $\pm\givz\cdot [M,\infty)$ under the map $x$, and $\gamma_M$ is a path (avoiding boundaries) that connects $\gamma_{-}$ and $\gamma_{+}$. Clearly, all of such paths are admissible.
Given a path $\gamma \in \mathcal P$, we consider
$$
\int_{\gamma}e^{-x(z)/\givz}d\zeta_0^{\bar\beta}(z)
=\int_{\gamma_{-}}e^{-x(z)/\givz}d\zeta_0^{\bar\beta}(z)
+\int_{\gamma_{M}}e^{-x(z)/\givz}d\zeta_0^{\bar\beta}(z)
+\int_{\gamma_{+}}e^{-x(z)/\givz}d\zeta_0^{\bar\beta}(z).
$$
Notice that $x(z)=\lambda_i^{r_i}$ near the boundary point $b_i$,
the first and third parts on the right-hand side of above equation give elements in $\mathbb C[[\givz^{-\frac{1}{r}}]]$, where $r=\prod r_i$.
For the second part, the integral is finite and it takes value in $\mathbb C[[\givz^{-1}]]\subset \mathbb C[[\givz^{-\frac{1}{r}}]]$.
Therefore, the integral
$\int_{\gamma}e^{-x(z)/\givz}d\zeta_0^{\bar\beta}(z)$ is a well-defined power series of
$\givz^{-\frac{1}{r}}$.
For such $\gamma\in \mathcal P$, we define
$$
\Phi(\gamma,-\givz):=\givz \cdot S^{*}(\givz)\cdot \sum_{\beta}\bar e_{\beta}\cdot \int_{\gamma}e^{-x(z)/\givz}d\zeta_0^{\bar\beta}(z).
$$

For an arbitrary admissible path $\gamma$, it is easy to see that there exists a path $\gamma'\in\mathcal P$ such that
$[\gamma'-\gamma]=0\in H_{1}(\Sigma,\mathbb C)$.
For such case, we define $\Phi(\gamma,-\givz)=\Phi(\gamma',-\givz)$. It is well defined since for any two paths $\gamma,\gamma'$ satisfying $[\gamma'-\gamma]=0\in H_{1}(\Sigma,\mathbb C)$, one has
$$
\int_{\gamma}e^{-x(z)/\givz}d\zeta_0^{\bar\beta}(z)
=\int_{\gamma'}e^{-x(z)/\givz}d\zeta_0^{\bar\beta}(z).
$$
The Lemma is proved.
\end{proof}
\begin{theorem}\label{thm:TR-des-(g,n)}
For $2g-2+n>0$ and admissible paths $\gamma_i$ associated with $e^{-x(z_i)/\givz_i}$, $i=1,\cdots,n$, we have
\beq\label{eqn:des-EO}
\int_{\gamma_1}\!\cdots\!\int_{\gamma_n}e^{-\sum_{i=1}^{n}x(z_i)/\givz_i}\omega_{g,n}
=\Big\<\frac{\Phi(\gamma_1,-\givz_1)}{\givz_1+\psi_1},\cdots,\frac{\Phi(\gamma_n,-\givz_n)}{\givz_n+\psi_n}\Big\>_{g,n}.
\eeq
For $(g,n)=(0,2)$ and admissible paths $\gamma_i$ associated with $e^{-x(z_i)/\givz_i}$, $i=1,2$, we have
\beq\label{eqn:des-EO-(0,2)}
\int_{\gamma_1}\!  \int_{\gamma_2}e^{-\sum_{i=1}^{2}x(z_i)/\givz_i}\omega_{0,2}
=\frac{\eta(\Phi(\gamma_1,-\givz_1),\Phi(\gamma_2,-\givz_2))}{-\givz_1-\givz_2}+\Big\<\frac{{\Phi}(\gamma_1,-\givz_1)}{\givz_1+\psi_1},\frac{{\Phi}(\gamma_2,-\givz_2)}{\givz_2+\psi_2}\Big\>_{0,2}.
\eeq
\end{theorem}
\begin{proof}
By Lemma~\ref{lem:path-class} and the formula of integration by parts, we have
$$\textstyle
\eta(\bar e_\beta,S(-\givz)\Phi(\gamma,-\givz))
=(-1)^k\givz^{k+1}\int_{\gamma}e^{-x(z)/\givz}d\zeta_k^{\bar\beta}(z).
$$
The equation~\eqref{eqn:des-EO} follows immediately from the structure of $\omega_{g,n}$~\eqref{eqn:wgn-zeta}
and the correspondence between ancestor invariants with descendent invariants~\eqref{eqn:des-anc-S}.

For $(g,n)=(0,2)$, by~\cite[Lemma 6.9]{FLZ20} (see equation~\eqref{eqn:lemFLZ} below) and the formula of integration by parts, we have
$$\textstyle
\big(\frac{1}{\givz_1}+\frac{1}{\givz_2}\big)\int_{\gamma_1}\int_{\gamma_2}e^{-\sum_{i=1}^{2}x(z_i)/\givz_i}\omega_{0,2}
=-\sum_\beta\int_{\gamma_1}\int_{\gamma_2}e^{-\sum_{i=1}^{2}x(z_i)/\givz_i}d\zeta^{\bar\beta}(z_1)d\zeta^{\bar\beta}(z_2).
$$
Divided by $(\frac{1}{\givz_1}+\frac{1}{\givz_2})$ on above equation, the right-hand side gives
$$\textstyle
-\frac{1}{\givz_1+\givz_2}\cdot\sum_{\beta}\eta(\bar e_\beta,S(-\givz_1)\Phi(\gamma_1,-\givz_1))\eta(\bar e_\beta,S(-\givz_2)\Phi(\gamma_2,-\givz_2))
$$
which equals  the right-hand side of equation~\eqref{eqn:des-EO-(0,2)} by equation~\eqref{eqn:W-S}.
\end{proof}

Now we consider the generating series version of the relation of TR descendents and geometric descendents.
By using the $S$-matrix $S(\givz)$, we define another basis near the boundary points
\beq\label{eqn:chi-zeta}
d \chi_l^j:=\sum_{k\geq 0;\, i=1,\cdots,N}(-1)^k(S^{*}_k)^j_i d \zeta_{k+l}^{i}.
\eeq
We note that $d\chi_l^j$ is not necessarily a global $1$-form on $\Sigma$, but always  has a well-defined expansion near the boundary points.
Conversely, near the boundary points, we have
\beq\label{eqn:zeta-chi}
d\zeta_k^{i}=\sum_{l\geq0;\, j=1,\cdots,N}(S_l)^i_j d\chi_{l+k}^j.
\eeq
By taking expansion of $d\chi_l^j$ near the boundary points, we define ${\bf t}={\bf t(p)}$ as follows:
\beq\label{eqn:t=tp}
d\chi_l^{j}\big|_{b_i}=\sum_{k\geq a'_{i}(l,j)} {c'}^{l,j}_{k,i}\, d\lambda_i^{-k},
\quad i=1,\cdots,m,
\qquad \Longrightarrow\qquad
t^{j}_l({\bf p})=\sum_{i=1}^{m}\sum_{k\geq a'_{i}(l,j)} {c'}^{l,j}_{k,i}\, p^i_k.
\eeq
Here $a'_i(l,j)$ are some positive integers depending on $i$, $l$, and $j$.
Similarly as the ancestor version, since $x$ has pole at the boundary points, we have $a'_i(l+1,j)>a'_i(l,j)$.

\begin{theorem}\label{thm:D-Z}
The total descendent potential $\cD({\bf t};\hbar)$ is related with $Z({\bf p};\hbar)$ by the following equation
\beq\label{eqn:D-Z}
e^{\frac{1}{2}\QD({\bf p},{\bf p})}\cdot \cD(\hbar\cdot{\bf t};\hbar)|_{\bf t=t(p)}
= e^{\frac{1}{\hbar}J_{-}({\bf t})|_{\bf t=t(p)}}\cdot Z({\bf p};\hbar),
\eeq
where $\QD({\bf p},{\bf p})=Q({\bf p},{\bf p})-W({\bf t},{\bf t})|_{\bf t=t(p)}$
and the coordinates transformation ${\bf t}={\bf t(p)}$ is defined by equation~\eqref{eqn:t=tp}.
\end{theorem}
\begin{proof}
Notice the relation of $d\zeta_k^{i}$ and $d\chi_{l}^j$ (equation~\eqref{eqn:zeta-chi}), by taking expansion near the boundary we have
$$
\cA({\bf s(t)};\hbar)|_{\bf t=t(p)}=\cA({\bf s}({\bf p});\hbar),
$$
where ${\bf s(t)}$ is given by ${\bf s}(\givz)=[S(\givz){\bf t}(\givz)]_{+}$, ${\bf t(p)}$ is given by~\eqref{eqn:t=tp}, and ${\bf s(p)}$ is given by~\eqref{eqn:s-p}, respectively.
The equation~\eqref{eqn:D-Z} follows from the relation of $\cD$ and $\cA$ (see equation~\eqref{eq:descendent-ancestor}, take $\ttau=0$) and the relation of $Z$ and $\cA$ (see equation~\eqref{eqn:DKP-ACohFT}).
\end{proof}
\begin{proposition}\label{prop:Q=W}
The term $\QD({\bf p},{\bf p})$ in equation~\eqref{eqn:D-Z} vanishes if and only if the following equations holds: for each pair $(i_1,i_2)\in [m]\times [m]$, near $z_1=b_{i_1}$ and $z_2=b_{i_2}$,
\beq\label{eqn:vanif2}
\sum_{i,j=1}^{N}d\chi^i(z_1)\eta_{ij}d\chi^j(z_2)
=-\bigg(d_1\circ\frac{1}{dx(z_1)}+d_2\circ\frac{1}{dx(z_2)}\bigg)
\bigg(\frac{\delta_{i_1,i_2}d\lambda_{i_1}(z_1)d\lambda_{i_2}(z_2)}{(\lambda_{i_1}(z_1)-\lambda_{i_2}(z_2))^2}\bigg).
\eeq
\end{proposition}
\begin{proof}
By the definitions of $W({\bf t,t})$, $Q({\bf p,p})$, the coordinate transformation ${\bf t(p)}$, it is clear that $W({\bf t},{\bf t})|_{\bf t=t(p)}=Q({\bf p},{\bf p})$ if and only if the following equations hold: for each pair $(i_1,i_2)\in [m]\times [m]$,  near $z_1=b_{i_1}$ and $z_2=b_{i_2}$,
\beq\label{eqn:W-Q}
\sum_{k,l\geq0;\, i,j\in [N]}\eta(\phi_i,W_{k,l}\phi_j)d\chi_k^i(z_1)d\chi_l^j(z_2)
=B(z_1,z_2)-\frac{\delta_{i_1,i_2}d\lambda_{i_1}(z_1)d\lambda_{i_2}(z_2)}{(\lambda_{i_1}(z_1)-\lambda_{i_2}(z_2))^2}.
\eeq
Notice that both sides of these equations are regular,
these equations hold if and only if they hold after the action of $-d\circ\frac{1}{dx(z_1)}-d\circ\frac{1}{dx(z_2)}$ on both sides.
On one hand, by using equation~\eqref{eqn:W-S} and equation~\eqref{eqn:zeta-chi},
after the action, the left-hand side of equation~\eqref{eqn:W-Q} gives
\beq\label{eqn:act-W}
\sum_{i,j}d\zeta^{i}(z_1)\eta_{ij}d\zeta^{j}(z_2)-\sum_{i,j}d\chi^i(z_1) \eta_{ij}d\chi^j(z_2).
\eeq
On another hand, by~\cite[Lemma 6.9]{FLZ20}, we have
\beq\label{eqn:lemFLZ}
-\bigg(d_1\circ\frac{1}{dx(z_1)}+d_2\circ\frac{1}{dx(z_2)}\bigg)\big(B(z_1,z_2)\big)
=\sum_{i,j}d\zeta^{i}(z_1)\eta_{ij}d\zeta^{j}(z_2).
\eeq
The Proposition follows from equations~\eqref{eqn:W-Q}, \eqref{eqn:act-W} and \eqref{eqn:lemFLZ}.
\end{proof}

\section{Non-perturbative generating series}\label{sec:NP-TR}
In this section, we first introduce the non-perturbative multi-differentials and the generating series $Z^{\rm NP}_{\mu,\nu}({\bf p};\hbar;w)$ for the spectral curve $\cC=(\Sigma,x,y)$, where $\frak g(\Sigma)\geq 1$.
With the explicit definition of $Z^{\rm NP}_{\mu,\nu}$, the integrability conjecture (Conjecture~\ref{conj:EO-m-KP}) stated in the Introduction becomes fully stated.
We then introduce the non-perturbative geometric potentials $\cA_{\mu,\nu,\tau}^{\rm NP}({\bf s};\hbar;w)$ and $\cD_{\mu,\nu}^{\rm NP}({\bf t};\hbar;w)$, and establish the non-perturbative version of TR-Geo correspondence.
As a result, the generalization of the Witten conjecture (Conjecture~\ref{conj:generalized-witten}) is also fully stated.

\subsection{Non-perturbative multi-differentials}
The non-perturbative contributions come from the non-vanishing $B$-cycle periods of the multi-differentials $\omega_{g,n}$.
We introduce $\omega_{g,n_1}^{(n_2)}$ which is an $n_1$-differential valued $n_2$-tensor whose elements are given by
\beq
\omega_{g,n_1,j_1,\cdots,j_{n_2}}^{(n_2)}(z_{[n_1]}):=\oint_{z_{n_1+1}\in \cyB_{j_1}}\cdots \oint_{z_{n_1+n_2}\in \cyB_{j_{n_2}}} \omega_{g,n_1+n_2}(z_1,\cdots,z_{n_1+n_2}).
\eeq
\begin{remark}
	One can also define $\omega_{g,n_1}^{(n_2)}$ by taking derivatives with respect to the filling fractions. The equivalence of these two definitions comes from the variational formula of the topological recursion. See~\cite[\S 5.3, Theorem 5.1]{EO07} for details.
\end{remark}

According to the structure of $\omega_{g,n}$ (see equation~\eqref{eqn:wgn-zeta}), for $2g-2+n_1+n_2>0$, we have the following structure of $\omega_{g,n_1}^{(n_2)}$:
$$
\omega_{g,n_1,j_1,\cdots,j_{n_2}}^{(n_2)}(z_{[n_1]})=\sum_{\substack{k_1,\cdots,k_{n_1}\geqslant 0\\ \beta_1,\cdots,\beta_{n_2}\in [N]}}
\<\bar e_{\beta_1}\bar\psi^{k_1},\cdots,\bar e_{\beta_{n_1}}\bar\psi^{k_{n_1}}\>^{(n_2),\, \Omega}_{g,n_1,j_1,\cdots,j_{n_2}}
d\zeta^{\bar\beta_1}_{k_1}(z_1)\cdots d\zeta^{\bar\beta_{n_1}}_{k_{n_1}}(z_{n_1}).
$$
Similarly, by picking a choice of local coordinates $\Lambda=(\lambda_{1},\cdots,\lambda_{m})$ and taking expansions of $\omega_{g,n_1}^{(n_2)}$ near the boundary points $z_1=b_{i_1},\cdots,z_{n_1}=b_{i_{n_1}}$,
we have
\beq
\omega_{g,n_1,j_1,\cdots,j_{n_2}}^{(n_2)}(z_{[n_1]})=\sum_{k_1,\cdots, k_{n_1}\geq 1}\<\alpha^{i_1}_{k_1},\cdots,\alpha^{i_{n_1}}_{k_{n_1}}\>_{g,n_1,j_1,\cdots, j_{n_2}}^{(n_2), \, \Lambda}
\cdot \frac{d\lambda_{i_1,1}^{-k_1}\cdots d\lambda_{i_{n_1},n_1}^{-k_{n_1}}}{k_1\cdots k_{n_1}},
\eeq
For the cases $\omega_{0,n_1}^{(n_2)}$ with $n_1+n_2=2$, the 2-differential $\omega_{0,2}^{(0)}=\omega_{0,2}$ is expanded as before (equation~\eqref{eqn:stable-omega-gn-boundary}),
the differential $\omega^{(1)}_{0,1,j}(z)$ near the boundary point $z=b_i$ has expansion
$$
\omega^{(1)}_{0,1,j}(z)=\oint_{z'\in\cyB_j}B(z,z')=2\pi {\bf i}du_j(z)
=\sum_{k\geq 1}\<\alpha_k^{i}\>_{0,1,j}^{(1)}\frac{d\lambda_{i}^{-k}}{k},
$$
and the $0$-differential $\omega_{0,0,i,j}^{(2)}$ is just a constant given by
$$
\omega_{0,0,i,j}^{(2)}=2\pi {\bf i}\oint_{z\in \cyB_j}du_i(z)=2\pi {\bf i}\, \tau_{ij}=\<\>_{0,0,i,j}^{(2)}.
$$

\subsection{Non-perturbative TR generating series}
Given the non-perturbative invariants, we use a theta-function $\theta=\theta[\cmn](w|\tau')$ to encode all possible $B$-cycle integrals and define the TR non-perturbative generating series as follows.
\begin{definition}
	Given the spectral curve $\cC=(\Sigma,x, y)$,
	the non-perturbative total ancestor potential $\cA_{\mu,\nu,\tau'}^{\rm NP}({\bf s};\hbar;w)$ is defined by
	$$
	\cA^{\rm NP}_{\mu,\nu,\tau'}({\bf s};\hbar;w):=
	\exp\bigg(\sum_{2g-2+n_1+n_2>0 } \hbar^{2g-2+n_2}\frac{\<{\bf s}(\bar\psi),\cdots,{\bf s}(\bar\psi)\>_{g,n_1}^{(n_2),\Omega} \cdot \nabla_w^{\otimes n_2}}{n_1!n_2!}\bigg)\theta[\cmn](w|\tau'),
	$$
	and the non-perturbative TR descendent potential $Z_{\mu,\nu,\tau'}^{\rm NP}({\bf p};\hbar;w)$ is defined by
	$$
	Z^{\rm NP}_{\mu,\nu,\tau'}({\bf p};\hbar;w):=
	\exp\bigg(\sum_{\substack{2g-2+n\geq0\\n_1+n_2=n}} \hbar^{2g-2+n}\frac{\<{\bf p}(\alpha),\cdots,{\bf p}(\alpha)\>_{g,n_1}^{(n_2), \Lambda} \cdot \nabla_w^{\otimes n_2}}{n_1!n_2!}\bigg)\theta[\cmn](w|\tau').
	$$
	Here ${\bf s}(\bar\psi)=\sum_{k\geq 0}\sum_{a=1}^{N}s_k^a\phi_a\bar\psi^k$, ${\bf p}(\alpha)=\sum_{k\geq 1}\sum_{i=1}^{m}\frac{1}{k}p_k^i \alpha_k^i $,
	 $\nabla_{w}=\frac{1}{2\pi {\bf i}}(\pd_{w_1},\cdots,\pd_{w_{\frak g}})$ and
	 $$\<-\>_{g,n_1}^{(n_2),*}\cdot \nabla_{w}^{\otimes n_2}
	 =\frac{1}{(2\pi {\bf i})^{n_2}}\sum_{j_1,\cdots,j_{n_2}=1}^{\frak g}\<-\>_{g,n_1,j_1,\cdots,j_{n_2}}^{(n_2),*}\pd_{w_{j_1}}\cdots\pd_{w_{j_{n_2}}},
	 $$
where $*=\Omega, \Lambda$.
\end{definition}
We note that we cannot take $\tau'=0$ in $\cA^{\rm NP}_{\mu,\nu,\tau'}({\bf s};\hbar;w)$ since this will lead to a divergent result in the theta-function.
However, we can take $\tau'=0$  in $Z^{\rm NP}_{\mu,\nu,\tau'}({\bf p};\hbar;w)$, because the operator $e^{\frac{1}{2}\omega_{0,0}^{(2)}\nabla_{w}^{\otimes 2}}$ (which comes from $\omega_{0,2}$ and thus is not contained on the ancestor side) shifts $\tau'$ in $\theta[\cmn](w|\tau')$ by $\tau$:
 $$
 e^{\frac{1}{2}\omega_{0,0}^{(2)}\nabla_{w}^{\otimes 2}}\big(\theta[\cmn](w|\tau')\big)=\theta[\cmn](w|\tau'+\tau),
 $$
 where we have used
 \beq\label{eqn:dw2-dtau}
 \frac{1}{2}\cdot 2\pi {\bf i}\, \tau\cdot \nabla_{w}^{\otimes 2}\big(\theta[\cmn](w|\tau')\big)
 =\sum_{i,j}\tau_{ij}\pd_{\tau'_{ij}}\big(\theta[\cmn](w|\tau')\big).
 \eeq
We see that we can take $\tau'=0$ in $Z^{\rm NP}_{\mu,\nu,\tau'}({\bf p};\hbar;w)$ without losing any information, so we define
$$
Z^{\rm NP}_{\mu,\nu}({\bf p};\hbar;w):=Z^{\rm NP}_{\mu,\nu,\tau'}({\bf p};\hbar;w)|_{\tau'=0}.
$$
Similar as the perturbative version, we have the following non-perturbative analogue of the formula~\eqref{eqn:DKP-ACohFT}:
\beq\label{eqn:DKP-ACohFT-NP}
Z^{\rm NP}_{\mu,\nu}({\bf p};\hbar;w):=e^{\frac{1}{2}Q({\bf p,p})}\cdot e^{2\pi {\bf i} \widehat{u(z)}\cdot \nabla_{w}}\big(\cA_{\mu,\nu,\tau}^{\rm NP}(\hbar\cdot {\bf s(p)};\hbar;w)\big),
\eeq
where $\widehat{u(z)}=\frac{1}{2\pi {\bf i}}\sum_{k,i}\<\alpha^i_{k}\>_{0,1}^{(1)}\frac{p^i_k}{k}$, ${\bf s(p)}$ is given by~\eqref{eqn:s-p}, and we have used equation~\eqref{eqn:dw2-dtau}.

\begin{remark}
	If we take $w=\frac{1}{2\pi {\bf i} \hbar}\oint_{\cyB}y(z)dx(z)-\frac{\tau}{2\pi {\bf i} \hbar}\oint_{\cyA} y(z)dx(z)$, and take ${\bf p=0}$ or $p_k=\lambda_1^{-k}-\lambda_2^{-k}$, then, up to some simple factors, our definition recovers Eynard--Mari\~{n}o's definition~\cite{EM11} of the non-perturbative partition function and Borot--Eynard's definition~\cite{BE12} of spinor kernel (or $(1|1)$-kernel in~\cite{BE15}), respectively.
\end{remark}

\begin{conjecture}[= Conjecture~\ref{conj:EO-m-KP}]\label{conj:NP-TR-KP}
The generating series $Z({\bf p};\hbar)$ is a tau-function of the $m$-component KP hierarchy when ${\frak g}(\Sigma)=0$, whereas $Z^{\rm NP}_{\mu,\nu}({\bf p};\hbar;w)$ is a tau-function of the $m$-component KP hierarchy when ${\frak g}(\Sigma)\geq 1$.
\end{conjecture}
\begin{remark}
In~\cite{BE12}, Borot and Eynard proposed several equivalent conjectures regarding the integrability of non-perturbative topological recursion. While their framework relies on spinor kernels, ours is formulated using generating series, which are better suited for applications in enumerative geometry. Despite the differences in presentation and underlying formalism, we anticipate that our conjecture is equivalent to theirs. For recent developments related to their approach, we direct the reader to~\cite{ABDKS23,ABDKS24}.
\end{remark}

For the cases with only one boundary point,
we choose the local coordinate $\lambda$ near the boundary point such that $x(z)\in \mathbb C[\lambda]$, then
$Z^{\rm NP}_{\mu,\nu}({\bf p};\hbar;w)$ satisfies equation~\eqref{eqn:red-Z} too:
\beq\label{eqn:red-ZNP}
\widehat{x(z)^n}(Z^{\rm NP}_{\mu,\nu}({\bf p};\hbar;w))=0,\qquad n\geq 1.
\eeq
This is because we always have
\beq\label{eqn:red-Z-NP}
\mathop{\Res}_{z\to b} x(z)^k du_j(z)=\frac{1}{2\pi {\bf i}}\oint_{z'\in\cyB_j}\mathop{\Res}_{z\to b}x(z)^k B(z,z')
=-\frac{1}{2\pi {\bf i}}\oint_{z'\in\cyB_j}dx(z')^k=0.
\eeq
\begin{remark}
When setting  $\hbar=0$ in $Z^{\rm NP}_{\mu,\nu}({\bf p};\hbar;w)$, we obtain
$$
	Z^{\rm NP}_{\mu,\nu}({\bf p};\hbar;w)|_{\hbar=0}=e^{\frac{1}{2}Q({\bf p,p})}\cdot \theta[\cmn](w+\widehat{u(z)}|\tau).
$$
Denoting the right-hand side of this equation by $\mathcal T({\bf p})$, we recognize it as a classical result, originally due to Krichever~\cite{Kri77} (see also~\cite{Dub81, Mum}), that  $\mathcal T({\bf p})$
 serves as a tau-function of the KP hierarchy.
This confirms Conjecture~\ref{conj:NP-TR-KP} at the ``classical" level,
while our result extends this classical framework to a ``quantum" version

Conversely, Shiota~\cite{Shi86} proved that   $\mathcal T({\bf p})$, constructed using the theta function on a principally polarized abelian variety $X=\mathbb C^{\frak g}/(\mathbb Z^{\frak g}+\tau\mathbb Z^{\frak g})$, is a tau-function of the KP hierarchy if and only if $X$ is isomorphic to the Jacobian variety of a complete smooth complex curve of genus $\frak g$ (see also~\cite{Mul84}).
In other words, the KP integrability of $\mathcal T({\bf p})$ characterizes Jacobian varieties among principally polarized abelian varieties.
This leads to an interesting question: what is the geometric meaning of the integrability of  $Z^{\rm NP}_{\mu,\nu}({\bf p};\hbar;w)$?
\end{remark}
\subsection{Non-perturbative geometric generating series}
Now we turn to the geometry side to define the non-perturbative geometric generating series.

We first consider the ancestor side.
Introduce a sequence of vector fields $\sv_j$, $j=1,\cdots,\frak g$:
\beq\label{def:varphi}
\sv_j=\sum_{\beta}\bar e_{\beta}\cdot \oint_{z\in \cyB_j}d\zeta^{\bar\beta}(z)
=-2\pi {\bf i}\sum_{\beta}\frac{du_j}{dy}\Big|_{z=z^\beta}\cdot e_{\beta}.
\eeq
Notice that for $k\geq 1$, $\oint_{z\in \cyB_i}d\zeta_k^{\bar\beta}(z)=0$,
by the structure of $\omega_{g,n}$ (equation~\eqref{eqn:wgn-zeta}), we have for $2g-2+n_1+n_2>0$:
$$
\omega_{g,n_1,j_1,\cdots,j_{n_2}}^{(n_2)}=\sum_{\substack{k_1,\cdots,k_{n_1}\geqslant 0\\ \beta_1,\cdots,\beta_{n_1}\in [N]}}
\<\bar e_{\beta_1}\bar\psi^{k_1},\cdots,\bar e_{\beta_{n_1}}\bar\psi^{k_n},\sv_{j_1},\cdots,\sv_{j_m}\>^{\Omega}_{g,n_1+n_2}
d\zeta^{\bar\beta_1}_{k_1}(z_1)\cdots d\zeta^{\bar\beta_{n_1}}_{k_{n_1}}(z_{n_1}).
$$
By viewing the vector fields $\sv_j=\sum_{a}\varphi^a_j\phi_a$ as $\sum_{a}\varphi^a_j\pd_{s_0^a}$, $j=1,\cdots,\frak g$, potential $\cA_{\mu,\nu,\tau'}^{\rm NP}({\bf s};\hbar;w)$ can be expressed as follows:
\beq\label{eqn:NPA-formula}
\cA_{\mu,\nu,\tau'}^{\rm NP}({\bf s};\hbar;w)=e^{\hbar\, \sv\cdot \nabla_{w}}\big(\cA({\bf s};\hbar)\cdot\theta[\cmn](w|\tau')\big).
\eeq

Then we turn to the descendent side. Inspired by the ancestor side, we introduce the following definition:

\begin{definition}
	The non-perturbative total descendent potential $\cD_{\mu,\nu,\tau'}^{\rm NP}$ is defined as follows:
	$$
	\cD_{\mu,\nu,\tau'}^{\rm NP}({\bf t};\hbar;w) :=e^{\hbar\, \sv\cdot \nabla_{w}}\big(\cD({\bf t};\hbar)\cdot \theta[\cmn](w|\tau')\big),
	$$
where $\sv_j=\sum_{a}\varphi^a_j\phi_a$ is viewed as $\sum_{a}\varphi^a_j\pd_{t_0^a}$.
\end{definition}
By the relation of $\cD$ and $\cA$ (see \S \ref{sec:descedentforCohFT}), we have the following formula
\beq\label{eqn:descendent-ancestor-NP}
\cD_{\mu,\nu,\tau'}^{\rm NP}({\bf t};\hbar;w) :=
e^{F_{\rm un}({\bf t};\hbar)+\frac{1}{\hbar}\<\sv\>_{0,1}\cdot \nabla_{w}+\frac{1}{\hbar}W({\bf t},\sv\cdot \nabla_w)}\big(\cA^{\rm NP}_{\mu,\nu,\tilde \tau'}({\bf s(t)};\hbar;w)\big),
\eeq
where
$W({\bf t},\sv\cdot \nabla_w)=\sum_{k}\eta(W_{0,k}t_k,\sv\cdot \nabla_w)$
and $\tilde\tau'_{ij}:=\tau'_{ij}+\frac{1}{2\pi {\bf i}}\<\sv_i,\sv_j\>^{\cD}_{0,2}$,
and the coordinate transformation ${\bf s(t)}$ is given by ${\bf s}(\givz)=[S(\givz){\bf t}(\givz)]_{+}$.
Here we have used equation~\eqref{eqn:dw2-dtau}.
Inspired by the TR non-perturbative descendent potential, we define
$$
\cD_{\mu,\nu}^{\rm NP}({\bf t};\hbar;w):=\cD_{\mu,\nu,\tau'}^{\rm NP}({\bf t};\hbar;w)\big|_{\tau'_{ij}=\tau_{ij}-\frac{1}{2\pi {\bf i}}\<\sv_i,\sv_j\>^{\cD}_{0,2}}.
$$

We have the following Proposition generalizing Theorem~\ref{thm:D-Z} to non-perturbative case.
\begin{proposition}\label{prop:D-Z-NP}
	The function $\cD^{\rm NP}_{\mu,\nu}({\bf t};\hbar;w)$ is related with $Z^{\rm NP}_{\mu,\nu}({\bf p};\hbar;w)$ by the following equation
	$$
	e^{\frac{1}{2}\QD({\bf p},{\bf p})}\cdot \cD_{\mu,\nu}^{\rm NP}(\hbar\cdot{\bf t};\hbar)|_{\bf t=t(p)}
	= e^{\frac{1}{\hbar}J_{-}({\bf t})|_{\bf t=t(p)}+\frac{1}{\hbar}\<\sv\>_{0,1}\cdot \nabla_{w}} \cdot Z^{\rm NP}_{\mu,\nu}({\bf p};\hbar;w),
	$$
	where $\QD({\bf p},{\bf p})=Q({\bf p},{\bf p})-W({\bf t},{\bf t})|_{\bf t=t(p)}$,
	and the coordinates transformation ${\bf t}={\bf t(p)}$ is defined by equation~\eqref{eqn:t=tp}.
\end{proposition}
\begin{proof}
	By equations~\eqref{eqn:DKP-ACohFT-NP} and~\eqref{eqn:descendent-ancestor-NP}, we just need to prove following equations:
	\beq\label{eqn:Wtv-u}
	W({\bf t(p)},\sv_j)=2\pi {\bf i}\widehat{u_j(z)},\qquad j=1,\cdots,\frak g.
	\eeq
	These equations are equivalent to prove that near each boundary point $b_i$, $i=1,\cdots,m$,
	$$\textstyle
	\sum_{k,a}\eta(W_{0,k}\phi_a , \sv_j)\cdot \chi^a_k(z)=2\pi {\bf i} (u_j(z)-u_j(b_i)),\qquad j=1,\cdots,\frak g.
	$$
	We consider the action of $-\frac{d}{dx(z)}$ on both sides of these equations. For the left-hand side
	$$\textstyle
	\sum_{k,a}\eta(W_{0,k}\phi_a , \sv_j)\cdot \big(-\frac{d}{dx(z)}\big)\chi^a_k(z)
	=\sum_{k,a}\eta(W_{0,k}\phi_a , \sv_j)\cdot \chi^a_{k+1}(z),
	$$
and this gives $\sum_a\eta(\phi_a,\sv_j)\cdot (\zeta^a_0(z)-\chi^a_0(z))$
since $W_{0,k}=S_{k+1}$.
For the right-hand side,
	$$\textstyle
	-2\pi {\bf i} \frac{d u_j(z)}{dx(z)}=-\oint_{z'\in \cyB_j}\frac{B(z,z')}{dx(z)}=\sum_a\eta(\phi_a,\sv_j) (\zeta_0^a(z)-\zeta_0^a(b_i)).
	$$
Notice that for a regular local form $df(z)$ near the boundary point $b_i$, by taking local coordinate $\lambda_i$ such that $x(z)=\lambda_i^{r_i}$ for some integer $r_i\geq 1$, one has
$$
\lim_{z\to b_i}x(z)\frac{df(z)}{dx(z)} =\lim_{\lambda_i\to \infty} \frac{-1}{r_i\lambda_i}\frac{df(z)}{d\lambda_i^{-1}}= 0.
$$
In particular, $\lim_{z\to b_i}x(z)\frac{du_j(z)}{dx(z)}=0$ and $\lim_{z\to b_i}x(z)\cdot\chi_{k+1}^a(z)=0$ for $k\geq 1$.
Therefore,
$$
0=\lim_{z\to b_i}x(z)\cdot\sum_a\eta(\phi_a,\sv_j) (\zeta_0^a(z)-\zeta_0^a(b_i))
=\lim_{z\to b_i}x(z)\cdot\sum_a\eta(\phi_a,\sv_j) (\chi_0^a(z)-\chi_0^a(b_i)),
$$
where we have used $\zeta_0^a(b_i)=\chi_0^a(b_i)$.
Since $\int_{b_i}^{z}(\chi_0^a(z)-\chi_0^a(b_i))dx$ is either zero or singular,
one must has $\sum_a\eta(\phi_a,\sv_j)(\chi_0^a(z)-\chi_0^a(b_i))=0$.
This gives us
$$\textstyle
\sum_a\eta(\phi_a,\sv_j)\cdot (\zeta^a_0(z)-\chi^a_0(z)) =\sum_a\eta(\phi_a,\sv_j)\cdot (\zeta_0^a(z)-\zeta_0^a(b_i)),
$$
and equation~\eqref{eqn:Wtv-u} is proved.
\end{proof}

\begin{conjecture}[= Conjecture~\ref{conj:generalized-witten}]
The generating series $e^{\frac{1}{2}\QD({\bf p},{\bf p})}\cdot \cD(\hbar\cdot{\bf t};\hbar)|_{\bf t=t(p)}$ is a tau-function of the $m$-component KP hierarchy when ${\frak g}(\Sigma)=0$, while $e^{\frac{1}{2}\QD({\bf p},{\bf p})}\cdot \cD_{\mu,\nu}^{\rm NP}(\hbar\cdot{\bf t};\hbar)|_{\bf t=t(p)}$ is a tau-function of the $m$-component KP hierarchy when ${\frak g}(\Sigma)\geq 1$.
Furthermore, the generating series in both cases satisfy  certain reduction structures.
\end{conjecture}

\section{Polynomial-reduced KP integrability for topological recursion}
\label{sec:TRKPproof}
In this section we prove Theorem~\ref{KP-TR} by using the Hirota quadratic equation (HQE) formulation of the polynomial-reduced KP hierarchy.

\begin{definition}\label{def:lambda-reduction-KP}
A tau-function $Z({\bf p})$ of the KP hierarchy is called a tau-function of the {\em polynomial-reduced} KP hierarchy
if there is a polynomial $x_{\lambda}\in\mathbb C[\lambda]$ such that
\beq\label{eqn:red-Z-x}
\widehat{x_\lambda^k}(Z({\bf p}))=0,\qquad k\geq 1,
\eeq
where the quantization is defined by $\widehat{\lambda^n}:=n\pd_{p_n}$, $n\geq 0$.

\end{definition}
We note that when $x(z)=c\cdot\lambda^{r}$ is a monomial for some non-zero constant $c$ and integer $r\geq 2$, the corresponding polynomial-reduced KP hierarchy is exactly the $r$KdV hierarchy, also called the Gelfand--Dickey hierarchy or $W_{r}$ hierarchy.

\subsection{HQE formulation of KP hierarchy and its polynomial reduction}

We begin with the HQE formulation of the KP hierarchy. We refer the reader to~\cite{BBT03,Kac90,MJD00} for more details.
A function $Z({\bf p})$ of infinitely many variables $p_1,p_2,\cdots$ is a tau-function of the KP hierarchy (with KP times $\{\frac{p_k}{k}\}_{k=1,2,\cdots}$) if and only if it satisfies the following HQE \cite{Kac90, MJD00}:
\beq\label{eqn:HQE-KP}
\mathop{\Res}_{\lambda=\infty} \, (\Gamma^{+}Z)({\bf p})\, (\Gamma^{-}Z)({\bf p'}) \, d\lambda=0,
\eeq
where $\Gamma^{\pm}$ are the {\it vertex operators} defined by
$$
\Gamma^{\pm}=\exp\bigg(\pm\sum_{k\geq 1}\frac{\lambda^k}{k}p_{k}\bigg) \exp\bigg(\mp\sum_{k\geq 1}\frac{1}{\lambda^{k}}\pd_{p_{k}}\bigg).
$$
The equation is interpreted in the following way.
We take the change of variables $p_k=q_k+q'_k$ and $p'_k=q_k-q'_k$, then we have
$\pd_{q_k}=\pd_{p_k}+\pd_{p'_k}$ and $\pd_{q'_k}=\pd_{p_k}-\pd_{p'_k}$.
This leads to the following transformed HQE:
\beq\label{eqn:HQE-des2}
\mathop{\Res}_{\lambda=\infty}
e^{2\sum_{k\geq 1}\frac{\lambda^k}{k} q'_k}
e^{-\sum_{k\geq 1}\frac{1}{\lambda^{k}}\pd_{q'_k}}
Z({\bf q+q'}) Z({\bf q-q'})\, d\lambda=0.
\eeq
The infinite sequence of differential equations of the KP hierarchy are then decoded from this HQE by expanding in variable ${\bf q'}$. As explained by Givental in \cite{Giv03}, the coefficient at each monomial ${\bf q'}^{\bf m}$ is a Laurent series in $\lambda^{-1}$, meaning the $\lambda$ degree is bounded above. Therefore, the expressions in HQEs \eqref{eqn:HQE-KP} and \eqref{eqn:HQE-des2} should be considered as expansions near $\lambda=\infty$.

Now we consider the HQE formulation for the polynomial-reduced KP hierarchy.
Notice that $[\widehat{\lambda^k},\Gamma^{\pm}]=\pm \lambda^k\cdot\Gamma^{\pm}$.
By considering the quantization $\widehat{x_\lambda^n}$ in terms of variable ${\bf p}$ and by taking the action of $\widehat{x_\lambda^n}$ on the HQE~\eqref{eqn:HQE-KP},
we see the HQE for polynomial-reduced KP hierarchy is equivalent to the following HQEs:
\beq\label{eqn:HQE-xKP}
\mathop{\Res}_{\lambda=\infty} x_\lambda^{n}\cdot (\Gamma^{+}Z)({\bf p})\, (\Gamma^{-}Z)({\bf p'}) \, d\lambda=0,\qquad n\geq 0.
\eeq
As in the HQE for KP hierarchy, the expression in~\eqref{eqn:HQE-xKP} should be understood as the Laurent series in $\lambda^{-1}$.

We view $x_{\lambda}$ as a local degree $r$ covering map from $\mathcal U\subset\Sigma$ to $\mathcal U'\subset\mathbb P^1$, for a point $u\in\mathcal U'$, we denote by $X_{u}=x_{\lambda}^{-1}(u)\subset \mathcal U$ the set of preimages of $u$.
When $u$ is taken to be an arbitrary point instead of a particular one, we omit the subscript $u$ in the notation $X_u$.
For $\sigma\in X$ we denote by $\lambda_{\sigma}$ the (local) coordinate of $\sigma$.
Since the summation $\sum_{\sigma\in X}\lambda_{\sigma}^{k}d\lambda_{\sigma}$ is invariant under the action of symmetric group $S_{r}:X\to X$,  it must be in the form of $f(x(z))dx(z)$ for some polynomial $f(x)$ when $k\geq 0$ or Taylor series $f(x)$ in variable $x^{-1}$ when $k<0$.
We conclude that a function $Z({\bf p})$ which satisfies condition~\eqref{eqn:red-Z-x}
is a tau-function of  polynomial-reduced KP hierarchy if and only if the following condition holds:
\beq\label{eqn:HQE-KP-x}
\sum_{\sigma\in X}(\Gamma^{+\sigma}Z)({\bf p})\, (\Gamma^{-\sigma}Z)({\bf p'}) \, d\lambda_{\sigma} \, {\text{ is regular in }} x,
\eeq
where
$$
\Gamma^{\pm\sigma}=\exp\bigg(\pm\sum_{k\geq 1}\frac{\lambda_{\sigma}^k}{k}p_{k}\bigg)
\exp\bigg(\mp\sum_{k\geq 1}\frac{1}{\lambda_{\sigma}^{k}}\pd_{p_{k}}\bigg).
$$
We note that the expression in~\eqref{eqn:HQE-KP-x} should be understood as the Laurent series in $x^{-1}$,
and we still call~\eqref{eqn:HQE-KP-x} the HQE.
Our goal is to prove the generating series $Z({\bf p};\hbar)$ for genus $0$ curve or its non-perturbative modification $Z^{\rm NP}_{\mu,\nu}({\bf p};\hbar;w)$ for higher genus curve
satisfies the condition~\eqref{eqn:HQE-KP-x}.
\begin{convention}
	To describe the results in a unified way, we use the notation $Z({\bf p};\hbar;w)$ (resp. $\cA({\bf s};\hbar;w)$) to denote $Z({\bf p};\hbar)$ (resp. $\cA({\bf s};\hbar)$) for genus 0 curve and to denote $Z^{\rm NP}_{\mu,\nu}({\bf p};\hbar;w)$  (resp. $\cA^{\rm NP}_{\mu,\nu,\tau}({\bf s};\hbar;w)$)  for higher genus curve.
	For genus 0 curve, we view $w$ and all of its shifts appear in the expression as $0$.
\end{convention}

\subsection{From TR descendent  to CohFT ancestor}
Now we return to the TR side.
By the relation of generating series $Z({\bf p};\hbar)$ with total ancestor potential $\cA({\bf s};\hbar)$ (equation \eqref{eqn:DKP-ACohFT}) or its non-perturbative analogue (equation \eqref{eqn:DKP-ACohFT-NP}),
we transform the HQE~\eqref{eqn:HQE-KP-x} for potential $Z$ into the quadratic equation (which we still call HQE) for potential $\cA$.

We introduce the inverse action of operator $D=-\frac{d}{dx(z)}$ on a function $f(z)$ of the Riemann surface $\Sigma$:
$$\textstyle
D^{-1}f(z)=-\int f(z)dx(z).
$$
We note here that $f(z)$ can be globally multi-valued on $\Sigma$, and we assume it is locally meromorphic, it is easy to see the same property holds for  $D^{-1}f(z)$.

In this section, we prove the following Proposition:
\begin{proposition}\label{prop:HQE-Z-to-A}
The HQE~\eqref{eqn:HQE-KP-x} for $Z({\bf p};\hbar;w)$ holds if the following HQE for $\cA(\hbar\cdot{\bf s};\hbar;w)$ holds:
\begin{align}
\sum_{\sigma\in X}
&\,  (\Gamma_{\mathcal A}^{+\sigma}\mathcal A)(\hbar\cdot{\bf s};\hbar;w-u(z_{\sigma})+\widehat{u(z)})
  \cdot (\Gamma_{\mathcal A}^{-\sigma} \mathcal A)(\hbar\cdot{\bf s'};\hbar;w+u(z_{\sigma})+\widehat{u(z')}) \nonumber\\
&\,\cdot  \tfrac{dh_{\bf c}(z_{\sigma})}{\theta[{\bf c}](u(z_{\sigma})|\tau)^2} \,
{\text{ is regular in }} x, \label{eqn:HQE-ancestor}
\end{align}
(recall $\tfrac{dh_{\bf c}(z)}{\theta[{\bf c}](u(z)|\tau)^2}=\frac{dz}{(z-b)^2}$ for genus-0 curve), where $u(z)=\int_{z'=b}^{z}du(z')$ and
$$
\Gamma^{\pm\sigma}_{\mathcal A}=\exp\Big(\pm\sum_{k,\beta}(-D)^{-k-1}\zeta^{\bar\beta}(z_{\sigma})s^{\bar\beta}_{k}\Big)
\exp\Big(\mp\sum_{k,\beta} D^k\zeta^{\bar\beta}(z_{\sigma})\pd_{s^{\bar\beta}_{k}}\Big).
$$
\end{proposition}
\begin{remark}
As was explained before, the expression in~\eqref{eqn:HQE-ancestor} should be understood as Laurent series in $x^{-1}$ near $x=\infty$.
There is an ambiguity of the term $(-D)^{-k-1}\zeta^{\bar\beta}(z_{\sigma})$ by choosing a at most degree $k$ polynomial of $x(z)$, it is easy to see this ambiguity does not affect the validity of condition~\eqref{eqn:HQE-ancestor}.
\end{remark}
\begin{proof}[Proof of Proposition~\ref{prop:HQE-Z-to-A}]
To write down the HQE for $\mathcal A(\hbar\cdot{\bf s};\hbar;w)$,
we just need to consider how the vertex operators $\Gamma^{\pm}$ go across $e^{\frac{1}{2}Q({\bf p},{\bf p})}$
and $e^{2\pi {\bf i} \widehat{u(z)}\cdot \nabla_w}$.
By using the following well-known identity:
\beq\label{eqn:BCH}
e^{A_1}e^{A_2}=e^{\sum_{n\geq0}\frac{1}{n!}\ad_{A_1}^{n}(A_2)}e^{A_1},
\eeq
we have
$$
\Gamma^{\pm}e^{\frac{1}{2}Q({\bf p},{\bf p})}
=e^{\frac{1}{2}Q({\bf p},{\bf p})\pm \frac{1}{2}[A^{\lambda},Q({\bf p},{\bf p})] +\frac{1}{4}[A^{\lambda},[A^{\lambda},Q({\bf p},{\bf p})]]}\Gamma^{\pm},
$$
where $A^{\lambda}=\sum_{k\geq 1}\frac{1}{\lambda^k}\pd_{p_k}$.
It is not hard to see
$$\textstyle
\frac{1}{2}[A^{\lambda},Q({\bf p},{\bf p})]=Q(\lambda,{\bf p}),\qquad
\frac{1}{4}[A^{\lambda},[A^{\lambda},Q({\bf p},{\bf p})]]=\frac{1}{2}Q(\lambda,\lambda),
$$
where the notations $Q(\lambda,{\bf p})$ and $Q(\lambda,\lambda)$ mean that we substitute the corresponding $p_k$ into $\lambda^{-k}$ in the expression.
Similarly and more easier, we have
$$
\Gamma^{\pm}e^{2\pi {\bf i} \widehat{u(z)}\cdot \nabla_w}
=e^{2\pi {\bf i} (\widehat{u(z)}\mp u(z))\cdot\nabla_w}\Gamma^{\pm}.
$$
We can simply change the variables of derivative part $\partial_{\bf p}$ into $\partial_{\bf s}$ in the vertex operator when acting on a function $\mathcal G({\bf s}({\bf p}))$:
$$
\exp\Big(\mp\sum_{k} \lambda^{-k}\pd_{p_{k}}\Big)
=\exp\Big(\mp\sum_{k,\beta} D^k\zeta^{\bar\beta}(z)\pd_{s^{\bar\beta}_{k}}\Big).
$$
We get
$$
\Gamma^{\pm}e^{\frac{1}{2}Q({\bf p},{\bf p})}e^{2\pi {\bf i} \widehat{u(z)}\cdot \nabla_w}
=e^{\frac{1}{2}Q({\bf p},{\bf p})+\frac{1}{2}Q(\lambda,\lambda)}
e^{2\pi {\bf i} (\widehat{u(z)}\mp u(z))\cdot \nabla_w}\widetilde\Gamma_{\cA}^{\pm},
$$
where
$$
\widetilde\Gamma_{\cA}^{\pm}=\exp\Big(\pm\sum_{k}\frac{\lambda^k}{k}p_k\mp Q(\lambda,{\bf p})\Big)
\exp\Big(\mp\sum_{k,\beta} D^k\zeta^{\bar\beta}(z)\pd_{s^{\bar\beta}_{k}}\Big).
$$

Then we translate the multiplication part of the vertex operator $\widetilde\Gamma_{\cA}^{\pm}$ in terms of the ancestor variables $s^{\bar\alpha}_k$. We introduce the quantization $\widehat{\zeta^{\bar\alpha}_k(z)}:=s^{\bar\alpha}_k$ and $\widehat{\lambda^{-k}}:=p_k$.
Consider the expansion of $B(z,z_2)$ at boundary point $z=b$, $z_2=b$
with the local coordinate $\lambda<\lambda_2$, we see the expression
$-\sum_{k}\frac{\lambda^k}{k}p_k+Q(\lambda,{\bf p})$
can be viewed as the quantization of the function $\int_{z} \int_{z_2} B(z,z_2)$ on the second variable.
We claim that at $z_2=b$, for a fixed local coordinate $\lambda_2$, one has
$$
B(z,z_2)+\sum_{k\geq 0}\sum_{\beta=1,\cdots,N}d(-D)^{-k-1}\zeta^{\bar\beta}(z) d D^k\zeta^{\bar\beta}(z_2)\in d_{z}d_{\lambda_2}\mathbb C[x(z)][[\lambda_{2}^{-1}]].
$$
The proof is given in Appendix (see Lemma~\ref{lem:B-x}). Hence,
\beq\label{eqn:B-s-p}
\sum_{k}\frac{\lambda^k}{k}p_k-Q(\lambda,{\bf p}) =\sum_{k,\beta}(-D)^{-k-1}\zeta^{\bar\beta}(z)s^{\bar\beta}_k({\bf p}) +\sum_{k}f_k(x(z))p_k,
\eeq
for some polynomials $f_k(x)$.
We obtain that the HQE \eqref{eqn:HQE-KP-x} for $Z$ is equivalent the following HQE for $\cA$:
\begin{align*}
	\sum_{\sigma\in X}\, e^{Q(\lambda_{\sigma},\lambda_{\sigma})}\,
	&\, \cdot (\Gamma_{\mathcal A}^{+\sigma}\mathcal A)(\hbar\cdot{\bf s};\hbar;w-u(z_{\sigma})+\widehat{u(z)})\, \nonumber\\
	&\,  \cdot (\Gamma_{\mathcal A}^{-\sigma} \mathcal A)(\hbar\cdot{\bf s'};\hbar;w+u(z_{\sigma})+\widehat{u(z')}) \cdot  d\lambda_{\sigma} \, {\text{ is regular in }} x,
\end{align*}
where
$$\textstyle
Q(\lambda,\lambda)=\int_{\lambda_1=\infty}^{\lambda}\int_{\lambda_2=\infty}^{\lambda}\Big(B(z_1,z_2)-\frac{d\lambda_1d\lambda_2}{(\lambda_1-\lambda_2)^2}\Big).
$$
By equations~\eqref{eqn:int-bergman} and \eqref{def:primeform}, we have
$$\textstyle
e^{Q(\lambda,\lambda)}d\lambda
=-\frac{d h_{\bf c}}{d\lambda^{-1}}(b)\cdot \frac{dh_{\bf c}(z)}{\theta[{\bf c}](u(z)|\tau)^2}.
$$
Note that $\frac{dh_{\bf c}}{d\lambda^{-1}}(b)$ is just a constant, the Proposition is proved.
\end{proof}
Let $\lambda'$ be another local coordinate near the boundary point $b$ such that $x(z)\in \mathbb C[\lambda']$,
we denote by $Z^{\lambda'}({\bf p'};\hbar;w)$ and $Z^{\lambda}({\bf p}, \hbar;w)$ the TR descendent generating series defined by $\lambda'$ and $\lambda$, respectively. We have the following corollary of Proposition~\ref{prop:HQE-Z-to-A}.
\begin{corollary}\label{cor:KP-coordchange}
The generating series $Z^{\lambda}({\bf p}, \hbar;w)$ is a tau-function of $x_{\lambda}$ polynomial-reduced KP hierarchy if and only if $Z^{\lambda'}({\bf p'};\hbar;w)$ is a tau-function of $x_{\lambda'}$ polynomial-reduced KP hierarchy.
In particular, if $x(z)=c\cdot (\lambda')^{r}$ is a monomial of $\lambda'$, the polynomial-reduced KP integrability of $Z^{\lambda}({\bf p}, \hbar;w)$ implies that $Z^{\lambda'}({\bf p'};\hbar;w)$ is a tau-function of $r$KdV hierarchy.
\end{corollary}
\begin{proof}
This following immediately from Proposition~\ref{prop:HQE-Z-to-A} by noticing that equation~\eqref{eqn:HQE-ancestor} does not depend on the choice of the local coordinate $\lambda$.
\end{proof}
\begin{remark} Corollary~\ref{cor:KP-coordchange} generalizes Kazarian's result~\cite[Theorem 2.5]{Kaz09} into the polynomial-reduced KP case from the point of view of topological recursion, from which the changing of quadratic part $Q^{\lambda}({\bf p},{\bf p})-Q^{\lambda'}({\bf p'},{\bf p'})$ arise naturally.
Here ${\bf p}$ and ${\bf p'}$ are related by the coordinate transformation between $\lambda$ and $\lambda'$ via the quantization $p_k=\widehat{\lambda^{-k}}$, $p_k=\widehat{\lambda'^{-k}}$.
\end{remark}
\subsection{From boundary point to branch point}
Now we explain that the expression in~\eqref{eqn:HQE-ancestor} can be understood not only as Laurent series in $x^{-1}$, but also as an analytic function in $x$, thanks to the {\it tameness}~\cite{Giv03} property of the total ancestor potential  $\cA$.

\begin{definition}We call a generating series $\cG({\bf s};\hbar)=e^{\sum_{g\geq 0}\hbar^{2g-2}\bar\cF_g({\bf s})}$ is {\it tame} if
$$
\frac{\pd^n \bar\cF_g({\bf s})}{\pd s_{k_1}^{i_1}\cdots \pd s_{k_n}^{i_n}}\bigg|_{{\bf s=0}}=0 \quad
{\text{ whenever }} \quad k_1+k_2+\cdots+k_n>3g-3+n.
$$
\end{definition}
In particular, $\bar\cF_g({\bf s})$ is a formal series $\sum_{m,n\in (\mathbb Z_{+})^{ N}} \bar\cF^{(g)}_{\vec{m},{\vec n}}(s_2,\cdots,s_{3g-2+{|{m}|}})\cdot (s_0)^{{m}}\cdot (s_1)^{{n}}$  where the coefficients $\bar\cF^{(g)}_{\vec{m},{\vec n}}$ are polynomials on $s_2,\cdots,s_{3g-2+|\vec{m}|}$.
Givental~\cite[Proposition 5]{Giv03} proved that any potential $\cA({\bf s};\hbar)$, which is reconstructed from $\widehat{R}$ (it is trivial to generalize this to allow a shift on the coordinates $s_{k\geq 2}$), is tame.
For the total ancestor potential of a CohFT, the tameness follows immediately from the definition since $\dim \Mbar_{g,n}=3g-3+n$.

\begin{lemma}\label{lem:HQE-globality}
The expression in the HQE \eqref{eqn:HQE-ancestor} can be extended to be a meromorphic form on $\mathbb P^1$.
\end{lemma}
\begin{proof}
We note firstly via the ramification map $x:\Sigma\to \mathbb P^1$, each locally-defined form $f(z)$ on $\Sigma$ defines a locally-defined form $F(p):=\sum_{\sigma\in x^{-1}(p)}f(z_{\sigma})$ on $\mathbb P^1$, thus the expression in~\eqref{eqn:HQE-ancestor} can be viewed as a locally-defined form on $\mathbb P^1$ near $x=\infty$.
For an arbitrary point $p\in \mathbb P^1$, we pick a path $\gamma$ that connects $p$ and $\infty$ and extend a locally-defined form $F$ near $\infty$ to $p$ along this path: for $p'\in \gamma$, define $F(p'):=\sum_{\sigma\in x^{-1}(p')}f(z_{\sigma})$. $F$ can be globally-defined by this way if the corresponding form $f(z)$ on $\Sigma$ is globally defined.
Therefore, to prove the Lemma, we just need to prove the expression in~\eqref{eqn:HQE-ancestor} will not change when $z_{\sigma}$ changes along cycles.
This is trivial for the rational curve.
In what follows, we study how the expression in~\eqref{eqn:HQE-ancestor} changes when $z_{\sigma}$ changes along cycles: $ z_{\sigma}\to z_{\sigma}+n\cdot \cyA+m\cdot \cyB$ for ${\frak g}(\Sigma)>0$.

Firstly, by equation~\eqref{eqn:quasiperiod-theta}, we see the differential $\frac{dh_{\bf c}(z_{\sigma})}{\theta[{\bf c}](u(z_{\sigma})|\tau)^2}$ becomes
$$
e^{4\pi i m\cdot u(z_{\sigma})+2\pi i m^{t}\tau m}\cdot \tfrac{dh_{\bf c}(z_{\sigma})}{\theta[{\bf c}](u(z_{\sigma})|\tau)^2}.
$$
Secondly, the operator $\Gamma_{\cA}^{\pm \sigma}$ becomes
$$
e^{\pm \sum_k\eta(m\cdot \varphi^{(-k-1)},s_k)}\cdot \Gamma_{\cA}^{\pm \sigma}\cdot e^{\mp m\cdot \varphi}.
$$
where $\varphi=(\varphi_1,\cdots,\varphi_{\frak g})$ is given by $\varphi_j=\sum_a \varphi_j^a\pd_{s_0^a}$ (see equation~\eqref{def:varphi})
and $\varphi^{(-k-1)}=(\varphi^{(-k-1)}_1,\cdots,\varphi^{(-k-1)}_{\frak g})$ is given by
$\varphi^{(-k-1)}_j=\sum_{\beta}\bar e_\beta \oint_{\cyB_j}d(-D)^{-k-1}\zeta^{\bar\beta}(z)$.
Thirdly, by equation~\eqref{eqn:quasiperiod-theta} and equation~\eqref{eqn:NPA-formula}, the potential $\cA({\bf s};\hbar;w\mp u(z_{\sigma})+\widehat{u(z)})$ becomes
$$
e^{-\pi i m^{t}\tau m\pm 2 \pi i m(w\mp u(z_{\sigma})+\widehat{u(z)}+\nu)\pm m\cdot \varphi \mp2\pi i (n\mu-m\nu)}
\cdot\cA({\bf s};\hbar;w\mp u(z_{\sigma})+\widehat{u(z)}).
$$
Lastly, recall equation~\eqref{eqn:B-s-p}, whose left-hand side is understood as the quantization of $-\int_z\int_{z'} B(z,z')$ on the second variable. When taking action $\oint_{\cyB_j}d$ on this equation, we see
$$\textstyle
2\pi i \widehat{u(z)}+\sum_{k}\eta(\varphi^{(-k-1)},s_k)=0.
$$
Put all these terms together, we see when $z_{\sigma}$ changes along cycles: $ z_{\sigma}\to z_{\sigma}+n\cdot \cyA+m\cdot \cyB$, the expression in~\eqref{eqn:HQE-ancestor} changes by the factor
$$
e^{2\pi i m\cdot \widehat{u(z)}+ \sum_k\eta(m\cdot \varphi^{(-k-1)},s_k)}\cdot
e^{-2\pi i m\cdot \widehat{u(z')}- \sum_k\eta(m\cdot \varphi^{(-k-1)},s'_k)}=1.
$$
The proof is finished.
\end{proof}
\begin{proposition}\label{prop:HQE-boundary-to-branch}
The HQE \eqref{eqn:HQE-ancestor} holds for $\cA(\hbar\cdot {\bf s};\hbar;w)$ if the following HQEs hold for $\cA(\hbar\cdot{\bf s};\hbar)$: for each $\beta=1,\cdots,N$,
\beq\label{eqn:HQE-ancestor-beta}
\sum_{\sigma\in X}   (\Gamma_{\mathcal A}^{+\sigma}\mathcal A)(\hbar\cdot {\bf s};\hbar)\,
\cdot (\Gamma_{\mathcal A}^{-\sigma} \mathcal A)(\hbar\cdot {\bf s'};\hbar)\,
\cdot \frac{dh_{\bf c}(z_{\sigma})}{\theta[{\bf c}](u(z_{\sigma})|\tau)^2} \, {\text{ is regular at }} x=x^\beta.
\eeq
\end{proposition}
\begin{proof}
We follow the method introduced in~\cite[Proposition 6]{Giv03}.
Similar as the explanation of the original HQE (see equation~\eqref{eqn:HQE-des2}), the expression in~\eqref{eqn:HQE-ancestor} can be rewritten as follows:
\begin{align}
\sum_{\sigma\in X} &\, e^{2\sum_{k,\bar\alpha}(-D)^{-k-1}\zeta^{\bar\alpha}(z_{\sigma})s'^{\bar\alpha}_k}
e^{-\sum_{k,\bar\alpha}D^k\zeta^{\bar\alpha}(z_{\sigma})\pd_{s'^{\bar\alpha}_k}}
\Big(\cA(\hbar\cdot({\bf s+s'});\hbar;w-u(z_\sigma)+\widehat{u(z)}+\widehat{u(z')})\nonumber \\
&\, \cdot \cA(\hbar\cdot({\bf s-s'});\hbar;w+u(z_\sigma)+\widehat{u(z)}-\widehat{u(z')}) \Big)
\cdot \frac{dh_{\bf c}(z_\sigma)}{\theta[{\bf c}](u(z_\sigma)|\tau)^2}. \label{eqn:HQE-ancestor-exp}
\end{align}
Notice that $\cA(\hbar\cdot {\bf s};\hbar)$ has form $\sum_{n\geq 0}\cA_n({\bf s})\cdot \hbar^n$ and $\cA_n$ consists of finitely many polynomials in finitely many variables $s_{0},s_1,\cdots,s_{M}$, $M=M(n)<\infty$.
Given a monomial of $\hbar$, ${\bf s}$ and ${\bf s'}$ in the expression, we see its coefficient consists of the form $\frac{dh_{\bf c}(z_\sigma)}{\theta[{\bf c}](u(z_\sigma)|\tau)^2}$ times the polynomials in $D^{k}\zeta^{\bar \beta}(z^\sigma)$, $(-D)^{-k-1}\zeta^{\bar\beta}(z_\sigma)$ and derivatives of $\theta[\cmn](w\mp u(z_\sigma)+\widehat{u(z)}\pm \widehat{u(z')}|\tau)$ with respect to $w$.
We have proved that these are global defined forms on $\mathbb P^1$.
Furthermore, it is easy to see the only possible singularity of these forms can appear at the boundary point and the critical point of $x(z)$.
Therefore, the expression is regular in $x$ if it is regular at $x=x^\beta$.

Now we suppose we have known the regularity of the expression in~\eqref{eqn:HQE-ancestor-beta} at $x=x^\beta$. Then by above explanation, for each monomial of $\hbar$, ${\bf s}$ and ${\bf s'}$, its coefficient, which takes the form $\frac{dh_{\bf c}(z_\sigma)}{\theta[{\bf c}](u(z_\sigma)|\tau)^2}$ times the polynomials in $D^{k}\zeta^{\bar \beta}(z^\sigma)$ and $(-D)^{-k-1}\zeta^{\bar\beta}(z_\sigma)$, is regular at $x=x^\beta$.
It is clear that the expression remains regular when multiplied by a function $\theta[\cmn](w\mp u(z_\sigma)+\widehat{u(z)}\pm \widehat{u(z')}|\tau)$ or its derivatives with respect to $w$.
Recall that (see equation~\eqref{eqn:NPA-formula})
$$
\cA(\hbar\cdot({\bf s\pm s'});\hbar;w\mp u(z_\sigma)+\widehat{u(z)}\pm \widehat{u(z')})
=e^{\varphi\cdot \nabla_{w}}\big(\cA(\hbar\cdot{\bf s\pm s'};\hbar)\cdot \theta[\cmn](w\mp u(z_\sigma)+\widehat{u(z)}\pm \widehat{u(z')}|\tau)\big).
$$
We see for each monomial $\hbar$, ${\bf s}$ and ${\bf s'}$ in expression~\eqref{eqn:HQE-ancestor-exp}, its coefficient is combined by finitely many coefficients of monomials of $\hbar$, ${\bf s}$ and ${\bf s'}$ in expression~\eqref{eqn:HQE-ancestor-beta}.
Therefore, the HQE~\eqref{eqn:HQE-ancestor-beta} implies the HQE \eqref{eqn:HQE-ancestor}. The Proposition is proved.
\end{proof}

\subsection{Finishing the proof of Theorem~\ref{KP-TR}}
Now we consider the HQE for ancestor potential $\cA$ at branch points, and use the formula~\eqref{eqn:formula-A} to translate the HQE for $\cA$ into the ones for $\cD^{\rm KW}$,
and finish the proof of Theorem~\ref{KP-TR} by applying the original Witten conjecture/Kontsevich theorem.
\begin{proposition}\label{prop:HQE-branch-to-KW}
	For each $\beta=1,\cdots, N$, the HQE \eqref{eqn:HQE-ancestor-beta} hold for $\cA(\hbar\cdot{\bf s};\hbar)$.
\end{proposition}
\begin{proof}
For $x\in \mathbb P^{1}$ near $x^\beta$, there are two types of preimages $\sigma\in X$ of $x$. The first type contains $r-2$ points $\sigma^{\beta}_{i}$, $i=1,\cdots,r-2$, away from the branch point $z^\beta$,
and the second type contains two points $\sigma_{\pm}^{\beta}$ near the branch point $z^\beta$.
For the first type, as we can see in the proof of the Proposition~\ref{prop:HQE-boundary-to-branch}, the summation in the expression in~\eqref{eqn:HQE-ancestor-beta} is regular.
For the second type, we have the local coordinates $\eta^{\beta}_\sigma=\pm \eta^\beta$ (Airy coordinates, see \S \ref{sec:def-TR} for the definition) for points $\sigma=\sigma_{\pm}^{\beta}$.
It remains to prove that
\beq\label{eqn:HQE-A-crit}
\sum_{\sigma=\sigma_{\pm}^{\beta}}  (\Gamma_{\mathcal A}^{+\sigma}\mathcal A)(\hbar\cdot{\bf s};\hbar)\,
\cdot (\Gamma_{\mathcal A}^{-\sigma} \mathcal A)(\hbar\cdot{\bf s'};\hbar)\,
\cdot \frac{dh_{\bf c}(z_\sigma)}{\theta[{\bf c}](u(z_\sigma)|\tau)^2} \, {\text{ is regular at }} \eta^\beta=0.
\eeq

Now we consider the local expansion of vertex operators $\Gamma_{\cA}^{\pm\sigma}$ at $z_\sigma=z^\beta$.
From the coordinate transformation ${\bf s}(\givz)=R(\givz)({\bf q}(\givz)-T(\givz))$
and equation~\eqref{eqn:zeta-xi}, we have
$$
\exp\Big(\mp\sum_{k,\bar\alpha} D^k\zeta^{\bar\alpha}(z_\sigma)\pd_{s^{\bar\alpha}_{k}}\Big)
=\exp\Big(\mp\sum_{k,\bar\alpha} \xi_k^{\bar\alpha}(z_\sigma)\pd_{q^{\bar\alpha}_{k}}\Big).
$$
Further by the relation of $\zeta^\alpha$ and $R$-matrix~\eqref{def:EO-R}, we have the expansion  locally at $z_\sigma=z^{\beta}$
$$\textstyle
\zeta^{\bar\alpha}(z_\sigma)|_{\beta}=\sum_{l\geq 0}(R_l)^{\bar\alpha}_{\bar\beta}\cdot D^{-l}\big(\frac{1}{\eta_\sigma^\beta}\big)
+ f_\sigma^{\bar\alpha}\big((\eta^\beta)^2\big) ,
$$
where $f_\sigma^{\bar\alpha}((\eta^\beta)^2)$ is a function of $(\eta^\beta)^2$ and is regular at $\eta^\beta=0$, which leads to
$$\textstyle
D^{-k}\zeta^{\bar\alpha}(z_\sigma)|_{\beta}=\sum_{l\geq 0}(R_l)^{\bar\alpha}_{\bar\beta}\cdot D^{-l-k}\big(\frac{1}{\eta^\beta_\sigma}\big)
+ D^{-k}f_\sigma^{\bar\alpha}\big((\eta^\beta)^2\big).
$$
Combining the computations above, we translate the vertex operators in terms of $\eta^{\beta}_\sigma$
\begin{align*}
\exp\Big(\pm\sum_{k,\bar\alpha}(-D)^{-k-1}\zeta_{\bar\alpha}(z_\sigma)|_{\beta}s^{\bar\alpha}_{k}\Big)
=e^{\pm C^{\beta}_\sigma}\cdot \exp\Big(\pm\sum_{k}(-D)^{-k-1}\big((\eta_\sigma^\beta)^{-1}\big) \big(q^{\bar\beta}_{k}-T_k^{\bar\beta}\big)\Big),
\end{align*}
where $C^{\beta}_\sigma=\sum_{k,l}(-D)^{-k-l-1}f_\sigma^{\bar\alpha}((\eta^\beta)^2)(R_{l})^{\bar\alpha}_{\bar\gamma}(q_k^{\bar\gamma}-T^{\bar\gamma}_k)$
is regular at $\eta^\beta=0$, meaning the coefficient of $q_k^{\bar\gamma}-T^{\gamma}_k$ is regular at $\eta^\beta=0$ for each $k,\bar\gamma$.
We fix the operator $D^{-1}$ by setting the integration constant to zero:
$$\textstyle
(-D)^{-k-1}\big((\eta_\sigma^\beta)^{-1}\big)=\frac{(\eta_\sigma^\beta)^{2k+1}}{(2k+1)!!}.
$$
We note here that different choices will differ by a at most degree $k$ polynomial on $(\eta^\beta)^2$ which will not affect the validity of the proof. Then locally at $z_\sigma=z^\beta$, by the identity~\eqref{eqn:BCH} we have the following equality
\beq\label{eqn:GammaA-R}
\Gamma^{\pm\sigma}_{\cA}e^{\frac{1}{2}V(\pd_{\bf{q}},\pd_{\bf{q}})}
=e^{\frac{1}{2}V(\pd_{\bf{q}},\pd_{\bf{q}}) \pm\frac{1}{2}[A_\sigma^{\beta},V(\pd_{\bf{q}},\pd_{\bf{q}})]
+\frac{1}{4}[A_\sigma^{\beta},[A_\sigma^{\beta} ,V(\pd_{\bf{q}},\pd_{\bf{q}})]]} \Gamma^{\pm\sigma}_{\cA},
\eeq
where $A_\sigma^{\beta}=\sum_{k}\frac{(\eta_\sigma^\beta)^{2k+1}}{(2k+1)!!}\big(q^{\bar\beta}_{k}-T_k^{\bar\beta}\big)$.
By using ~\eqref{eqn:BCH} again, equation~\eqref{eqn:GammaA-R} gives
\begin{align*}
\Gamma^{\pm\sigma}_{\cA}e^{\frac{1}{2}V(\pd_{\bf{q}},\pd_{\bf{q}})}
=&e^{\pm C^{\beta}_\sigma}e^{\frac{1}{2}V(\pd_{\bf{q}},\pd_{\bf{q}}) -\frac{1}{4}[A_\sigma^{\beta},[A_\sigma^{\beta} ,V(\pd_{\bf{q}},\pd_{\bf{q}})]]} \widetilde\Gamma^{\pm\sigma}_{\rm KW_{\beta}},
\end{align*}
where
$$
\widetilde\Gamma^{\pm\sigma}_{\rm KW_{\beta}}
=\exp\Big(\pm\sum_{k}\frac{(\eta_\sigma^\beta)^{2k+1}}{(2k+1)!!}(q^{\bar\beta}_{k}-T^{\bar\beta}_{k})\Big)
\exp\Big(\pm\frac{1}{2}[A_\sigma^\beta,V(\pd_{\bf{q}},\pd_{\bf{q}})]
\mp\sum_{k,\bar\alpha} \xi_k^{\bar\alpha}(z_\sigma)\pd_{q^{\bar\alpha}_{k}}\Big).
$$
Similar to the treatment with $Q({\bf p},{\bf p})$, we have
$$
\tfrac{1}{2}[A_\sigma^\beta,V(\pd_{\bf{q}},\pd_{\bf{q}})]=-V(\pd_{\bf{q}}, \eta_\sigma^\beta),\qquad
\tfrac{1}{4}[A_\sigma^\beta,[A_\sigma^\beta,V(\pd_{\bf{q}},\pd_{\bf{q}})]]=\tfrac{1}{2}V(\eta_\sigma^\beta, \eta_\sigma^\beta),
$$
where the notations $V(\pd_{\bf{q}}, \eta_\sigma^\beta)$ and $V(\eta_\sigma^\beta,\eta_\sigma^\beta)$ mean that we substitute the corresponding $\pd_{q_k^\beta}$ into $\frac{(\eta_\sigma^\beta)^{2k+1}}{(2k+1)!!}$ (and the $\pd_{q_k^\alpha}$ into $0$ for $\alpha\ne\beta$).
By the definition and direct computations,
$$\textstyle
V(\pd_{\bf{q}}, \eta_\sigma^\beta)+\sum_{k,\bar\alpha} \xi_k^{\bar\alpha}(z_\sigma)\pd_{q^{\bar\alpha}_{k}}
=\frac{(2k-1)!!}{(\eta_\sigma^{\beta})^{2k+1}}\pd_{q^{\bar\beta}_{k}}.
$$
Then we see
$$
\widetilde\Gamma^{\pm\sigma}_{\rm KW_{\beta}}
=\exp\Big(\pm \sum_{k}\frac{(\eta_\sigma^\beta)^{2k+1}}{(2k+1)!!}(q^{\bar\beta}_{k}-T^{\bar\beta}_k)\Big)
\exp\Big(\mp \sum_k\frac{(2k-1)!!}{(\eta_\sigma^{\beta})^{2k+1}}\pd_{q^{\bar\beta}_{k}}\Big).
$$
Since the terms $\pm T^{\bar\beta}_k$ cancel each other in the HQEs, $e^{\pm C^{\beta}_\sigma}$ is regular at $\eta^\beta=0$, and the functions $\cD^{\rm KW}(\hbar\cdot {\bf q}^{\alpha};\hbar)$ for $\alpha\ne \beta$ are regular factors, equation~\eqref{eqn:HQE-A-crit} is equivalent to the following condition:
$$
\sum_{\sigma=\sigma^{\beta}_{\pm}}
 \big(\Gamma_{\rm KW_{\beta}}^{+\sigma} \cD^{\rm KW}\big)(\hbar\cdot {\bf q}^{\beta};\hbar) \cdot
\big(\Gamma_{\rm KW_{\beta}}^{-\sigma}\cD^{\rm KW}\big)(\hbar\cdot{\bf q'}^{\beta};\hbar)
 \cdot  \frac{e^{-V(\eta_{\sigma}^\beta, \eta_{\sigma}^\beta)}dh_{\bf c}(z_\sigma)}{\theta[{\bf c}](u(z_\sigma)|\tau)^2} \,
{\text{ is regular at }} \eta^\beta=0,
$$
where
$$
\Gamma^{\pm \sigma}_{\rm KW_{\beta}}
=\exp\Big(\pm \sum_{k}\frac{(\eta_{\sigma}^\beta)^{2k+1}}{(2k+1)!!}q^{\bar\beta}_{k}\Big)
\exp\Big(\mp \sum_k\frac{(2k-1)!!}{(\eta_{\sigma}^{\beta})^{2k+1}}\pd_{q^{\bar\beta}_{k}}\Big).
$$
Similar as the ancestor case, this HQE should be understood as the differential $\frac{e^{-V(\eta_{\sigma}^\beta, \eta_{\sigma}^\beta)}dh_{\bf c}(z_\sigma)}{\theta[{\bf c}](u(z_\sigma)|\tau)^2}$ times polynomials in $\eta^\beta_{\sigma}$ and $(\eta^\beta_{\sigma})^{-1}$ for each monomial of $\hbar$, ${\bf q}^{\beta}$ and ${\bf q'}^{\beta}$.
Since the differential $\frac{e^{-V(\eta_{\sigma}^\beta, \eta_{\sigma}^\beta)}dh_{\bf c}(z_\sigma)}{\theta[{\bf c}](u(z_\sigma)|\tau)^2}$ is regular at $\eta^\beta=0$ and
$$\textstyle
\lim\limits_{\eta^\beta=0}\, \frac{e^{-V(\eta_{\sigma}^\beta, \eta_{\sigma}^\beta)}dh_{\bf c}(z_\sigma)}{\theta[{\bf c}](u(z_\sigma)|\tau)^2 d\eta^\beta_{\sigma}}
=\frac{1}{\theta[{\bf c}](u(z^\beta)|\tau)^2}\cdot \frac{dh_{\bf c}}{d\eta^\beta}(z^\beta)
$$
is a non-vanishing constant,
we see the HQE is equivalent to the following equation
\beq\label{eqn:HQE-KW}
\sum_{\sigma=\sigma^{\beta}_{\pm}}
(\Gamma_{\rm KW_{\beta}}^{+\sigma} \cD_\beta^{\rm KW})(\hbar\cdot {\bf q};\hbar)\,
(\Gamma_{\rm KW_{\beta}}^{-\sigma}\cD_\beta^{\rm KW})(\hbar\cdot {\bf q'};\hbar )\,
d\eta^{\beta}_{\sigma} \, {\text{ is regular at }} \eta^\beta=0.
\eeq
We note that equation~\eqref{eqn:HQE-KW} is exactly the HQE for the KdV hierarchy and thus follows from the original Witten conjecture/Kontsevich Theorem.
The Proposition is proved.
\end{proof}
Combining the Proposition~\ref{prop:HQE-Z-to-A}, the Proposition~\ref{prop:HQE-boundary-to-branch}
and the Proposition~\ref{prop:HQE-branch-to-KW}, we complete the proof of the Theorem~\ref{KP-TR}.

\section{The deformed negative $r$-spin Witten conjecture}
\label{sec:geo-negative-rspin}
The CohFT of the deformed negative $r$-spin theory, called the deformed Theta class and denoted by $\{\Theta^{r,\epsilon}_{g,n}\}_{2g-2+n>0}$, was recently introduced by Norbury~\cite{Nor23} for $r=2$ without deformation (i.e., $\epsilon=0$) and generalized by Chidambaram, Garcia-Failde and Giacchetto in~\cite{CGG22} for arbitrary integer $r\geq 2$ with deformation parameter $\epsilon$.
In~\cite{CGG22}, a global spectral curve, which we call the deformed $r$-Bessel curve, was found for the deformed Theta class.
In this section, we present the concrete results corresponding to the Picture shown in the Introduction of the negative $r$-spin theory and prove the deformed negative $r$-spin Witten conjecture.
\begin{theorem}
The deformed negative $r$-spin Witten conjecture (Theorem ~\ref{conj:negative-witten-conjecture}) holds.
Specifically, after the change of variables $t_k^a = -\frac{(-1)^{k}}{\sqrt{-r}}\frac{\Gamma(k+\frac{a}{r})}{\Gamma(\frac{a}{r})}  p_{rk+a}$, where $a=1,\cdots,r-1$ and $k\geq 0$,
the total descendent potential $\cD^{r,\epsilon}({\bf t};\hbar)$ gives a tau-function of $r$KdV hierarchy.
\end{theorem}

\subsection{The $\epsilon$-deformation of the Theta class and its relation with the CohFT associated with the $\epsilon$-deformed $r$-Bessel curve}
Now we recall the results~\cite{CGG22} of the Theta class $\Theta^{r}_{g,n}$ and its $\epsilon$-deformation $\Theta^{r,\epsilon}_{g,n}$.
\begin{definition}\label{def:class-r}
	For $2g-2+n>0$, the  negative $r$-spin Witten class and its  $\epsilon$-deformation are defined by the push-forward of top Chern class of $ \mathcal V^{r,-1}_{g,\vec{a}}$ (see \S \ref{sec:negative-r-spin}) along the forgetful map $f: \Mbar_{g,\vec{a}}^{r,-1}\to\Mgn$:
	\beq\label{def:-rspinclass}
	\Upsilon^{r}_{g,n}(\phi_{a_1},\cdots,\phi_{a_n}) := \frac{1}{r^{g-1}}f_*  \big((-1)^{\deg}\cdot c_{\rm top}( \mathcal V^{r,-1}_{g,\vec{a}} )\big),
	\eeq
	\begin{equation} \label{def:deformedclass}
		\Upsilon^{r,\epsilon}_{g,n}(-) :=  \sum_{m\geq 0} \frac{1}{m!} \,p^m_*\Upsilon^{r}_{g,n+m}(-,\epsilon    \phi_{0},\cdots,\epsilon      \phi_{0}) .
	\end{equation}
	Here $p^m: \Mbar_{g,n+m}\to \Mbar_{g,n}$ is the map that forgets the last $m$ marked points.
	The ancestor correlator is defined by
	$$
	\<  \phi_{a_1} \bar \psi_1^{k_1},\cdots, \phi_{a_n} \bar \psi_n^{k_n}\>_{g,n}^{r,\epsilon}
	:= \int_{\M_{g,n}} \Upsilon_{g,n}^{r,\epsilon}(\phi_{a_1},\cdots,\phi_{a_n})  \prod_{i=1}^{n}  \psi_i^{k_i}.
	$$
\end{definition}
\begin{remark}
	Since we have the following degree formula
	$$
	\deg  p^m_* \Upsilon^{r}_{g,n+m}(\phi_{a_1},\cdots,\phi_{a_n},    \phi_{0},\cdots,      \phi_{0})   =   D_{g,n}^{r,-1}( { \vec{a}})  -\tfrac{r-1}{r}m,
	$$
	the summation in equation~\eqref{def:deformedclass} is thus finite.
	The deformed class $\Upsilon^{r,\epsilon}_{g,n}(\phi_{a_1},\cdots,\phi_{a_n})  \in H^*(\M_{g,n})$
	is then a mixed degree class of degree no more than $ D_{g,n}^{r,-1}( { \vec{a}}) $.
\end{remark}

\begin{remark}\label{prop:D-A-identity-without-deformation}
	Since the $\psi$-classes  in $\M_{g,\vec{a}}^{r,-1}$ are the same as the pull back  of the $\psi$-classes in $\M_{g,n}$,
	we see
	$$
	\<  \phi_{a_1}  \psi_1^{k_1},\cdots, \phi_{a_n}  \psi_n^{k_n}\>_{g,n}^{r } = \<  \phi_{a_1} \bar \psi_1^{k_1},\cdots, \phi_{a_n} \bar \psi_n^{k_n}\>_{g,n}^{r,0}.
	$$
\end{remark}
\begin{definition}
The Theta class $\Theta^{r}_{g,n}$ and the $\epsilon$-deformed Theta class $\Theta^{r,\epsilon}_{g,n}$
is the restriction of the map $\Upsilon^{r}_{g,n}$ and $\Upsilon^{r,\epsilon}_{g,n}$ to $\Frob=\sspan_{\mathbb Q}(\phi_1,\cdots,\phi_{r-1})$, respectively, i.e.,
$$
\Theta^{r}_{g,n}:=\Upsilon^{r}_{g,n}\big|_{\Frob},\qquad
\Theta^{r,\epsilon}_{g,n}:=\Upsilon^{r,\epsilon}_{g,n}\big|_{\Frob}.
$$
\end{definition}

\begin{remark}
We note here that the definition of $\Theta^{r}$, as well as $\Theta^{r,\epsilon}$, is different with that in~\cite{CGG22} by a simple factor, the reason we take this definition is to make the descendent time variables $\{t^a_k\}_{a=1,\cdots,r-1; k=0,1,\cdots}$ to be the Witten's time variables introduced in~\cite[(1.5.7)]{Wit93}.
\end{remark}
It is proved in~\cite{CGG22} that the collection $\{\Theta^{r,\epsilon}_{g,n}\}_{2g-2+n>0}$ satisfies the axioms of an $(r-1)$-dimensional CohFT on $\Frob$ with symmetric $2$-form $\eta(\phi_a,\phi_b)=\delta_{a+b,r}$.
Similar as~\cite[Lemma 3.4]{CGG22}, by using the Chiodo formula~\cite{Chi08}, one can compute that for $a,b,c=0,\cdots, r-1$~{\footnote{In~\cite{CGG22}, they proved these qualities for $a,b,c=1,\cdots,r-1$, the same method can be applied for $a,b$ or $c$ taking value of $0$.}},
\beq\label{eqn:-rspin-corr(0,n)}
	\int_{\Mbar_{0,n+3}}\frac{1}{n!}\Upsilon^{r}_{0,n}(\phi_a,\phi_b,\phi_c,\phi_0,\cdots,\phi_0)=\delta_{a+b+c=(n+1)(r-1)}\cdot \frac{1}{r^{n}}.
\eeq
By equation~\eqref{eqn:-rspin-corr(0,n)}, the quantum product defined by the deformed Theta class is given by the following: if $a+b=(r-1)m+c$, where $0\leq c\leq r-2$, then
$$
\phi_a*\phi_b=\big(\tfrac{\epsilon}{r}\big)^{m}\cdot \phi_{c+1}.
$$
Moreover, following~\cite{CGG22}, we introduce the Euler vector field for the deformed negative $r$-spin theory:
$$\textstyle
E=(r-1)\phi_{r-1}-\sum_{a=1}^{r-1}\frac{a}{r}\, \ttau^a\phi_a.
$$
Then it is proved in~\cite{CGG22} that $\Theta^{r,\epsilon}$ is a homogeneous CohFT with respect to $E$, and the conformal dimension $\delta$ is $3$.

For the deformed Theta class, the canonical basis can be explicitly computed as follows.
Let $z^\beta=(\frac{\epsilon}{r})^{\frac{1}{r-1}}\cdot e^{\frac{2\pi \beta{\bf i}}{r-1}}$,
we define a matrix $\widetilde\Psi$ by
\beq\label{eqn:Psi-r}
\widetilde\Psi_j^\beta=(z^\beta)^{j+1},\qquad \beta,j=1,\cdots, r-1.
\eeq
Then we have the inverse matrix $(\widetilde\Psi^{-1})^{j}_{\beta}=\frac{1}{r-1}\frac{1}{(z^\beta)^{j+1}}$.
We introduce
\beq\label{def:e_alpha}\textstyle
e_\beta=\sum_{j=1}^{r-1}(\widetilde\Psi^{-1})^{j}_{\beta} \phi_j,\qquad \beta=1,\cdots,r-1.
\eeq
Then by direct computation we have $e_\beta*e_\gamma=\delta_{\beta\gamma}e_\beta$ and $\eta(e_\beta,e_{\gamma})=\delta_{\beta\gamma}(\Delta_{\beta}^{r,\epsilon})^{-1}$ where
$\Delta_{\beta}^{r,\epsilon}=(\frac{\epsilon}{r})^{\frac{r+2}{r-1}}\cdot e^{\frac{2\pi {\bf i}}{r-1}\cdot (r+2)\beta}$.
This shows that  $\Theta^{r,\epsilon}$ is semi-simple for $\epsilon\ne0$ and $\{e_\beta\}_{\beta=1}^{r-1}$ gives a canonical basis.
We denote by $\{\bar e_\beta=(\Delta^{r,\epsilon}_{\beta})^{\frac{1}{2}}e_{\beta}\}_{\beta=1}^{r-1}$ the normalized canonical basis.

Since $\epsilon$-deformed Theta class $\Theta^{r,\epsilon}$ is homogeneous (with $\delta=3$) and semi-simple for $\epsilon\ne 0$,
by the Givental--Teleman reconstruction theorem (and Remark~\ref{rem:reconstruction-0}),
 it is uniquely and explicitly determined by the data $(H,\eta,*,E)$, except for the degree $3g-3$ part of $\Theta^{r,\epsilon}_{g,0}$.
More precisely, the $R$-matrix $R^{r,\epsilon}(\givz)$ and the vacuum vector $\vac^{r,\epsilon}(\givz)$ can be computed by using equation~\eqref{eqn:R-matrix-euler} and~\eqref{eqn:vacuum-euler}, respectively. See also~\cite[Theorem 3.20]{CGG22}.

Let $\Omega_{g,n}^{r,\epsilon}= R^{r,\epsilon}\cdot T^{r,\epsilon} \cdot (\oplus_{\beta=1}^{r-1}\Omega^{\rm KW_\beta})$,
where $T^{r,\epsilon}(\givz)$ is the $T$-matrix determined by $R^{r,\epsilon}$ and $\vac^{r,\epsilon}(\givz)$ by equation~\eqref{eqn:T-vacuum}, then we have
$\Theta_{g,n}^{r,\epsilon} = 	\Omega_{g,n}^{r,\epsilon}$, except for the degree $3g-3$ part of $n=0$ case.
Furthermore, according to~\cite[Theorem 4.7]{CGG22}, the CohFT $\Omega^{r,\epsilon}$ is exactly the
the one associated with the following spectral curve (see \S \ref{sec:TR-CohFT} for the definition of the CohFT associated with the spectral curve)
\beq\label{data:spectralcurve}
\cC^{r,\epsilon}=\big(\, \mathbb P^1,\quad x(z)=-z^r+\epsilon z, \quad y(z)=\tfrac{\sqrt{-r}}{z}\, \big).
\eeq	
(One can also prove this by comparing their R-matrix and $\vac$-vector directly.)
We call $\cC^{r,\epsilon}$ the $\epsilon$-deformed $r$-Bessel curve.
In the following, we will not distinguish between $\Omega^{r,\epsilon}$ and the CohFT associated with $\cC^{r,\epsilon}$.

\begin{proposition} \label{Constants}
The relation between the $\epsilon$-deformed Theta class $\Theta^{r,\epsilon}_{g,n}$ and the CohFT $\Omega^{r,\epsilon}_{g,n}$ associated with the $\epsilon$-deformed $r$-Bessel curve is given by
$$
\Theta_{g,n}^{r,\epsilon} \ = \	\Omega_{g,n}^{r,\epsilon} -  \delta_{n,0} \cdot\frac{(-1)^{g-1}}{r^{g-1}}\frac{B_{2g}}{2g(2g-2)}\frac{1}{\epsilon^{2g-2}}\cdot {\bf 1}_{g},
$$
where ${\bf 1}_g$ is the generator of $H^{6g-6}(\Mbar_{g},\mathbb Q)$ satisfying $\int_{\Mbar_{g}}{\bf 1}_g=1$.
\end{proposition}
\begin{proof}
Just notice that the deformed Theta class $\Theta^{r,\epsilon}_{g,0}$ contains no degree $3g-3$ term (because $\deg\Theta^{r,\epsilon}_{g,0}\leq \frac{(r+2)(g-1)}{r}<3g-3$), and the degree $3g-3$ term of $\Omega^{r,\epsilon}_{g,0}$ is given by $\big(\int_{\Mbar_{g}} \Omega^{r,\epsilon}_{g,0}\big) \cdot {\bf 1}_g=\omega_{g,0}^{r,\epsilon}\cdot {\bf 1}_g$,
where $\omega^{r,\epsilon}_{g,0}$ is defined by~\eqref{def:wg0} for the deformed $r$-Bessel curve $\cC^{r,\epsilon}$.
The Proposition follows from the formula $\omega_{g,0}^{r,\epsilon}= \frac{(-1)^{g-1}}{r^{g-1}}\frac{B_{2g}}{2g(2g-2)}\frac{1}{\epsilon^{2g-2}}$,
and we prove this in Appendix~\ref{sec:omegag0}.
\end{proof}

\subsection{Correspondence between geometric descendents and TR descendents}\label{sec:geo-TR}
We prove the TR-Geo correspondence for the deformed negative $r$-spin theory.
For the geometric side, we have the following reconstruction formula of the descendent theory:
\begin{theorem}  \label{thm:descendentrspin}
	The total descendent potential $\cD^{r,\epsilon}({\bf t};\hbar)$ of the $\epsilon$-deformed negative $r$-spin theory can be reconstructed by the following Kontsevich--Manin type formula:
	$$
	\cD^{r,\epsilon}({\bf t};\hbar) =e^{F_{\rm un}^{r,\epsilon}({\bf t};\hbar)} \cdot \cA^{r,\epsilon}({\bf s(t)};\hbar),
	$$
where $F_{\rm un}^{r,\epsilon}({\bf t};\hbar) =\frac{1}{\hbar^2}\sum_{k,a}\<\phi_a\psi^{k}\>^{r,\epsilon}_{0,1}t_k^a
+\frac{1}{2\hbar^2}\sum_{k,l,a,b}\<\phi_a\psi^{k},\phi_b\psi^l\>^{r,\epsilon}_{0,2}t_k^a t_l^b$,
$\cA^{r,\epsilon}({\bf s};\hbar)$ is the total ancestor potential of $\Theta^{r,\epsilon}$,
and the coordinate transformation ${\bf s}={\bf s(t)}$ is given by ${\bf s}(\givz)=[S^{r,\epsilon}(\givz){\bf t}(\givz)]_{+}$.
Here $S^{r,\epsilon}(\givz)$ is the $S$-matrix of the deformed negative $r$-spin theory defined by
	$$
	\eta(\phi_a,S^{r, \epsilon}(\givz)\phi_b):=\eta(\phi_a,\phi_b)+\big\<\phi_a,\tfrac{\phi_b}{\givz-\psi}\big\>_{0,2}^{r,\epsilon}.
	$$
Moreover, let $J^{r, \epsilon}(\givz):=-r\phi_{r-1}+\<\frac{\phi_a}{\givz-\psi}\>_{0,1}^{r,\epsilon}\phi^a$ be the $J$-function of the deformed negative $r$-spin theory, then it is uniquely determined by the following Picard-Fuchs equation
$$
	(-1+\epsilon\pd_{\epsilon}-r\givz^{r-1}\pd_{\epsilon}^{r})J^{r,\epsilon}(\givz)=0,
$$
with initial condition $[\givz^{-i}]J^{r, \epsilon}(\givz)=\frac{\epsilon^{i+1}}{(i+1)!}\phi_i$, $i=1,\cdots,r-1$,
and $S^{r,\epsilon}(\givz)$ is determined by
$$
	S^{r,\epsilon,*}(\givz)\phi_i=\givz^i\pd^{i+1}_{\epsilon}J^{r, \epsilon}(\givz).
$$
\end{theorem}
\noindent The proof of this Theorem is given in the Appendix~\ref{sec:des-r}.

For the TR side, we have the following Theorem.
\begin{theorem}
Let $\tilde \gamma_k=e^{\frac{2k\pi {\bf i}}{r}}[0,\infty)\subset \mathbb C$ and $\gamma_k=\tilde\gamma_0-\tilde \gamma_k$, we have
\begin{align}
\eta(J^{r,\epsilon}(\givz),\Phi(\gamma_k,\givz))=&\int_{\gamma_k}e^{x(z)/\givz}y(z)dx(z),\label{eqn:int-J-r}\\
\eta(\phi^i,S^{r,\epsilon}(\givz)\Phi(\gamma_k,\givz))=&\int_{\gamma_k}e^{x(z)/\givz}\zeta^i(z)dx(z),\label{eqn:int-S-r}
\end{align}
where
$$\Phi(\gamma_k,\givz)= \frac{\sqrt{-r}}{r}\sum_{a=1}^{r-1}  \big(1-e^{\frac{2ka\pi{\bf i}}{r}}\big) \Gamma\Big(\frac{a}{r}\Big)\givz^{\frac{a}{r}} \phi^a.$$
\end{theorem}
\begin{proof}
Clearly, for $k=1,\cdots,r-1$, $\gamma_k$ are admissible paths associated with $e^{x(z)/\givz}$.
We consider the integral
$$\textstyle
\mathcal I^{r,\epsilon}_k(\givz):= \int_{\gamma_k} e^{x(z)/\givz} y(z)dx(z).
$$
Then by admissible condition,
$\mathcal I^{r,\epsilon}_k(\givz)= -\givz \int_{\gamma_k} e^{\frac{-z^r+\epsilon \cdot z}{\givz}} dy(z)$ and it satisfies the Picard-Fuchs equation
$$
\big(-1+\epsilon\pd_{\epsilon} - r \givz^{r-1} \partial_{\epsilon} ^{r} \big)\,  \mathcal I^{r,\epsilon}_k(\givz) =0.
$$
Now we compute the Taylor expansion of $\mathcal I^{r,\epsilon}_k(\givz)$:
$$\textstyle
\partial_\epsilon^{m+1} \mathcal I^{r,\epsilon}_k(\givz) |_{\epsilon =0}
=\givz^{\frac{m}{r}-m}\frac{\sqrt{-r}}{r}\int_{\gamma_k} z^{m-r}\, e^{- {z^r }}  dz^r
= \givz^{\frac{(1-r)m}{r} } \frac{\sqrt{-r}}{r} \big(1-e^{\frac{ 2km\pi {\bf i} }{r}}\big) \Gamma\big(\frac{m}{r}\big),
$$
where we have used $\Gamma(\alpha)=\int_{\tilde \gamma_0} x^{\alpha-1} e^{-x} dx   $. Hence we obtain
$$\textstyle
\mathcal I^{r,\epsilon}_k(\givz)  = \frac{\sqrt{-r}}{r} \sum_{a=1}^{r-1}  \big(1-e^{\frac{ 2ka\pi  {\bf i} }{r}}\big)  \Gamma\big(\frac{a}{r}\big)\givz^{\frac{a}{r}}  I_a(\givz),
$$
where
$$\textstyle
I_a(\givz)=-r\cdot\delta_{a,r-1}+\sum_{n\geq 0} \frac{\Gamma(\frac{rn+a }{r})}{\Gamma(\frac{a}{r})} \, \frac{\epsilon^{rn+a+1}}{(rn+a+1)!} \, \givz^{-n(r-1)-a}.
$$
It is clear that $J^{r,\epsilon}(\givz)=\sum_{a=1}^{r-1}I_a(\givz)\phi_a$ since the right-hand side solves the Picard-Fuchs equation and satisfies the initial condition as the $J^{r,\epsilon}(\givz)$ does, this proves equation~\eqref{eqn:int-J-r}.

Furthermore, notice that on one hand,
$$\textstyle
\eta(\phi_i,S^{r,\epsilon}(\givz)\Phi(\gamma_k,\givz))
=\givz^i\partial_\epsilon^{i+1} \mathcal I^{r,\epsilon}_k(\givz)
= -\int_{\gamma_k} e^{\frac{-z^r+\epsilon \cdot z}{\givz}}z^{i+1} dy(z)
=\sqrt{-r}\int_{\gamma_k} e^{\frac{-z^r+\epsilon \cdot z}{\givz}} \frac{z^{i-1}}{x'(z)} dx(z),
$$
and on the other hand, by the definition equation~\eqref{def:zeta}, it is easy to see $\zeta^{\beta}=\sqrt{-r}\frac{(z^{\beta})^2}{z-z^\beta}$. Then by using the $\widetilde\Psi$-matrix~\eqref{eqn:Psi-r} we have $\zeta^i=\sqrt{-r}\frac{z^{r-i-1}}{x'(z)}$.
This proves equation~\eqref{eqn:int-S-r}.
\end{proof}
We note here that it is easy to see the matrix $(1-e^{\frac{ 2ka\pi  {\bf i} }{r}})_{k,a=1,\cdots,r-1 }$ is invertable, in fact, it has determinant $\prod_{a=1}^{r-1}(1-e^{\frac{ 2ka\pi  {\bf i} }{r}})\prod_{1\leq a< b\leq r-1}(e^{\frac{ 2kb\pi  {\bf i} }{r}}-e^{\frac{ 2ka\pi  {\bf i} }{r}})\ne 0$.
Hence equation~\eqref{eqn:des-EO}, \eqref{eqn:des-EO-(0,2)} and ~\eqref{eqn:int-J-r} determine all the descendent invariants of the deformed negative $r$-spin theory.
\subsection{Proof of the deformed negative $r$-spin Witten conjecture}
We choose the local coordinate $\lambda$ to be the the solution of $x(z)=-\lambda^{r}$ satisfying $\lim_{z\to \infty}\lambda/z=1$.
By considering the expansion of $\omega_{g,n}$ given by equation~\eqref{eqn:stable-omega-gn-boundary},
we get the correlators $\<-\>_{g,n}^{\lambda}$ and generating series $Z({\bf p};\hbar)$ given by equation~\eqref{eqn:generating-series-m-KP-F-A-TR}.
To specify the case of negative $r$-spin, we denote the generating series by
$Z^{r,\epsilon}({\bf p};\hbar)$.
By Theorem~\ref{KP-TR}, $Z^{r,\epsilon}({\bf p};\hbar)$ is a tau-function of KP hierarchy,
and by Corollary~\ref{cor:red-Z-x}, we know that $Z^{r,\epsilon}({\bf p};\hbar)$ does not depend on $p_{rm}$, $m\geq 1$, thus $Z^{r,\epsilon}({\bf p};\hbar)$ is a tau-function of $r$KdV hierarchy.

Notice that $\chi^i=\int_{z=\infty}^{z}d\chi_0^i(z)$ (see equation~\eqref{eqn:chi-zeta}) satisfies $\eta(\phi^i,\Phi(\gamma_k,\givz))=\int_{\gamma_k}e^{x(z)/\givz}\chi^{i}dx(z)$, and $\Res_{\lambda\to \infty}x(z)^{n}d\chi^i=0$ for $n\geq 0$, this uniquely determines $\chi^i=-\frac{1}{\sqrt{-r}}\lambda^{-i}$, $i=1,\cdots,r-1$.
Therefore, for $n\geq 0$,
$$
\chi_n^i=-\frac{(-1)^{n}}{\sqrt{-r}}\frac{\Gamma(n+\frac{i}{r})}{\Gamma(\frac{i}{r})}\lambda^{-i-rn}.
$$
This gives the coordinate transformation~\eqref{eqn:wittentime} according to the definition equation~\eqref{eqn:t=tp}.
It is easy to check equation~\eqref{eqn:vanif2} for this case, thus by Theorem~\ref{thm:D-Z} we know that
$$
e^{\sum_{g\geq 2}\frac{\hbar^{2g-2}}{(\sqrt{-r}\epsilon)^{2g-2}}\frac{B_{2g}}{2g(2g-2)}}\cdot\cD^{r,\epsilon}(\hbar\cdot{\bf t};\hbar)|_{\bf t=t(p)}
=e^{\frac{1}{\hbar}\sum_{k,a}\<\phi_a\psi^{k}\>^{r,\epsilon}_{0,1}t_k^a}
\cdot Z^{r,\epsilon}({\bf p};\hbar),
$$
where the factor on the left-hand side of above equation comes from the difference between $\cD^{r,\epsilon}({\bf t};\hbar)$ with the total descendent potential defined by topological recursion.
By Theorem~\ref{KP-TR}, the deformed negative $r$-spin Witten conjecture (Theorem~\ref{conj:negative-witten-conjecture})  is proved.

\begin{remark}
The original negative $r$-spin Witten conjecture was also proved independently by Alexandrov, Bychkov, Dunin-Barkowski, Kazarian, and Shadrin~\cite{ABDKS23}.
\end{remark}

\section{KdV integrability for the Hurwitz space $M_{1,1}$ via the Weierstrass curve}
\label{sec:TR-weierstrass}

In this section we first review the identification of the CohFT associated with the Hurwitz space $M_{1,1}$ and that associated with the Weierstrass curve $\cC=(\Sigma_{\tau},x,y)$, where $\Sigma_{\tau}$ is the elliptic curve with moduli parameter $\tau$, and
$$
x(z)=\tfrac{\epsilon^2}{4\pi {\bf i}}\cdot \big(\wp(z,\elltau)+G_2(\elltau)\big),\qquad
y(z)=\sqrt{2\pi {\bf i}}\cdot z.
$$
Here $\wp$ is the Weierstrass P-function, and $G_2$ is the weight-$2$ holomorphic Eisenstein series.
The choice of $y(z)$ is made to make the Frobenius structure concise, see \S \ref{sec:CohFT-Weierstrass}.
Next, we establish the correspondence between the descendant invariants of the Hurwitz space
$M_{1,1}$ and those of the Weierstrass curve, proving the first part of Theorem~\ref{thm:elliptic-KdV}.
Finally, we demonstrate that the resulting non-perturbative total descendant potential is a tau-function of the KdV hierarchy, completing the proof of the second part of Theorem~\ref{thm:elliptic-KdV}.

\subsection{Frobenius manifold and CohFT associated with the Weierstrass curve} \label{sec:CohFT-Weierstrass}
In \S \ref{sec:TR}, we have introduced how a CohFT is constructed from a spectral curve, in this section, we show details for the Weierstrass curve.
Before giving the explicit computations, we introduce the following deformation $\cC^{\ttau}=(\Sigma_{\tau},x_{\ttau},y_{\ttau})$ of $\cC$ with parameters $\{\ttau^i\}_{i=0}^2$,  where
$$
x_{\ttau}(z)=\tfrac{(\epsilon+\ttau^1)^2}{4\pi {\bf i}}\cdot \big(\wp(z,\elltau+\ttau^2)+G_2(\elltau+\ttau^2)\big)+\ttau^0,\qquad
y_{\ttau}(z)=\sqrt{2\pi {\bf i}}\cdot z.
$$
Clearly, $\cC=\cC^{\ttau}|_{\ttau=0}$.
We will prove the parameter space defines a Frobenius manifold with flat coordinate $\ttau$ and we consider $\Frob = \sspan\{\phi_i=\pd_{\ttau^i}\}_{i=0}^{2}$ as the state space
~\footnote{In this section, we slightly shift the superscripts of the flat coordinate $\ttau^{i}$ and the subscripts of the flat basis $\{\phi_i=\pd_{\ttau^i}\}$ such that $i$ is start from $0$.}.

\subsubsection{Critical points computations}
It is well-known that the  Weierstrass P-function  $\wp$ satisfies
 the differential equation
 $$\textstyle
\wp'^2=4\wp^3-60G_4\wp-140G_6=4\prod_{\beta=1}^{3}(\wp-u^\beta),
$$
where for $k\geq 1$, $G_{2k}(\elltau)=\sum_{m,n\in\mathbb Z^2\setminus \{(0,0)\}}\frac{1}{(m+n\elltau)^{2k}}$ is the holomorphic Eisenstein series of weight $2k$,
and $u^\beta$, $\beta=1,2,3$, satisfies
$$
u^1+u^2+u^3=0,\qquad
u^1u^2+u^1u^3+u^2u^3=-15G_4,\qquad
u^1u^2u^3=35G_6.
$$
Thus the critical points of $x_{\ttau}(z)$ are  given by
$$\textstyle
z^1=\frac{1}{2},\qquad
z^2=\frac{\elltau+\ttau^2}{2},\qquad
z^3=\frac{1+\elltau+\ttau^2}{2},
$$
with critical values $x^\beta=x_{\ttau}(z^\beta)=\frac{(\epsilon+\ttau^1)^2}{4\pi {\bf i}}\big(u^\beta(\tau+\ttau^2)+G_2(\tau+\ttau^2))+\ttau^0$, $\beta=1,2,3$.

The critical values $\{x^\beta\}$ can be considered as the canonical coordinates, such that the canonical basis   is given by $e_\beta=\frac{\pd}{\pd x^\beta}$.
The quantum product is   $e_\beta*_{\ttau}e_{\gamma}=\delta_{\beta\gamma}e_\beta$,
and symmetric bilinear form is determined by
$$\textstyle
\Delta_\beta^{-1}=\eta(e_\beta,e_\beta)=\frac{y'_{\ttau}(z^\beta)^2}{x''_{\ttau}(z^\beta)}=\frac{(4\pi {\bf i})^2}{(\epsilon+\ttau^1)^2\cdot (12 u^\beta(\tau+\ttau^2)^2-60G_4(\tau+\ttau^2))}.
$$
In the following, we will use the following notations:
$$
\tilde t^1=t^1+\epsilon,\qquad
G_{2k}=G_{2k}(\tau+\ttau^2),\qquad
u^\beta=u^\beta(\ttau+\ttau^2).
$$
\subsubsection{The $\widetilde\Psi$ matrix}
We now compute the transformation matrix between the canonical basis and the basis $\{\phi_i\}$.
By the Ramanujan identities of Eisenstein series,
$$
\pd_{\elltau}G_2=\frac{5G_4-G_2^2}{2\pi {\bf i}},\qquad
\pd_{\elltau}G_4=\frac{7G_6-2G_2G_4}{\pi{\bf i}},\qquad
\pd_{\elltau}G_6=\frac{30G_4^2-21G_2G_6}{7\pi{\bf i}},
$$
one can compute $\widetilde\Psi^\beta_i:=\frac{\pd x^\beta}{\pd \ttau^i}$, $i=0,1,2$:
$$
(\widetilde\Psi^\beta_i)=\left(\begin{array}{ccc}
1 & \frac{\tilde\ttau^1}{2\pi {\bf i}}(G_2+u^1) & \frac{(\tilde\ttau^1)^2}{4\pi {\bf i}}\big(\pd_{\elltau}(G_2)+\frac{15u^1\pd_{\elltau}(G_4)+35\pd_{\elltau}(G_6)}{3(u^1)^2-15G_4}\big)\\
1 & \frac{\tilde\ttau^1}{2\pi {\bf i}}(G_2+u^2) & \frac{(\tilde\ttau^1)^2}{4\pi {\bf i}}\big(\pd_{\elltau}(G_2)+\frac{15u^2\pd_{\elltau}(G_4)+35\pd_{\elltau}(G_6)}{3(u^2)^2-15G_4}\big) \\
1 & \frac{\tilde\ttau^1}{2\pi {\bf i}}(G_2+u^3) & \frac{(\tilde\ttau^1)^2}{4\pi {\bf i}}\big(\pd_{\elltau}(G_2)+\frac{15u^3\pd_{\elltau}(G_4)+35\pd_{\elltau}(G_6)}{3(u^3)^2-15G_4}\big)
\end{array}\right).
$$

\subsubsection{The symmetric $2$-form, quantum product, and the Euler vector field}
By the change of coordinate, one can compute the symmetric $2$-form of $\phi_i=\pd_{\ttau^i}$,
which gives
  $$
\eta(\phi_i,\phi_j)=\delta_{i+j,2}.
$$
This means $\ttau$ is a flat coordinate and $\{\phi_i\}_{i=0}^{2}$ is the corresponding flat basis.

By the transformation matrix $\widetilde\Psi$ and the quantum product $e_\beta*_{\ttau}e_\gamma=\delta_{\beta\gamma}e_\beta$,
one can compute the quantum product of $\phi_i$ under this flat basis and obtain $\phi_0*_{\ttau}\phi_i=\phi_i$,
\begin{align*}
\phi_1*_{\ttau}\phi_1=&\textstyle\,  \phi_2+\frac{3}{2\pi {\bf i}}\cdot\frac{\tilde\ttau^1}{1!}\cdot G_2\cdot \phi_1
+\frac{3}{2\pi {\bf i}}\cdot\frac{(\tilde\ttau^1)^2}{2!}\cdot \pd_{\elltau}(G_2) \cdot \phi_0, \\
\phi_1*_{\ttau}\phi_2=&\textstyle\,   \frac{3}{2\pi {\bf i}}\cdot\frac{(\tilde\ttau^1)^2}{2!}\cdot \pd_{\elltau}(G_2) \cdot \phi_1
+\frac{3}{2\pi {\bf i}}\cdot\frac{(\tilde\ttau^1)^3}{3!}\cdot \pd_{\elltau}^2(G_2)\cdot\phi_0 ,\\
\phi_2*_{\ttau}\phi_2=&\textstyle\,  \frac{3}{2\pi {\bf i}}\cdot\frac{(\tilde\ttau^1)^3}{3!}\cdot \pd_{\elltau}^2(G_2)\cdot\phi_1
+\frac{3}{2\pi {\bf i}}\cdot\frac{(\tilde\ttau^1)^4}{4!}\cdot \pd_{\elltau}^3(G_2)\cdot\phi_0 .
\end{align*}
Particularly, we see this Frobenius manifold has a flat unit ${\bf 1}=\phi_0$.

Furthermore, let $E=\sum_\beta x^\beta\pd_{x^\beta}$,  then one has
$$
E=\ttau^0\frac{\pd }{\pd \ttau^0}+\frac{1}{2}(\ttau^1+\epsilon)\frac{\pd }{\pd \ttau^1}.
$$
It is easy to check that $[E,\phi_i]=(\frac{i}{2}-1)\phi_i$ and $[E,\phi_i*_{\ttau}\phi_j]=(\frac{i+j}{2}-1)\phi_i*_{\ttau}\phi_j$,
which implies that $E$ satisfies equation~\eqref{eqn:Euler-qp} and \eqref{eqn:Euler-eta} with $\delta=1$.
In other words, the Frobenius manifold $\Frob$ is homogeneous of conformal dimension $1$ with respect to $E$ and the grading operator $\mu$ is given by $\mu(\phi_i)=(\frac{i}{2}-\frac{1}{2})\phi_i$, $i=0,1,2$.

This Frobenius manifold coincides with the Frobenius manifold structure on the Hurwitz space $M_{1,1}$ introduced by Dubrovin in~\cite[Example 5.6]{Dub96}.

\subsubsection{CohFT}
By Givental's reconstruction procedure, starting from a homogeneous semi-simple Frobenius manifold containing a flat unit,
one can consider the QDE:
\beq\label{eqn:QDE-elliptic}
\givz\pd_{\ttau^i}\mathcal S(\givz)=\phi_{i}*_{\ttau}\mathcal S(\givz).
\eeq
The QDE has a fundamental solution in the form $\Psi^{-1}R^{\ttau}(\givz)e^{X/\givz}$, where $R^{\ttau}(\givz)$ is a formal matrix valued power series of form $R^{\ttau}(\givz)=\id +\sum_{k\geq 1}R^{\ttau}_k\givz^k$ and satisfies $R^{\ttau,*}(-\givz)R^{\ttau}(\givz)=\id$.
Moreover, the matrix $R^\ttau(\givz)$ is uniquely determined by the Euler equation~\cite{Giv01a,Dub96}:
\beq\label{eqn:euler-R}
(\givz\pd_{\givz}+\sum x^\beta\pd_{x^\beta})R^{\ttau}(\givz)=0.
\eeq
On the other hand,  for the CohFT associated with curve $\cC^{\ttau}$, it is proved in~\cite{DNOPS19} that its $R$-matrix $\tilde R^{\ttau}(\givz)$ defined by the topological recursion~\eqref{def:EO-R} satisfies the Euler equation~\eqref{eqn:euler-R}, and thus $\tilde R^{\ttau}(\givz)=R^{\ttau}(\givz)$.
This proves that the CohFT for the topological recursion coincide with the one reconstructed from the Frobenius manifold structure on the Hurwitz space $M_{1,1}$ via Givental's reconstruction procedure.
We denote by $\cA^{\ttau}({\bf s};\hbar)$ the total ancestor potential for this CohFT.

\subsection{Correspondence between geometric descendents and TR descendents}\label{sec:TR-Geo-Weierstrass}  \label{Ge-TRforellipticcurve}
In this subsection, we prove the first part of Theorem~\ref{thm:elliptic-KdV}, establishing the relation of geometric descendents with the TR descendents.

We start with the definition of the descendent invariant of the Hurwitz space $M_{1,1}$.
We consider the fundamental solution of QDE~\eqref{eqn:QDE-elliptic} in the form
$$\textstyle
S^{\ttau}(\givz)=\id+\sum_{k\geq 1}S^{\ttau}_k\givz^{-k},
$$
satisfying $S^{\ttau,*}(-\givz)S^{\ttau}(\givz)=\id$.
By assuming homogeneity condition~\cite{Giv01a}
\beq\label{eqn:S-euler-elliptic}
(\givz\pd_\givz + E)S^{\ttau}(\givz)=[S^{\ttau}(\givz),\mu],
\eeq
the solution $S^{\ttau}(\givz)$ is uniquely determined up to some constants.
For this case, there is only one constant that needs to be fixed.
Precisely, by solving QDE $\pd_{\ttau^i}S^{\ttau}_1=\phi_i*$,
and assuming $E(S^{\ttau}_1)=S^{\ttau}_1+[S^{\ttau}_1,\mu]$, we have
$$
S^{\ttau}_1=\left(\begin{array}{ccc}
\ttau^0 &  \frac{3}{2\pi {\bf i}}\cdot\frac{(\tilde\ttau^1)^3}{3!}\cdot \pd_{\elltau}(G_2) & \frac{3}{2\pi {\bf i}}\cdot\frac{(\tilde\ttau^1)^4}{4!}\cdot \pd_{\elltau}^2(G_2) \\
\tilde\ttau^1 & \ttau^0+\frac{3}{2\pi {\bf i}}\cdot\frac{(\tilde\ttau^1)^2}{2!}G_2 &  \frac{3}{2\pi {\bf i}}\cdot\frac{(\tilde\ttau^1)^3}{3!}\cdot \pd_{\elltau}(G_2)  \\
\elltau+\ttau^2+c & \tilde\ttau^1 & \ttau^0
\end{array}\right),
$$
with one undetermined constant $c$.
For simplicity, we take the constant $c$ to be $0$.
For $k\geq 2$, by using QDE~\eqref{eqn:QDE-elliptic}, equation~\eqref{eqn:S-euler-elliptic} gives
\beq\label{eqn:euler-S}\textstyle
(k+\mu_j-\mu_i)(S^{\ttau}_k)^i_j=\sum_{a=0}^{2}(E*_{\ttau})^i_a(S^{\ttau}_{k-1})^a_j,
\eeq
notice that $k+\mu_j-\mu_i\geq k+(-\frac{1}{2})-\frac{1}{2}>0$, $(S^{\ttau}_k)^i_j$ is uniquely determined by induction.
The $J$-function $J^{\ttau}(\givz)$ is given by $J^{\ttau}(\givz)=\givz S^{\ttau,*}(\givz){\bf 1}$.
We get the total descendent potential $\cD({\bf t};\hbar)$ by the Kontsevich--Manin formula~\eqref{eq:descendent-ancestor}.

Now we are ready to prove the first part of Theorem~\ref{thm:elliptic-KdV}.
\begin{proof}[Proof of the first part of Theorem~\ref{thm:elliptic-KdV}]
Let $\lambda$ be the local coordinate near the boundary $z=0$ defined by $x_{\ttau}(z)=\frac{\lambda^2}{2}$ and satisfies $\lim_{z\to 0}\lambda\cdot z=-\frac{\tilde\ttau^1}{\sqrt{2\pi{\bf i}}}$. We define
$$
\chi^i=\delta_{i,1}\lambda^{-1},\qquad
\chi^i_k:=D^k\chi^i=\delta_{i,1}(2k-1)!!\lambda^{-2k-1},\quad k\geq 0.
$$
Now we consider the expansion of $d\zeta^{\bar \beta}$ near the boundary, by direct computation we have (there is no $\lambda^{-2k}$ terms in the expansion because $\Res_{z=0} x_{\ttau}(z)^{k}d\zeta^{\bar\beta}(z)=0$, $\forall k\geq 0$)
\beq\label{eqn:zeta-chi-wei}
d\zeta^{\bar\beta}=\sum_{k\geq 0}\Psi^{\bar\beta}_i(\tilde S^{\ttau}_k)^{i}_1d\chi^1_k
=\Psi^{\bar\beta}_1d\chi_0^1+\sum_{k\geq 1}\Psi^{\bar\beta}_i(\tilde S^{\ttau}_k)^{i}_1d\chi^1_k.
\eeq
To determine $\tilde S^{\ttau}_k$ for $k\geq 1$, we need the following equation:
\beq \label{eqn:Dzeta-Psi-d}\textstyle
(x-x^\beta)dD\zeta^{\bar\beta}(z)=\sum_{\gamma}\big(\mu+\tfrac{3}{2}\big)^{\bar\beta}_{\bar\gamma}\cdot d\zeta^{\bar\gamma}(z),
\qquad \beta=1,\cdots,N.
\eeq
The proof of this equation is straightforward since all the terms are explicit, we omit the details here.
By this equation, and the expansion of $d\zeta^{\bar\beta}$, we have
$$
\frac{\lambda^2}{2}\sum_{k\geq 0}\Psi^{\bar\beta}_i(\tilde S^{\ttau}_k)^{i}_1d\chi^1_{k+1}
-\sum_{k\geq 0}\sum_{i}\Psi^{\bar\beta}_i(\mu_i+\tfrac{1}{2})(\tilde S^{\ttau}_k)^{i}_1d\chi^1_k
=x^\beta\sum_{k\geq 0}\Psi^{\bar\beta}_i(\tilde S^{\ttau}_k)^{i}_1d\chi^1_{k+1}.
$$
By comparing the coefficient of $\lambda^{-2k-1}$, we get
$$
\Psi^{\bar\beta}_i(k+\mu_i)(\tilde S^{\ttau}_k)^{i}_1 =\Psi^{\bar\beta}_i(E*)^{i}_{j}(\tilde S^{\ttau}_{k-1})^{j}_1.
$$
This coincides with equation~\eqref{eqn:euler-S} and thus $(\tilde S^{\ttau}_n)^{\bar\beta}_1=(S^{\ttau}_n)^{\bar\beta}_1$.
The first part of the Theorem~\ref{thm:elliptic-KdV} follows immediately.
\end{proof}
This result can be also written in the form of Laplace transform.
We assume $\givz~\in~\frac{(\ttau^1)^2}{2\pi{\bf i}}(0,\infty)$,
then $\gamma=[0,1)$ is an admissible path associated with $e^{-x_{\ttau}(z)/\givz}$ and we have
$$
\givz\int_{\gamma}e^{-x_{\ttau}(z)/\givz}d\chi_k^{1}=\int_{\gamma}e^{-x_{\ttau}(z)/\givz}\chi_k^{1}dx=\sqrt{2\pi\givz}\cdot (-\givz)^{-k}.
$$
By equation~\eqref{eqn:zeta-chi-wei}, and by using $(\tilde S^{\ttau}_n)^{\bar\beta}_1=(S^{\ttau}_n)^{\bar\beta}_1$,
$$
\givz\int_{\gamma}e^{-x_{\ttau}(z)/\givz}d\zeta^{\bar\beta}=\eta(\bar e_\beta,S^{\ttau}(-\givz)\Phi(\gamma,-\givz)),
$$
where $\Phi(\gamma,-\givz)=\sqrt{2\pi\givz}\cdot\phi_1$.
Furthermore, notice that
$$
\int_{\gamma} e^{-x_{\ttau}(z)/\givz} y_{\ttau}dx_{\ttau}
=\givz\int_{\gamma} e^{-x_{\ttau}(z)/\givz} dy_{\ttau}
=\givz\int_{\gamma} e^{-x_{\ttau}(z)/\givz}\cdot \frac{\sqrt{2\pi{\bf i}}}{x'_{\ttau}(z)}\cdot dx_{\ttau}
=\givz^2\int_{\gamma} e^{-x_{\ttau}(z)/\givz}d\bigg(\frac{\sqrt{2\pi{\bf i}}}{x'_{\ttau}(z)}\bigg),
$$
and
$$
\sum_{\beta}\Delta_{\beta}^{-\frac{1}{2}}d\zeta^{\bar \beta}
=-\sum_{\beta}\frac{y'_{\ttau}(z^\beta)}{x''_{\ttau}(z^\beta)}\mathop{\Res}_{z'=z^\beta}\frac{B(z',z)}{z'-z^\beta}
=-\sqrt{2\pi{\bf i}}\sum_{\beta}\mathop{\Res}_{z'=z^\beta}\frac{B(z',z)}{x'_{\ttau}(z')}
=d\bigg(\frac{\sqrt{2\pi{\bf i}}}{x'_{\ttau}(z)}\bigg).
$$
This proves
$$
\int_{\gamma} e^{-x_{\ttau}(z)/\givz} y_{\ttau}dx_{\ttau}=\givz\cdot \eta(\phi_0,  S^{\ttau}(-\givz)\Phi(\gamma,-\givz))=-\eta(J^{\ttau}(-\givz),\Phi(\gamma,-\givz)).
$$
This verifies the Conjecture~\ref{conj:laplace-J} for the Weierstrass curve case.

\subsection{KdV integrability for the descendent theory of the Hurwitz space $M_{1,1}$}
Now we show the non-perturbative modification $\cD^{\rm NP}_{\mu,\nu}({\bf t};\hbar;w)$ of the potential $\cD({\bf t};\hbar)$
and prove the second part of the Theorem~\ref{thm:elliptic-KdV}.
As the integrability does not depend on $\ttau$, we consider the theory at $\ttau=0$, and the corresponding spectral curve is the Weierstrass curve $\cC$ itself.

For this case, we have $\Sigma_{\elltau}=\mathbb C/(\mathbb Z+\elltau\mathbb Z)$ where $\elltau\in\mathbb H$,  the upper half plane of $\mathbb C$, the $A$, $B$-cycles are taken to be line segment connecting $0$ with $1$ and $\elltau$, respectively.
Furthermore,
$$
B(z_1,z_2)=\big(\wp(z_1-z_2,\elltau)+G_2(\elltau)\big) dz_1dz_2,
$$
and we have
$$
\oint_{z'\in A}B(z,z')=0,\qquad \oint_{z'\in B}B(z,z')=2\pi {\bf i}\, dz.
$$
Now we have $\sv=-2\pi {\bf i}\sum_{\beta}\frac{dz}{dy}|_{z=z^\beta}\cdot e_\beta=-\sqrt{2\pi {\bf i}}\, \phi_0$, and
by solving $x(z)=\frac{\lambda^2}{2}$ near the boundary point $z=0$ where $\lambda$ is chosen to satisfy $\lim_{z\to 0}\lambda\cdot z=-\frac{\epsilon}{\sqrt{2\pi{\bf i}}}$, we have
$$
\sqrt{2\pi {\bf i}}\, z=-\epsilon\, \lambda^{-1}-\tfrac{\epsilon^3 G_2}{4\pi {\bf i}}\, \lambda^{-3}+\cdots.
$$
These computations gives us the explicit formulae for non-perturbative functions
\begin{align*}
\cA^{\rm NP}_{\mu,\nu,\tau}({\bf s};\hbar;w)=&\, e^{-\frac{\hbar}{\sqrt{2\pi {\bf i}}}\pd_{s_0^0}\cdot \pd_w}(\cA({\bf s};\hbar)\cdot \theta[\cmn](w|\tau))\\
\cD^{\rm NP}_{\mu,\nu}({\bf t};\hbar;w)=&\,
e^{F_{\rm un}({\bf t};\hbar)+\frac{1}{\hbar}\<\sv\>_{0,1}\cdot \nabla_{w}+\frac{1}{\hbar}W({\bf t},\sv\cdot \nabla_w)}\big(\cA^{\rm NP}_{\mu,\nu,\tau}({\bf s(t)};\hbar;w)\big),\\
Z^{\rm NP}_{\mu,\nu}({\bf p};\hbar;w)=&\, e^{\frac{1}{2}Q({\bf p,p})}\cdot e^{\widehat z\pd_w}
\big(\cA^{\rm NP}_{\mu,\nu,\tau}(\hbar\cdot {\bf s(p)};\hbar;w)\big).
\end{align*}
Furthermore, by Proposition~\ref{prop:Q=W}, $Q({\bf p,p})=W({\bf t,t})|_{\bf t=t(p)}$, we have equation:
$$
\cD^{\rm NP}_{\mu,\nu}({\bf t};\hbar;w)|_{\bf t=t(p)}
=e^{\frac{1}{\hbar}\<{\bf t}(\psi)\>_{0,1}|_{\bf t=t(p)}}\cdot e^{\frac{1}{\hbar}\<\sv\>_{0,1}\nabla_{w}} Z^{\rm NP}_{\mu,\nu}({\bf p};\hbar;w),
$$
where  ${\bf t}={\bf t}({\bf p})$ is given by $t^i_k=\delta_{i,1}(2k-1)!!p_{2k+1}$.
Put all these together, we obtain the proof of the second part of the Theorem~\ref{thm:elliptic-KdV}.
\begin{proof}[Proof of the second part of Theorem~\ref{thm:elliptic-KdV}]
Fix the choice of the local coordinate $\lambda$ defined by $x(z)=\frac{\lambda^2}{2}$. Then by Theorem~\ref{KP-TR} and by Corollary~\ref{cor:red-Z-x}, we know that $Z^{\rm NP}_{\mu,\nu}({\bf p};\hbar;w)$ is a family of tau-functions of KdV hierarchy.
As the exponential of the linear function $\<{\bf t}(\psi)\>_{0,1}|_{\bf t=t(p)}$ will not change the KP integrability, we see
 $\cD^{\rm NP}_{\mu,\nu}({\bf t};\hbar;w)|_{\bf t=t(p)}$ gives a family of tau-functions of the KdV hierarchy with parameters $\mu,\nu$, and $w$.
\end{proof}

Now we study $\cD^{\rm NP}_{\mu,\nu}(\hbar\cdot{\bf t(p)};\hbar;w)$ from the perspective of the KdV hierarchy.
Let
$$
U({\bf p};\hbar):=\pd_{p_1}^2\log\big(\cD^{\rm NP}_{\mu,\nu}(\hbar\cdot {\bf t(p)};\hbar;w)\big),
$$
then $U({\bf p};\hbar)$ gives a solution to the KdV equations.
Moreover, it is well known that $U({\bf p};\hbar)$ is uniquely determined by the initial condition $U(p_1;\hbar)$.
\begin{conjecture}\label{conj:initial-u-kdv}
The initial condition $U(p_1;\hbar)$ of the solution $U({\bf p};\hbar)$ to the KdV equations is given by the following formula:
$$\textstyle
U(p_1;\hbar)=\frac{3 G_2\cdot (\epsilon+\hbar\cdot p_1)^2}{4\pi {\bf i}}+\frac{\hbar^2}{8(\epsilon+\hbar\cdot p_1)^2}
+\pd_{p_1}^2\log\big(\theta[\cmn](w-\tfrac{1}{2\sqrt{2\pi {\bf i}}\, \hbar}\cdot (\epsilon+\hbar\cdot p_1)^2;\tau)\big).
$$
In particular, take $\mu=\nu=\frac{1}{2}$, then we have
$$\textstyle
U(p_1;\hbar)=\frac{G_2\cdot w\cdot \hbar}{\sqrt{2\pi {\bf i}}}+\frac{\hbar^2}{8(\epsilon+\hbar\cdot p_1)^2}
+\pd_{p_1}^2\log\big(\sigma(w-\tfrac{1}{2\sqrt{2\pi {\bf i}}\, \hbar}\cdot (\epsilon+\hbar\cdot p_1)^2;\tau)\big),
$$
where $\sigma(w;\tau)$ is the Weierstrass sigma-function.
\end{conjecture}
\begin{proposition}
Conjecture~\ref{conj:initial-u-kdv} is equivalent to the following equations: for $g\geq 2$,
$$
\omega_{g,0}=0,
$$
where $\omega_{g,0}$ is defined by equation~\eqref{def:wg0} via the topological recursion on the Weierstrass curve.
\end{proposition}
\begin{proof}
We introduce the following notation convention: for a function $f=f({\bf p})$ (resp. $f({\bf s})$ or $f({\bf t})$),
we denote by $f(p_1)$ (resp. $f(s_0^1)$ or $f(t_0^1)$) the restriction of $f$ to $p_k=0$ for all $k\ne 1$ (resp. $s_k^a=0$ or $t_k^a=0$ for all $(k,a)\ne (0,1)$).

By using the ancestor string equation (see~\cite{GZ25}):
$(\frac{\pd }{\pd s_0^0}-\sum_{k,a}s_{k+1}^a\frac{\pd }{\pd s_k^a} -\frac{\eta(s_0,s_0)}{2\hbar^2})\cA({\bf s};\hbar)=0$,
we get
$$\textstyle
\cA^{\rm NP}_{\mu,\nu,\tau}(\hbar\cdot s_0^1;\hbar;w)=\cA(\hbar\cdot s_0^1;\hbar)\cdot \theta[\cmn]\big(w-\frac{\hbar}{2\sqrt{2\pi {\bf i}}}\cdot (s_0^1)^2;\tau\big).
$$
Furthermore, by direct computations, we have
$\<\phi_0\>_{0,1}=\frac{\epsilon^2}{2}$, $W(t^1_0\phi_1,\phi_0)= \epsilon \, t_0^1$, and this gives
$$\textstyle
\cD^{\rm NP}_{\mu,\nu}(\hbar\cdot t_0^1;\hbar;w)=
\cD(\hbar\cdot t_0^1;\hbar)\cdot \theta[\cmn]\big(w-\frac{1}{2\sqrt{2\pi {\bf i}}\, \hbar}\cdot (\epsilon+\hbar\cdot t_0^1)^2;\tau\big).
$$
We obtain
$$\textstyle
U(p_1;\hbar)=
\pd_{t_0^1}^2\log(\cD(\hbar\cdot t_0^1;\hbar))\big|_{t_0^1=p_1}
+\pd_{p_1}^2\log\big(\theta[\cmn](w-\frac{1}{2\sqrt{2\pi {\bf i}}\, \hbar}\cdot (\epsilon+\hbar\cdot p_1)^2;\tau)\big).
$$

Let $\cD^{\ttau}({\bf t};\hbar):=\cD^{\ttau}({\bf t}+\ttau;\hbar)$ be the shifted total descendent potential, then we have
$$\textstyle
\cD(t_0^1;\hbar)=\cD^{\ttau}({\bf 0};\hbar)|_{\ttau^i=\delta_{i,1}t_0^1}
=\exp\big(\frac{1}{\hbar^2}F_0(\ttau)+F_1(\ttau)+\sum_{g\geq 2}\hbar^{2g-2}\omega_{g,0}^{\ttau}\big)|_{\ttau^i=\delta_{i,1}t_0^1},
$$
where for $g\geq 2$, $\omega_{g,0}^{\ttau}$ is defined by equation~\eqref{def:wg0} via the topological recursion on the shifted Weierstrass curve $\cC^{\ttau}$.
Since for the curve $\cC^{\ttau}$, the parameter $\ttau^1$ can be recovered from the parameter $\epsilon$ by $\epsilon\to \epsilon+\ttau^1$, we have
$$\textstyle
\cD(t_0^1;\hbar)
=\exp\big(\frac{1}{\hbar^2}F_0(\epsilon)+F_1(\epsilon)+\sum_{g\geq 2}\hbar^{2g-2}\omega_{g,0}\big)|_{\epsilon\to \epsilon+t_0^1},
$$
where $F_{g}(\epsilon)=F_{g}(\ttau)|_{\ttau^i=\delta_{i,1}\epsilon}$, $g=0,1$.
The Proposition follows from $\frac{\pd^2 F_{0}(\epsilon)}{\pd \epsilon\pd \epsilon}=(S_1)^{1}_{1}
=\frac{3  \epsilon^2\cdot G_2}{4\pi {\bf i}}$,
and
$\frac{\pd^2 F_{1}(\epsilon)}{\pd \epsilon\pd \epsilon}=\<\phi_1,\phi_1\>_{1,2}=\frac{1}{8\epsilon^2}$.
\end{proof}

By direct computations, we have checked
$\omega_{g,0}=0$ for $g=2,3,4$.

\begin{appendices}

\addtocontents{toc}{\protect\setcounter{tocdepth}{0}}

\section{Expansion of Bergman kernel at boundary point}
\label{sec:Bergman-kernel}
We consider the expansion of  $B(z_1,z_2)$ with the second variable $z_2$ near the boundary point $b$, and study its behavior in the first variable $z_1$.
\begin{definition}
For two meromorphic $2$-forms $w_1(z_1,z_2)$, $w_2(z_1,z_2)$ on $\Sigma\times\Sigma$,
we say $w_1(z_1,z_2)\sim_{x}w_2(z_1,z_2)$ if near the point $z_2=b$
$$
w_1(z_1,z_2)-w_2(z_1,z_2) \in d_{z_1}d_{\lambda_2}\mathbb C[x(z_1)][[\lambda_2^{-1}]].
$$
\end{definition}
\begin{lemma}\label{lem:equi-x}
Let $w(z_1,z_2)$ be a 2-form admitting expansion $\sum_{k\geq 1} df_k(z_1)d\lambda_2^{-k}$ near the boundary point $z_2=b$,
where $f_k(z)$ are meromorphic functions with only pole structure at $z=b$,
then $w(z_1,z_2)\sim_{x}0$ if and only if for any $m\geq 0$, the expansion of $(d\circ\frac{1}{dx(z_1)})^{m}w(z_1,z_2)$ at $z_2=b$ is meromorphic with respect to $z_1$ and has only possible pole at $z_1=b$.
\end{lemma}
\begin{proof}
Clearly, if $w(z_1,z_2)\sim_{x}0$, then for each $k\geq 1$, $f_k(z_1)$ has form of a polynomial of $x(z_1)$, and thus $(d\circ\frac{1}{dx(z_1)})^{m}w(z_1,z_2)$ is meromorphic with respect to $z_1$ and has only possible pole at $z_1=b$.
Conversely,
notice that $x(z_1)$ has pole at $z_1=b$, we know the order of pole of $(d\circ\frac{1}{dx(z_1)})^{m+1}df_k(z_1)$ at $z_1=b$ strictly decreases as $m$ increases.
Hence, if the integer $m$ is large enough, $(d\circ\frac{1}{dx(z_1)})^{m}df_k(z_1)$ has no pole on $\Sigma$ and has zero at $z_1=b$ of any order.
Therefore, there exist some integer $m$ (depends on $k$) such that $(d\circ\frac{1}{dx(z_1)})^{m}df_k(z_1)=0$, and by taking integration, $f_k(z)$ must be a polynomial of $x(z)$.
\end{proof}

\begin{lemma}\label{lem:B-x}
Let $B_{x}(z_1,z_2)$ be a 2-form defined by
$$
B_{x}(z_1,z_2):=-\sum_{k\geq 0}\sum_{\beta=1,\cdots,N} d(-D)^{-k-1}\zeta^{\bar\beta}(z_1)dD^{k}\zeta^{\bar\beta}(z_2),
$$
we have
$$
B(z_1,z_2)\sim_{x} B_{x}(z_1,z_2).
$$
\end{lemma}
\begin{proof}
It is easy to see both $B$ and $B_x$ can be expanded at $z_2=b$ and satisfy the condition assumed in Lemma~\ref{lem:equi-x}.
Thus by Lemma~\ref{lem:equi-x} we just need to prove that for any $m\geq 0$, the expansion of $(d\circ\frac{1}{dx(z_1)})^{m}(B(z_1,z_2)-B_{x}(z_1,z_2))$ at $z_2=b$ is meromorphic with respect to $z_1$ with only possible pole at $z_1=b$.

Notice that the only possible poles of $(d\circ\frac{1}{dx(z_1)})^{m}(B(z_1,z_2)-B_{x}(z_1,z_2))$ with respect to $z_1$ appear at $z_1=z^\gamma$, $\gamma=1,\cdots,N$, and $z_1=b$, we just need to check that the local behavior of $(d\circ\frac{1}{dx(z_1)})^{m}(B(z_1,z_2)-B_{x}(z_1,z_2))$ at $z_1=z^{\gamma}$, $\gamma=1,\cdots,N$, and prove that there is no pole.
Locally, by using airy coordinate $\eta_1^{\gamma}$, we have $d\circ\frac{1}{dx(z_1)}=d\circ\frac{1}{\eta^\gamma_1d\eta^\gamma_1}$. Therefore,
we just need to prove the coefficients of $(\eta^\gamma_1)^{\rm even}d\eta_1^{\gamma}$ in the local expansion of $B(z_1,z_2)$ and $B_{x}(z_1,z_2)$ at $z_1=z^\gamma$ are equal.
This is equivalent to the following equation:
\beq\label{eqn:B-Bx}
B(z_1, z_2)-B(\bar{z}_1,z_2)=B_x(z_1, z_2)-B_x(\bar{z}_1,z_2),
\eeq
where $\bar{z}_1$ is the involution of $z_1$ near $z^{\gamma}$.
By equation~\eqref{def:EO-R} and equation~\eqref{eqn:localxi}, we have
$$
\zeta^{\bar\beta}(z_1)=\sum_{l\geq 0}(-1)^{l} (R_{l})^{\bar\beta}_{\bar \gamma}\frac{(\eta_1^{\gamma})^{2l-1}}{(2l-1)!!}+f^{\bar\beta}\big((\eta_1^{\gamma})^{2}\big)
$$
where $f^{\beta}\big((\eta_1^{\gamma})^{2}\big)$ is a power series of $(\eta_1^{\gamma})^{2}$.
By direct computation,
$$
(-D)^{-k-1}\zeta^{\bar\beta}(z_1)=\sum_{l\geq 0}(-1)^{l} (R_{l})^{\bar\beta}_{\bar \gamma}\frac{(\eta_1^{\gamma})^{2l+2k+1}}{(2l+2k+1)!!}
+(-D)^{-k-1}f^{\bar\beta}\big((\eta_1^{\gamma})^{2}\big)
\mod\mathbb C[x(z_1)].
$$
This gives
$$
(-D)^{-k-1}\zeta^{\bar\beta}(z_1)-(-D)^{-k-1}\zeta^{\bar\beta}(\bar z_1)
=2\sum_{l\geq 0}(-1)^l (R_l)^{\bar\beta}_{\bar \gamma}\frac{(\eta_1^{\gamma})^{2l+2k+1}}{(2l+2k+1)!!}.
$$
Therefore, near the critical point $z_1=z^{\gamma}$, by using equation~\eqref{eqn:zeta-xi} we have
\beq\label{eqn:equiv-Bx}
B_{x}(z_1,z_2)-B_{x}(\bar z_1,z_2)=-2\sum_{k\ge 0}\frac{(\eta_1^\gamma)^{2k}}{(2k-1)!!}d\eta^\gamma_1d\xi^{\bar\gamma}_{k}(z_2),
\eeq
By comparing equation~\eqref{eqn:equiv-Bx} with~\eqref{eqn:equiv-B}, we obtain \eqref{eqn:B-Bx} and the Lemma is proved.
\end{proof}

\section{Some results for the deformed negative $r$-spin theory}
\subsection{Explicit formula for $\omega_{g,0}$ of the deformed $r$-Bessel curve}
\label{sec:omegag0}
\begin{proposition}
For the $\epsilon$-deformed $r$-Bessel curve,
$$
\cC^{r,\epsilon}=\big(\, \mathbb P^1,\quad x(z)=-z^r+\epsilon z, \quad y(z)=\tfrac{\sqrt{-r}}{z}\, \big),
$$
the function $\omega^{r,\epsilon}_{g,0}$, $g\geq 2$, associated with $\cC^{r,\epsilon}$ and defined by equation~\eqref{def:wg0}, has formula
$$
\omega^{r,\epsilon}_{g,0}=\frac{(-1)^{g-1}}{r^{g-1}}\frac{B_{2g}}{2g(2g-2)}\frac{1}{\epsilon^{2g-2}}.
$$
\end{proposition}
\begin{proof}
For $r=2$, the formula has been proved in~\cite{IKT19}.
For arbitrary $r\geq 2$, we consider a further deformation of the deformed $r$-Bessel curve  $\cC^{r,\epsilon,t}=(\mathbb P^1, x(z), y_t(z))$, where
$$\textstyle
x(z)=-z^r+\epsilon z,\qquad y_{t}(z)=\frac{\sqrt{-r}}{z}+\frac{t}{z^2}.
$$
Let $R^{r,\epsilon,t}$ (resp. $R^{r,\epsilon}$) and $T^{r,\epsilon,t}$ (resp. $T^{r,\epsilon}$) be the $R$-matrix and $T$-vector of $\cC^{r,\epsilon,t}$ (resp. $\cC^{r,\epsilon}$) defined by~\eqref{def:EO-R} and \eqref{def:EO-T}. 
Then it is easy to see
$R^{r,\epsilon,t}=R^{r,\epsilon}$ and $\lim_{t\to 0}T^{r,\epsilon,t}=T^{r,\epsilon}$.
This proves that
\beq\label{eqn:omegag0-lim}
\lim_{t=0}\omega^{r,\epsilon,t}_{g,0}=\omega^{r,\epsilon}_{g,0},
\eeq
where $\omega_{g,0}^{r,\epsilon,t}$ is the function associated with $\cC^{r,\epsilon,t}$ defined by equation~\eqref{def:wg0}.

Now we swap $x(z)$ and $y_t(z)$ of $\cC^{r,\epsilon,t}$ and make a coordinate changing $z\to t/z$ on the spectral curve $\mathbb P^1$, we get new spectral curve $\cC^{\dagger,r,\epsilon,t}=(\mathbb P^1,  x^{\dagger}_t(z), y^{\dagger}_t(z))$, where
$$\textstyle
x^{\dagger}_t(z)=\frac{z^2}{t}+\sqrt{-r}\frac{z}{t},\qquad
y^{\dagger}_t(z)=-\frac{t^r}{z^r}+\frac{t\cdot\epsilon}{z},
$$
and new function $\omega_{g,0}^{\dagger,r,\epsilon,t}$ of $\cC^{\dagger,r,\epsilon,t}$.
By $x-y$ symmetry~\cite{ABDKS22,EO07b},
\beq\label{eqn:omegag0-xy}
\omega_{g,0}^{\dagger,r,\epsilon,t}=\omega_{g,0}^{r,\epsilon,t}.
\eeq
		
We further consider a scaling on $x^{\dagger}_t(z)$ and $y^{\dagger}_t(z)$ to get $\tilde\cC^{\dagger,r,\epsilon,t}=(\mathbb P^1,  \tilde x^{\dagger}_t(z), \tilde y^{\dagger}_t(z))$ and $\tilde\omega_{g,0}^{\dagger,r,\epsilon,t}$, where
$$\textstyle
\tilde x^{\dagger}_t(z)=-t\cdot x^{\dagger}_t(z)=-z^2-\sqrt{-r}z,\qquad \tilde y^{\dagger}(z)=-y^{\dagger}(z)/t=\frac{t^{r-1}}{z^r}-\frac{\epsilon}{z},
$$
By equation~\eqref{eqn:scaling-xy}, we have
\beq\label{eqn:omegag0-scaling}
\tilde\omega_{g,0}^{\dagger,r,\epsilon,t}=\omega_{g,0}^{\dagger,r,\epsilon,t}.
\eeq
		
Now we can take limit $t\to 0$ on $\tilde\cC^{\dagger,r,\epsilon,t}$, then by the result for $r=2$ and equation~\eqref{eqn:scaling-xy}
\beq\label{eqn:omegag0-lim-tilde}
\lim_{t\to 0}\tilde\omega_{g,0}^{\dagger,r,\epsilon,t}
=\frac{\sqrt{-2}^{2g-2}}{\epsilon^{2g-2}}\cdot \frac{(-1)^{g-1}}{2^{g-1}}\frac{B_{2g}}{2g(2g-2)}\frac{1}{(-\sqrt{-r})^{2g-2}}.
\eeq
The proposition follows from equations~\eqref{eqn:omegag0-lim}, \eqref{eqn:omegag0-xy}, \eqref{eqn:omegag0-scaling} and \eqref{eqn:omegag0-lim-tilde}.
\end{proof}
	
\subsection{$S$-matrix of the deformed negative r-spin theory via twisted theory} \label{sec:des-r}	
\begin{proof}[Proof of Theorem \ref{thm:descendentrspin}]
We introduce the twisted theory of the negative $r$-spin theory
by considering the natural action of $\lambda\in\mathbb C^{\times}$ scaling the fibers of vector bundle $\cV^{r,-1}_{g,\vec{a}}$ by multiplication.
The descendent correlator for the twisted theory is defined in a similar way as the one for the untwisted theory defined in Definition~\ref{def:des-neg-r} by replacing $c_{\rm top}(\cV^{r,-1}_{g,\vec{a}})$ with the equivariant Euler class $e_{\mathbb C^{\times}}( \mathcal V^{r,-1}_{g,\vec{a}} )$ which is defined by
$$
e_{\mathbb C^{\times}}(\cV^{r,-1}_{g,\vec{a}})=\exp\bigg(\ln(\lambda)\ch_{0}+\sum_{m\geq 1}(-1)^{m-1}\frac{(m-1)!}{\lambda^m}\ch_{m}(\cV^{r,-1}_{g,\vec{a}})\bigg).
$$
Since
$\lim_{\lambda\to 0}e_{\mathbb C^{\times}}(\cV^{r,-1}_{g,\vec{a}})=c_{\rm top}(\cV^{r,-1}_{g,\vec{a}})$,
the negative $r$-spin theory can be recovered from the twisted theory by taking non-equivariant limit $\lambda\to0$.
Following the discussion in~\cite{CR10} (see details in the second part of the proof for the Proposition 4.1.5 in~\cite{CR10}), we extend the symmetric $2$-form $\eta(-,-)$ on $\Frob$ to $\bar\Frob=\Frob\oplus\mathbb Q\phi_0$ by setting
$\eta(\phi_0,\phi_a)=\delta_{a,0}\cdot\frac{1}{\lambda}$,
then it is clear that the extended symmetric $2$-form $\eta$ is non-degenerated and we have $\phi^0=\lambda\phi_0$.
The total descendent potential $\cD^{r,\epsilon,\rm tw}({\bf t};\hbar)$ with $\epsilon$-shifting of the twisted theory is defined in a similar way as the one $\cD^{r,\epsilon}({\bf t};\hbar)$ for the untwisted theory defined by~\eqref{def:partitionfunction-r}.
We note here that we still take
${\bf t}={\bf t}(\psi)=\sum_{k\geq 0, 1\leq a\leq r-1}t^a_k\phi_a\psi^k\in\Frob[[\psi]]$
 and thus
$\lim_{\lambda\to0}\cD^{r,\epsilon,\rm tw}({\bf t};\hbar)=\cD^{r,\epsilon}({\bf t};\hbar)$.
		
Similarly as the definition~\ref{def:class-r}, we define the twisted class $\Upsilon^{r,\rm tw}_{g,n}$ and $\Upsilon^{r,\epsilon,\rm tw}_{g,n}$  by replacing $c_{\rm top}(\cV^{r,-1}_{g,\vec{a}})$ with $e_{\mathbb C^{\times}}( \mathcal V^{r,-1}_{g,\vec{a}} )$ in equation~\eqref{def:-rspinclass} and \eqref{def:deformedclass}, respectively. Then the choice of the symmetric $2$-form ensures that these classes are CohFTs on the space $(\bar\Frob,\eta)$.
An alternative way to define the shifted twisted ancestor correlator is
$$
\<  \phi_{a_1}\bar\psi_1^{k_1},\cdots, \phi_{a_n}\bar\psi_n^{k_n}\>_{g,n}^{r,\epsilon,\tw}
:=\sum_{m\geq 0}\frac{\epsilon^m}{m!}\frac{(-1)^{D^{r,-1}_{g,n+m}(\vec{a}+\vec{0}_m)}}{r^{g-1}}
\int_{\Mbar_{g,\vec{a}+\vec{0}_m}^{r,-1}}  e_{\mathbb C^{\times}}( \mathcal V^{r,-1}_{g,\vec{a}+\vec{0}_m} )\cdot\prod_i  \bar\psi_i^{k_i},
$$
where $\bar\psi_i$ is the pull back of the $\psi_i$ on the $\Mbar_{g,n}$ via map $\M_{g,\vec{a}+\vec{0}_m}^{r,-1}\to \Mbar_{g,n+m}\to\Mbar_{g,n}$.
We denote by $\cA^{r,\epsilon,\rm tw}({\bf s};\hbar)$ the total ancestor potential of the shifted twisted ancestor correlators.

By comparing the difference between $\psi_i$ and $\bar\psi_i$ as the discussion in~\cite[Appendix 2]{CG07}, we have the following formula:
\beq\label{eqn:DA-tw}
\cD^{r,\epsilon,\rm tw}({\bf t};\hbar)=e^{F^{r,\epsilon,\tw}_{\rm un}({\bf t};\hbar)}\cdot \cA^{r,\epsilon,\rm tw}({\bf s(t)};\hbar),
\eeq
where
$F^{r,\epsilon,\tw}_{\rm un}({\bf t};\hbar) =\frac{1}{\hbar^2}\sum_{k,a}\<\phi_a\psi^{k}\>^{r,\epsilon,\tw}_{0,1}t_k^a+\frac{1}{2\hbar^2}\sum_{k,l,a,b}\<\phi_a\psi^{k},\phi_b\psi^l\>^{r,\epsilon,\tw}_{0,2}t_k^a t_l^b$
and the coordinate transformation ${\bf s}={\bf s(t)}$ is given by ${\bf s}(\givz)=[S^{r,\epsilon,\tw}(\givz){\bf t}(\givz)]_{+}$.
Here $S^{r,\epsilon,\tw}(\givz)$ is the $S$-matrix of the twisted negative $r$-spin theory defined by
$$
\eta(\phi_a,S^{r, \epsilon,\tw}(\givz)\phi_b):=\eta(\phi_a,\phi_b)+\Big\<\phi_a,\frac{\phi_b}{\givz-\psi}\Big\>_{0,2}^{r,\epsilon,\tw}.
$$

Let $J^{r,\epsilon,\rm tw}(\givz):=\frac{-r}{1-r\lambda}\phi_{r-1} +\epsilon\phi_0+\<\frac{\phi_a}{\givz-\psi}\>_{0,1}^{r,\epsilon,\rm tw}\phi^a$ be the $J$-function of the twisted negative $r$-spin theory,
then it is clear that $\pd_{\epsilon}J^{r,\epsilon,\rm tw}(\givz)=S^{r,\epsilon,\tw,*}(\givz)\phi_0$.
We note here that the genus zero correlators of the twisted theory satisfies the following topological recursion relation (which is a direct corollary of~\eqref{eqn:DA-tw}):
\beq\label{eqn:TRR-tw}
\<\phi_a\psi^k,\phi_b\psi^m,\phi_c\psi^n\>_{0,3}^{r,\epsilon,\rm tw}
=\<\phi_a\psi^{k-1},\phi_{d}\>_{0,2}^{r,\epsilon,\rm tw}\<\phi^d,\phi_b\psi^m,\phi_c\psi^n\>_{0,3}^{r,\epsilon,\rm tw}.
\eeq
Notice that $\lim_{\lambda\to0}\phi^0=0$, we have
$\lim_{\lambda\to0}\eta(\phi^0,S^{r,\epsilon,\rm tw}(\givz)\phi_a):=\delta_{a,0}$,
we see that $\lim_{\lambda\to 0}S_k^{r,\epsilon,\tw}$, $k\geq0$, maps $\Frob\to \Frob$,
and thus $\lim_{\lambda\to 0}S^{r,\epsilon,\tw}(\givz)|_{\Frob}=S^{r,\epsilon}(\givz)$.
Clearly, $\lim_{\lambda\to 0}(S^{r,\epsilon,\tw})^{-1}(\givz)|_{\Frob}=(S^{r,\epsilon})^{-1}(\givz)$. This proves the first part of the Theorem~\ref{thm:descendentrspin}.
		
Now we consider the quantum product $*_{\epsilon,\rm tw}$ defined by
$\eta(\phi_a*_{\epsilon,\tw}\phi_b,\phi_c)=\<\phi_a,\phi_b,\phi_c\>^{r,\epsilon,\tw}_{0,3}$.
By similar reason as above, we have  $\lim_{\lambda\to0}\phi_a*_{\epsilon,\rm tw}$, $a=0,\cdots,r-1$, preserves $\Frob$ and
$\lim_{\lambda\to0}\phi_b*_{\epsilon,\rm tw}|_{\Frob}=\phi_b*_{\epsilon}$, $b=1,\cdots,r-1$.
Moreover, we define $\phi_0*_{\epsilon}:=\lim_{\lambda\to0}\phi_0*_{\epsilon,\rm tw}|_{\Frob}$, then by equation~\eqref{eqn:-rspin-corr(0,n)},  for $a,b=1,\cdots,r-1$,
$$
\phi_0*_{\epsilon}\phi_a
=\left\{\begin{array}{cc}
	\phi_{a+1} & 1\leq a\leq r-2\\
	\frac{\epsilon}{r}\phi_1 & a=r-1
\end{array}\right..
$$
Consider the quantum differential equation for the twisted $S$-matrix (which can be seen by taking $m=n=0$ and $c=0$ in equation~\eqref{eqn:TRR-tw}):
$$
\givz\pd_{\epsilon}S^{r,\epsilon,\tw}(\givz)=\phi_0*_{\epsilon,\tw}S^{r,\epsilon,\tw}(\givz).
$$
By taking limit $\lambda\to 0$, the QDE gives
$$
(\pd_{\epsilon}S^{r,\epsilon}_k\phi_a,\phi^b)=(S^{r,\epsilon}_{k-1}\phi_a,\phi_0*\phi^b),
$$
and we have
$$
\pd_{\epsilon}(S^{r,\epsilon}_{k+1})_{a}^{b}=(S^{r,\epsilon}_k)^{b-1}_{a}+\delta_{b,1}\frac{\epsilon}{r}(S^{r,\epsilon}_{k})^{r-1}_{a}.
$$
This solves
$$
(S^{r,\epsilon}_k)_{a}^b=\left\{\begin{array}{cc}
(-1)^k\frac{\Gamma(\frac{a}{r})}{\Gamma(\frac{a}{r}+k-m)}\frac{(-\epsilon)^{m}}{m!},
& m=\frac{rk+a-b}{r-1}\in \mathbb Z_{+}\\
0 , & {\rm otherwise}
\end{array}\right. .
$$
Furthermore, by equation~\eqref{eqn:TRR-tw},
$$
\<\phi_a\psi^{k},\phi_0,\phi_0\>_{0,3}^{r,\epsilon}=\<\phi_a\psi^{k-1},\phi^b\>_{0,2}^{r,\epsilon}\<\phi_b,\phi_0,\phi_0\>_{0,3}^{r,\epsilon}=(S^{r,\epsilon}_{k})^{r-1}_{a},
$$
by integration,
$$
\<\phi_a\psi^k\>_{0,1}^{r,\epsilon}
=\left\{\begin{array}{cc}
	(-1)^{k+m}\frac{\Gamma(\frac{a}{r})}{\Gamma(\frac{a}{r}+k-m)}\frac{\epsilon^{m+2}}{(m+2)!} & m=\frac{rk+a}{r-1}-1\in\mathbb Z_{+}\\
	0 & {\text{otherwise}}
\end{array}\right.
.
$$
The second part of the Theorem follows immediately from these explicit computations.
\end{proof}
\end{appendices}

\vspace{2em}

\section*{List of Symbols}
\addcontentsline{toc}{section}{List of Symbols}

\vspace{1em}
\begin{longtable}{@{\hspace{0.5em}} l @{\hspace{2em}} p{0.75\textwidth}}

$(\Frob,\eta)$ & state space $\Frob$ with a non-degenerate two form $\eta$\\

$\Mbar_{g,n}$ & moduli space of stable curves of genus $g$ with $n$ marked points\\

$\Omega$, $\Omega^{\ttau}$ & CohFT and shifted CohFT along $\ttau$ (Definition~\ref{def:cohft})\\

${\bf v}(\givz)$, ${\bf v}^{\ttau}(\givz)$  & vacuum vector of CohFT (Definition~\ref{def:vacuum}) \\

$e_\beta$, $\bar e_\beta$, $\phi_i$ & canonical basis, normalized canonical basis and flat basis \\

$s_k^{\bullet},t_k^{\bullet}$ & $\bullet=\beta,\bar\beta, i$, ancestor (resp. descendent) coordinate with respect to basis $e_\beta\bar\psi_k$, $\bar e_\beta \bar\psi^k$, $\phi_i\bar\psi^k$ (resp. $e_\beta\psi_k$, $\bar e_\beta \psi^k$, $\phi_i\psi^k$)\\

$R(\givz)$, $R^{*}(\givz)$ & $R$-matrix and its adjoint with respect to $\eta$ \\

$V(\givz, \givw)$ & $V$-matrix, $V(\givz, \givw)=\frac{{\rm Id}-R^*(-\givz)R(-\givw)}{\givz+\givw}$ \\

$T(\givz)$ & $T$-vector, $T(\givz)={\bf \bar 1}\givz- \givz R(\givz)^{-1}\vac(\givz)$, $\bar{\bf 1}=\sum_\alpha \bar e_\alpha$ \\

$E$, $E_0$ & the Euler vector field and the Euler vector field at $\ttau=0$\\

$S(\givz)$, $S^{*}(\givz)$ & $S$-matrix and its adjoint with respect to $\eta$ \\

$W(\givz, \givw)$ & $W$-matrix, $W(\givz, \givw)=\frac{S^{*}(\givz)S(\givw)-\id}{\givz^{-1}+\givw^{-1}}$ \\

$\nu(\givz)$ & $\nu$-vector (Definition~\ref{def:dvac})\\

$J(\givz)$ & $J$-function, $J(-\givz)=-\givz S^{*}(-\givz)\dvac(\givz)$ \\

$\mathcal{C}=(\Sigma,x,y)$ & spectral curve data  (\S \ref{sec:intro1})\\

$\theta[\cmn](w|\tau)$ & theta function with characteristic $[\cmn]$\\

$B(z_1,z_2)$ & Bergman kernel of $\Sigma$\\

$\omega_{g,n}$ & multi-differentials on spectral curve (Definition~\ref{def:omegagn})\\

$\Phi(\gamma, -\givz)$ & class associated to admissible path $\gamma$ (Lemma~\ref{lem:path-class})\\

$d\xi_k^{\bullet}(z)$ & $\bullet=\beta,\bar\beta, i$, global meromorphic 1-form on $\Sigma$ (equation~\eqref{def:xi})\\

$d\zeta_k^{\bullet}(z)$ & $\bullet=\beta,\bar\beta, i$, global meromorphic 1-form on $\Sigma$ (equation~\eqref{eqn:zeta-xi})\\

$d\chi_k^{\bullet}(z)$ & $\bullet=\beta,\bar\beta, i$, 1-form locally near boundaries (equation~\eqref{eqn:zeta-chi})\\

$d\lambda_i^{-k}(z)$ & 1-form locally near boundary $b_i$  \\

$\<-\>_{g,n}^{\bullet}$ & $\bullet=\Omega$, $\cD$, $\rm TR$, $\Lambda$, the geometric ancestor, descendent, TR ancestor, TR descendent correlators \\

$\cA({\bf s}; \hbar)$ & total ancestor potential \\

$\cD({\bf t}; \hbar)$ & total descendent potential  \\

$Z({\bf p}; \hbar)$ & TR total descendent potential   \\

$\square^{\rm NP}$ & $\square=\cA,\cD, Z$, non-perturbative generating series  (\S \ref{sec:NP-TR})\\

\end{longtable}

\end{document}